\documentclass[11pt]{article}
\usepackage{geometry}              
\usepackage{graphicx}
\usepackage{amsmath,amssymb,bbm}
\usepackage{xcolor,comment}
\usepackage{todonotes}
\usepackage{hyperref,cite}

\usepackage{amsmath,amsfonts,stmaryrd,tabularx,amssymb,rotating,enumerate,dsfont,xspace}
\usepackage{mathtools,alltt}
\usepackage{bbm}
\usepackage{eucal}
\usepackage{latexsym}

\usepackage[curve,cmtip]{xypic}
\usepackage{verbatim}
\usepackage{changebar}
\usepackage{wrapfig}

\usepackage{tikz-cd}
  \usetikzlibrary{arrows,arrows.meta}
  \tikzcdset{arrow style=tikz, diagrams={>=stealth'}}

\usepackage[all]{xy}
\SelectTips{cm}{}

\newdir{ >}{{}*!/-10pt/\dir{>}}
\newdir{ (}{{}*!/-10pt/\dir^{(}}

\numberwithin{equation}{section}

\newtheorem{theorem}{Theorem}[section]
\newtheorem{corollary}[theorem]{Corollary}
\newtheorem{fact}[theorem]{Fact}
\newtheorem{proposition}[theorem]{Proposition}
\newtheorem{lemma}[theorem]{Lemma}
\newtheorem{definition}[theorem]{Definition}
\newtheorem{example}[theorem]{Example}

\newtheorem{remark}{Remark}
\newenvironment{proof}{\begin{trivlist}\item[]{\bf
Proof.}}{\hfill {\sc qed}\end{trivlist}}

\newcommand{\bi}{\begin{itemize}}
\newcommand{\ei}{\end{itemize}}
\newcommand{\benum}{\begin{enumerate}}
\newcommand{\eenum}{\end{enumerate}}
\newcommand{\ba}{\begin{array}}
\newcommand{\ea}{\end{array}}

\newcommand{\takeout}[1]{\empty}
\newcommand{\rem}[1]{\empty}

\newcommand{\blue}[1]{\textcolor{blue}{#1}}

\usepackage{tcolorbox,color}
\definecolor{cornellred}{RGB}{196,18,48}
\definecolor{dartmouthgreen}{RGB}{0,112,60}
\definecolor{mcgillred}{RGB}{237,27,47}

\newcommand{\To}{\Rightarrow}		

\newcommand{\sgoes}[1]{\stackrel{#1}{\longrightarrow}}
\newcommand{\sfrom}[1]{\stackrel{#1}{\longleftarrow}}
\newcommand{\natIso}{\stackrel{\sim}{\To}}		
\newcommand{\inj}{\rightarrowtail}
\newcommand{\surj}{\twoheadrightarrow}
\def\subto{\hookrightarrow}  

\newcommand{\lsub}[1]{{#1}_0}
\newcommand{\reach}[1]{\mathit{reach}(#1)}

\def\id{{\mathit{id}}}
\def\Id{{\mathit{Id}}}
\def\inc{{\sf{in}}} 
\def\arKl[#1,#2]{\ar[#1,#2]|-{\circ}}

\newcommand{\tup}[1]{\langle #1 \rangle}
\def\lsem{\llbracket}
\def\rsem{\rrbracket}
\newcommand{\sem}[1]{\llbracket #1 \rrbracket}
\newcommand{\langsem}[1]{\langle\!| #1 |\!\rangle}

\newcommand{\boxmod}[1]{[#1]}
\newcommand{\LLang}[1]{\mathbb{L}(#1)} 
\newcommand{\sm}[1]{s^{#1}} 
\renewcommand{\th}[1]{th^{#1}} 
\newcommand{\Trc}{\Alph^*\Omega} 

\newcommand\subs{\mathrel{\subseteq}}
\newcommand\len[1]{|#1|}
\newcommand\set[2]{\{#1 \mid #2 \}}
\newcommand\eps{\varepsilon}
\newcommand\eqdef\triangleq
\newcommand{\Sec}{Section}
\newcommand\Iff{\Leftrightarrow}

\def\o{\circ}
\def\x{\times}

\def\wt#1{\widetilde{#1}}

\def\ol#1{\overline{#1}}

\newcommand{\str}[1]{\mathcal{#1}}

\renewcommand{\phi}{\varphi}
\renewcommand{\rho}{\varrho}

\def\cat#1{\ensuremath{\mathsf{#1}}\xspace}
\newcommand\functor[1]{\ensuremath{\mathop{\mathit{#1}}}}

\newcommand{\catC}{\mathcal{C}}
\newcommand{\catD}{\mathcal{D}}
\newcommand{\catA}{\mathcal{A}}
\newcommand{\Set}{\cat{Set}}

\def\op{\mathsf{op}} 
\def\SetOp{{\Set^\op}}

\newcommand\Setop{\ensuremath{\SetOp}\xspace}
\newcommand\Qop{\ensuremath{\cPo^\op}\xspace}

\newcommand{\Coalg}{\mathsf{Coalg}}
\newcommand{\Alg}{\mathsf{Alg}}

\newcommand{\EM}[1]{\cat{EM}(#1)} 
\newcommand{\Aut}{\mathsf{Aut}}
\newcommand{\CABA}{\ensuremath{\mathsf{CABA}}\xspace}

\def\FVec{\cat{F}\Veck} 
\def\k{\mathbbm{k}}
\def\Veck{\textsf{Vec}_{\k}}
\def\SMod{{\S\textsf{Mod}}} 

\def\Pow{\mathcal{P}} 
\def\Powf{\mathcal{P}_\omega} 
\newcommand{\cPo}{\mathcal{Q}} 

\DeclareMathOperator{\Hom}{Hom} 
\def\D{\mathcal{D}_\omega} 
\def\N{\mathcal{N}} 

\renewcommand{\L}{P}  
\newcommand{\Pred}{P}  
\newcommand{\Spc}{S}  
\newcommand{\coun}{\varepsilon}  
\newcommand{\dua}{\Delta}  %
\newcommand{\FreeD}{\Phi_\catD}  
\newcommand{\ForgD}{U_\catD}  
\newcommand{\xitrc}{\xi^{\mathsf{trc}}}
\newcommand{\rhotrc}{\rho^{\mathsf{trc}}}

\def\st{\mathsf{st}}  

\renewcommand{\det}[1]{{#1}^\sharp} 
\newcommand{\lif}[2]{{#1}_{#2}}  
\newcommand{\Flam}{\lif{F}{\lambda}} 

\newcommand{\rev}{\mathit{rev}} 

\newcommand{\transp}[1]{#1^\flat} 

\renewcommand{\Alph}{\Sigma} 
\newcommand{\Lang}{\mathcal{L}}

\newcommand{\Two}{\mathbbm{2}} 
\DeclareMathOperator{\Uf}{Uf} 
\newcommand{\BA}{\cat{BA}}
\newcommand\CBA{\cat{CBA}}
\newcommand{\FBA}{\Phi_{\mathsf{BA}}} 

\newcommand\At{\functor{At}}

\newcommand\EMN{\ensuremath{\EM N}\xspace}

\newcommand\setcompl[1]{{\sim}#1}
\newcommand\up{\ensuremath{\mathop{\uparrow}}}
\renewcommand\up[1]{\ensuremath{#1{\uparrow}}}
\renewcommand\angle[1]{\langle{#1}\rangle}

\newcommand{\A}{\mathcal{A}}
\renewcommand{\L}{\mathcal{L}}

\newcommand{\fext}[1]{{#1}^\sharp}  
\newcommand{\ext}[1]{\fext{#1}}
\newcommand{\detA}{\mathop{\mathrm{det}}{\A}}

\newcommand{\TTwo}[1]{2^{2^{#1}}}
\newcommand\lam[2]{\lambda#1\kern1pt.\kern1pt#2}
\newcommand\trtr[1]{{#1}^{\flat\flat}}

\newcommand{\EMNtoCABA}{D}

\newcommand{\V}{\mathcal{V}_\bbS} 

\newcommand{\KHaus}{\mathsf{KHaus}} 
\newcommand{\CCStar}{{\mathsf{CUC^*Alg}}} 
\newcommand{\Spec}{\mathit{Spec}} 
\newcommand{\Cont}{\mathit{C}} 
\newcommand{\SubD}{\mathcal{D}_{\leq 1}} 
\newcommand{\pleq}{\preceq} 
\newcommand{\icong}[1]{\equiv_{#1}} 
\newcommand{\norm}[1]{\|#1\|}

\def\S{\mathbb{S}} 
\newcommand{\Obs}{\Omega} 
\newcommand{\obs}{\omega} 
\newcommand{\I}{I} 
\newcommand{\B}{B} 

\newcommand\emptyword\varepsilon
\def\Nat{\mathbb{N}}

\def\Real{\mathbb{R}}
\newcommand{\bbS}{\mathbb{S}}
\newcommand{\bbN}{\mathbb{N}}
\newcommand{\bbP}{\mathbb{P}}
\newcommand{\bbR}{\mathbb{R}}

\newcommand{\cB}{\mathcal{B}}

\newcommand{\cI}{\mathcal{I}}

\newcommand{\cS}{\mathcal{S}}

\newcommand\restr[2]{{
  \left.\kern-\nulldelimiterspace 
  #1 
  \vphantom{\big|} 
  \right|_{#2} 
  }}

\newif\ifdraft\drafttrue 

\title{Minimisation in Logical Form}
\author{
  Nick Bezhanishvili \and  Marcello Bonsangue \and
  Helle Hvid Hansen \and  Dexter Kozen \and
  Clemens Kupke \and   Prakash Panangaden \and Alexandra Silva
}
\date{\today} 
\begin{document}
\maketitle


\section{Introduction}

%

The role of category theory, algebra, and logic in deepening our understanding of semantics and algorithms in Computer Science has long been one of Samson's flagships. His seminal paper {\it Domain Theory in Logical Form}~\cite{Abr91}  studies the connection between program logic and domain theory via Stone duality. This is an example of a fundamental duality in Computer Science between semantics, operational or denotational, and syntax, provided as a logic or specification language.

Building on Stone's celebrated representation theorems for Boolean algebras \cite{Sto36} and distributive lattices \cite{Sto37}, categorical dualities linking algebra and topology~\cite{Joh82} have been widely used in logic and theoretical computer science 
\cite{CZ97,BdRV01,GHKLMS03}. With algebras corresponding to the syntactic, deductive side of logical systems, and topological spaces to their semantics, Stone-type dualities  provide a powerful mathematical framework for studying various properties of logical systems. More recently,  it has also been fruitfully explored in more algorithmic applications, notably in understanding minimisation of various types of automata~\cite{AdamekBHKMS12,BKP:WoLLIC,BBRS12,BBHPRS:ACM-TOCL,KlinR16,Rot16,ganty}. Among these, \cite{BKP:WoLLIC} and \cite{BBHPRS:ACM-TOCL} had striking similarities in the approach yet it was not clear whether the differences could be harmonised in a uniform way. The main aim of this paper is to find a unifying perspective on the minimisation constructions in \cite{BKP:WoLLIC} and \cite{BBHPRS:ACM-TOCL}. Duality will play a central role in achieving our aim of unification of approaches, which puts us on the path forged by Samson. 

In \cite{BKP:WoLLIC}, the authors adopt the coalgebraic perspective on automata and use a dual equivalence between the category of coalgebras and a category of algebras to explain minimisation. The key observation is that the algebras considered provide semantics of a modal logic. This algebraic semantics is dual to coalgebraic semantics in which 
logical equivalence coincides with trace equivalence. 
From this coalgebra-algebra duality it follows that maximal quotients of coalgebras correspond to minimal subobjects of algebras. In order to explain the minimisation algorithm, the authors exploit duality to prove that the maximal quotient of a coalgebra can be constructed by computing
the subalgebra of definable predicates in the dual modal algebra.
The examples in \cite{BKP:WoLLIC} include partially observable DFAs, linear weighted automata viewed as
coalgebras over finite-dimensional vector spaces, and belief automata which are coalgebras on compact Hausdorff spaces.

In \cite{BBHPRS:ACM-TOCL},
Brzozowski's double-reversal minimisation algorithm for deterministic finite automata 
(with both initial and final states) was described categorically and its correctness  explained 
via the duality between reachability and observability, whose origins trace back to seminal work of Kalman in control theory \cite{Kalman59}. Kalman's work was extended to automata theory in a collection of papers by Arbib and Manes~\cite{AZ69,Arbib74,AM75,AM75:Fuzzy,AM75:slatt,AM80:Hankel,AM80}. The work of \cite{BBHPRS:ACM-TOCL} is closely related to these and includes generalisations of Brzozowski's algorithm to Moore 
and weighted automata over commutative semirings. 

The contributions of the present paper are as follows.
\begin{enumerate}
\item A categorical framework within which minimisation algorithms can be understood and different  approaches unified (Section \ref{sec:dual-aut}). We start with a comparison between the approaches in \cite{BKP:WoLLIC} and \cite{BBHPRS:ACM-TOCL} (Section~\ref{sec:previous}) and then proceed to 
a general setup for different automata types based on algebra and coalgebra (Section \ref{sec:aut-alg-coalg}).
Section~\ref{sec:dual-adj-lift} 
includes the categorical picture that unifies the work in \cite{BKP:WoLLIC} and \cite{BBHPRS:ACM-TOCL}: in a nutshell, it is a stack of three interconnected adjunctions. It starts with a base dual adjunction that is subsequently lifted to a dual adjunction between coalgebras and algebras, and finally to a dual adjunction between automata.
Section~\ref{sec:trace-logic} extends this categorical picture place to include trace logic.
Section \ref{sec:reach-obs},
presents an abstract understanding of reachability and observability,
and finally everything is summarised and abstract minimisation algorithms are stated in Section~\ref{sec:abstract-brz}.
\item A thorough illustration of the general framework instantiated to concrete examples. In Section \ref{sec:old-ex}), we revisit a range of examples stemming from previous approaches: deterministic Kripke frames, weighted automata, and topological automata (belief automata). In Section \ref{sec:alternating}, we include an extensive new example on alternating automata, which uses the duality of complete atomic Boolean algebras and sets. 
For weighted automata, we use our framework to extend a well-known result for weighted automata over a field~\cite{Sch61} to weighted automata over a principal ideal domain: the minimal weighted automaton over a principal ideal domain always exists, and, as expected, it has a state space smaller or equal than that of the original automaton.
\end{enumerate}
We conclude the paper with a review of related work (Section \ref{sec:related-work}).


\section{Preliminaries}
\label{sec:prelim}

In this section, we fix notation and recall basic definitions of coalgebras and algebras.
For a more detailed introduction to coalgebra, we refer to \cite{Rutten00}.
For general categorical notions, see e.g.~\cite{AdaHerStr90:ACC}.
We assume familiarity with classic automata such as (non)deterministic finite automata,
and Moore automata.

Categories are denoted by $\catC, \catD, \ldots$,
objects of categories by $X,Y,Z,\ldots$,
and arrows/morphisms of categories by $f,g,h,\dots$. 
We denote by $\Set$ the category of sets and functions.
Let $X_1, X_2$ be in $\catC$.
The product of $X_1$ and $X_2$ (if it exists) is denoted by $X_1 \times X_2$
with projection maps $\pi_i\colon X_1 \times X_2 \to X_i$, $i=1,2$.
Similarly, their coproduct (if it exists) is written $X+Y$
with coprojection maps $\inc_i \colon X_i \to X_1 + X+2$.
In $\Set$, $X \times Y$ and $X+Y$ are the usual constructions of cartesian product and disjoint union.
Let $X$ be an object in $\catC$ and $A$ be a set.
Assuming $\catC$ has products,
then $X^A := \prod_A X$ denotes the $A$-fold product of $X$ with itself.
Similarly, if $\catC$ has coproducts, then $A \cdot X := \coprod_A X$ denotes the $A$-fold coproduct of $X$ with itself.

The covariant powerset functor $\Pow\colon \Set\to\Set$ sends a set $X$ to its powerset $\Pow(X)$ and a function $f\colon X \to Y$ to the direct-image map $\Pow(f)\colon \Pow(X) \to \Pow(Y)$.
The contravariant powerset functor $\cPo\colon \Set\to\SetOp$ also sends a set $X$ to its powerset, now denoted $\cPo(X)$, and a function $f \colon X \to Y$ to its inverse-image map 
$\cPo(f)\colon \Pow(Y) \to \Pow(X)$.

\takeout{For a set $A$ of letters,
$A^*$ denotes the set of all words (i.e., finite sequences) over $A$, and $\epsilon$ the empty word; and $w_1\cdot w_2$ (and $w_1w_2$) the concatenation of words $w_1,w_2 \in A^*$.
}

\subsection{Coalgebras, Algebras and Monads}
\label{sec:coalg-alg}

Given an endofunctor $F \colon \catC \to\catC$,
an \emph{$F$-coalgebra} is a pair $(X, \gamma\colon X\to FX)$,
where $X$ is a $\catC$-object and $\gamma\colon X\to FX$ is a $\catC$-arrow.
The functor $F$ specifies the type of the coalgebra (which may be thought of as the type of observations and transitions), and the structure map $\gamma$ specifies the dynamics.
An \emph{$F$-coalgebra morphism} from an $F$-coalgebra $(X,\gamma)$ to an
$F$-coalgebra $(Y,\delta)$ is a $\catC$-arrow $h\colon X \to Y$ that preserves
the coalgebra structure, {\em i.e.}, $\delta\circ h = Fh \circ \gamma$.
$F$-coalgebras and $F$-coalgebra morphisms form a category denoted by $\Coalg_\catC(F)$. A \emph{final $F$-coalgebra} is a final object in $\Coalg_\catC(F)$, i.e., an $F$-coalgebra $(\Omega,\omega)$ is final if for all $T$-coalgebras $(X,\gamma)$ there is a unique $F$-coalgebra morphism $h \colon (X,c) \to (\Omega,\omega)$.


An $F$-algebra is the dual concept of an $F$-coalgebra.
An \emph{$F$-algebra} is a pair $(X, \alpha)$,
where $X$ is a $\catC$-object and $\alpha\colon FX \to X$ is a $\catC$-arrow.
Now, the functor $F$ can be seen to specify the type of operations of the algebra. An \emph{$F$-algebra morphism} from an $F$-algebra $(X,\alpha)$ to an
$F$-algebra $(Y,\beta)$ is a $\catC$-arrow $h\colon X \to Y$ that preserves the
algebra structure, {\em i.e.}, $h\circ \alpha = \beta \circ Fh$.
$F$-algebras and $F$-algebra morphisms form a category denoted by $\Alg_\catC(F)$. An \emph{initial $F$-algebra} is an initial object $(A,\alpha)$ in $\Alg_\catC(F)$, i.e., for all $F$-algebras $(X,\beta)$ there is a unique $F$-algebra morphism $h\colon (A,\alpha) \to (X,\beta)$.

A \emph{monad} (on $\catC$) is a triple $(T, \eta, \mu)$ consisting of a functor
$T\colon \catC\to\catC$ and two natural transformations
$\eta\colon \Id \to T$ (the unit) and $\mu\colon TT \to T$ (the multiplication) satisfying $\mu \o \eta T = \id_T = \mu \o T\eta$ and $\mu \o T\mu = \mu \o \mu T$.
For brevity, we will sometimes refer to a monad simply by its functor part, leaving the unit and multiplication implicit.
An \emph{Eilenberg-Moore $T$-algebra} \takeout{for a monad $(T,\eta,\mu)$}
is a $T$-algebra $(A, \alpha)$ such that $\alpha \o \eta_A = \id_A$ and
$\alpha \o \mu_A = \alpha \o T\alpha$.
Eilenberg-Moore $T$-algebras and $T$-algebra morphisms form a category denoted by $\EM{T}$.
In particular, for every $X$ in $\catC$, $(TX, \mu_X)$ is the free
Eilenberg-Moore $T$-algebra on $X$,
i.\,e., for every 
$(A,\alpha)$ in $\EM{T}$ and
every $\catC$-arrow $f\colon X \to A$ there is a unique $T$-algebra morphism
(called the \emph{free extension of $f$}) $\ext f\colon (TX,\mu_X) \to (A,\alpha)$ such that $\ext f \o \eta_X = f$.
\takeout{
\begin{equation}
  \label{not:ext}
  \vcenter{
    \xymatrix{
      TTX
      \ar[r]^-{\mu_X}
      \ar[d]_{T\ext f}
      &
      TX
      \ar[d]^{\ext f}
      &
      X
      \ar[l]_-{\eta_X}
      \ar[ld]^-f
      \\
      TA
      \ar[r]_-\alpha
      &
      A
    }
  }
\end{equation}
}
Notice also that we have $\ext f = \alpha \o Tf$.


\subsection{Determinisation}
\label{sec:det}

Let $(T,\eta,\mu)$ be a monad on $\Set$ and 
$F\colon \Set \to\Set$ a functor given by
$FX = B \times X^A$ where $A$ is a set and $B$ is the carrier of an Eilenberg-Moore $T$-algebra $(B,\beta)$.
Then $FT$-coalgebras can be seen as automata with input alphabet $A$, output in $B$ and branching structure given by $T$. For example, nondeterministic automata are $F\Pow$-coalgebras where $FX = 2 \times X^A$
and $\beta = \lor \colon \Pow{2} \to 2$ is the join (or max).
Such $FT$-coalgebras can be ``determinised'' using a generalisation of the classic powerset construction \cite{SilvaBBR10}, and the result can be seen as an $F$-coalgebra in the category $\EM{T}$.
We follow \cite{Bartels:PhD,Jacobs:bialg-dfa-regex} in explaining this general construction. As shown in \cite{Jacobs:bialg-dfa-regex}, there is a so-called distributive law $\lambda\colon TF \To FT$ of the monad $(T,\eta,\mu)$ over the functor $F$ given by
\begin{equation}
 \xymatrix@C=5em{
  \lambda_X \colon T (B \times X^A) \ar[r]^-{\tup{T{\pi_1},T{\pi_2}}}
     & TB\times T(X^A) \ar[r]^-{\beta\times \st}
     & B\times (TX)^A
}
\label{eq:lambda}
\end{equation}
where $\st\colon T \o (-)^A \To (-)^A \o T$ is the strength natural transformation that exists for all monads on $\Set$.
Such a distributive law $\lambda$ corresponds to a lifting of $F\colon \Set\to\Set$ to a functor $\Flam\colon \EM{T} \to \EM{T}$ \cite{Johnstone:Adj-lif},
and it induces a functor
$\det{(-)} \colon \Coalg_\Set(FT) \to \Coalg_{\EM{T}}(\Flam)$
which sends an $FT$-coalgebra $\gamma = \tup{o,t}\colon X \to B \times (TX)^A$
to its determinisation 
$\det{\gamma} = F\mu_X \circ\lambda_{TX} \circ T\gamma$,
that is,
\begin{equation}
\det{\gamma} =
\xymatrix{
TX \ar[r]^-{T\gamma}
& T(B \times (TX)^A) \ar[r]^-{\lambda_{TX}}
& B \times (TTX)^A \ar[r]^-{B \times (\mu_X)^A}
& B \times (TX)^A
}
\label{eq:det}
\end{equation}
Another perspective is that $\lambda$ induces an Eilenberg-Moore $T$-algebra structure $\alpha$ on $FTX$, and $\det{\gamma}\colon (TX,\mu_S) \to (FTX,\alpha)$ is the free extension of $\gamma$ induced by $\alpha$. This also justifies our use of the notation $\det{(-)}$.

The determinisation $\det{\gamma}$ can be seen as a Moore automaton in $\EM{T}$. 
We will use the determinisation construction in order to place alternating automata and weighted automata in our general minimisation framework.

\section{Minimisation via Dual Adjunctions}
\label{sec:dual-aut}


\subsection{Unifying Previous Approaches}
\label{sec:previous}

One aim of this paper is to find a unifying perspective on the automata minimisation constructions in
\cite{BKP:WoLLIC} and \cite{BBHPRS:ACM-TOCL}.
We therefore start by summarising the two papers, and discuss the differences and similarities.

First we establish some terminology regarding key notions.
A classic DFA is reachable if all states are reachable by reading some word from the initial state,
it is observable if no two states accept the same language, and
it is minimal if it is both reachable and observable.
These notions can be generalised to other types of automata using that
automata are in a sense both algebras and coalgebras as we will explain in
Section~\ref{sec:aut-alg-coalg}.
We will call an algebra \emph{reachable} if it has no proper subalgebras,
and a coalgebra is \emph{observable} if it has no proper quotients.
A (generalised) automaton is then \emph{minimal} if its algebra part is reachable and
its coalgebra part is observable.
Note that in the literature, observable coalgebras are usually called minimal coalgebras.

In \cite{BKP:WoLLIC}, (generalised) Moore automata (without initial state) are modelled as coalgebras
for a functor $F = B \times (-)^\Alph$ on base categories of algebras or topological spaces.
The main observation used in \cite{BKP:WoLLIC} is that for many types of such coalgebras,
one can define a category of algebras that is dually equivalent to the category of coalgebras.
This dual equivalence can be seen as a generalisation of the Jonsson-Tarski duality known from modal logic,
which in turn arises from Stone duality.
The algebras in \cite{BKP:WoLLIC} are therefore understood as modal algebras, i.e.,
they consist of an algebra (that describes a propositional logic, e.g., Boolean logic)
expanded with the modal operators.
From this coalgebra-algebra duality it follows that maximal quotients of coalgebras
correspond to minimal subobjects of algebras.

The main contribution of \cite{BKP:WoLLIC} can then be formulated as follows:
Letting $\gamma$ be a coalgebra,
the minimal subalgebra of its dual modal algebra $\alpha$
consists of the predicates over $\gamma$ that are definable in the modal logic.
Hence to make $\gamma$ observable, compute the subalgebra of definable subsets
(which is reachable by construction), and dualise to obtain an observable coalgebra.
Although, this is not stated in \cite{BKP:WoLLIC}, for classic automata
the computation of definable subsets corresponds to the partition refinement algorithm.

The minimisation-via-duality approach of \cite{BKP:WoLLIC} was shown to apply to
partially observable DFAs 
(using duality of finite sets and finite Boolean algebras),
linear weighted automata 
(using the self-duality of vector spaces),
and belief automata viewed as coalgebras on compact Hausdorff spaces
(using Gelfand duality).
Moreover, for each of these examples it is shown that the definable subsets are determined by the
subsets definable in the trace logic fragment consisting of formulas of the shape 
$\boxmod{a_0} \cdots \boxmod{a_n}p$

In \cite{BBHPRS:ACM-TOCL},
Brzozowski's double-reversal minimisation algorithm \cite{Brz62} for classic automata 
was described categorically.
The Brzozowski algorithm works as follows.
Starting with a classic, finite (possibly nondeterministic) automaton accepting a language $\Lang$,
reverse the transitions,
swap initial and final states, and make the result deterministic using the subset construction.
This reversed automaton accepts the reversed language $\rev(\Lang)$.
Take the reachable part of the reversed automaton.
Now, do all of this again.
The result will be a reachable and observable (i.e., minimal) classic deterministic automaton accepting $\Lang$.
\rem{(In practice, one would do the subset construction and reachability on-the-fly such that only the reachable part is actually constructed,
rather than constructing the full powerset and only afterwards take the reachable part.)
}
The correctness of the algorithm was explained in \cite{BBHPRS:ACM-TOCL}
via the duality between reachability and observability known from control theory
(cf.~\cite{Kalman59,AZ69,Arbib74}.
This duality arises from a dual adjunction between algebras and coalgebras,
and therefore only works in one direction,
so to speak, namely, a reachable algebra dualises to an observable coalgebra, but not vice versa.
This, however, is sufficient to formalise Brzozowski's algorithm in terms of a
dual adjunction between categories of automata (with both initial and final states).

Generalisations of the Brzozowski algorithm were then formulated in \cite{BBHPRS:ACM-TOCL}
for Moore automata (over $\Set$)
and weighted automata, which include nondeterministic and linear weighted automata as instances.
More precisely, weighted automata were first determinised into
Moore automata over semimodules, and after the reverse-determinise step,
the semimodule structure is forgotten in order to take the reachable part.
Example~8.3 in \cite{BBHPRS:ACM-TOCL} illustrates that one generally wants to take a
subsemimodule that spans the reachable part, but this was not fully formalised.
One aim of the present paper is to make this part precise.

We summarise the main differences and similarities.
In \cite{BKP:WoLLIC}, the minimisation-via-duality approach
produces from a coalgebra (with structured state space),
an observable coalgebra of the same type.
In \cite{BBHPRS:ACM-TOCL}, the Brzozowski-based approach
starts with a $\Set$-based automaton
that possibly has branching structure specified by a monad $T$.
This automaton is determinised to yield a Moore automaton
over the category $\EM{T}$ of Eilenberg-Moore algebras for $T$,
and the result is a reachable and observable (i.e., minimal) Moore automaton
over $\EM{T}$.
If the automaton has no branching, we just proceed with Brzozowki over $\Set$.
In Appendix~\ref{app:example}, we give a small example illustrating the difference between
the two minimisation constructs on a concrete DFA.
In \cite{BKP:WoLLIC}, the perspective is based on modal logic.
Language semantics and reachability of automata is not an explicit part of the story,
although it is implicitly present via trace logic, however the connection to reachability
(in the usual set-theoretic sense) is not made.
In \cite{BBHPRS:ACM-TOCL}, the perspective is language-based. No link is made to modal logic.
\rem{
Towards the aim of finding a unified categorical of both algorithms,
a few observations can be made already.
A dual equivalence is a special case of a dual adjunction.
}

In the remainder of this section, we present a categorical picture
that unifies both approaches.
In particular,
our picture formalises the role of trace logic in the minimisation algorithms. Some of the technical details of this part are known from ~\cite{Rot16,KerKoeWes14,HermidaJacobs98,BBHPRS:ACM-TOCL} -- precise connections are detailed throughout the sections and in Section~\ref{sec:related-work}.

\subsection{Automata, Algebras and Coalgebras}
\label{sec:aut-alg-coalg}

Throughout this paper, we let $\Alph$ be a finite set. We will consider different types of automata, but they will all have input alphabet $\Alph$.

A classic deterministic automaton (on alphabet $\Alph$) consists of
a set $X$ (the state space),
a transition map $t \colon X \to X^\Alph$ (or equivalently $t \colon \Alph \times X \to X$), 
an acceptance map $f \colon X \to 2$, and
an initial state $i \colon 1 \to X$.
We generalise this basic definition to arbitrary categories as follows.

\begin{definition}\label{def:C-aut}
Let $\catC$ be a category, and let $\I$ and $\B$ be objects in $\catC$.
A~\emph{$\catC$-automaton 
(with initialisation in $\I$ and output in $\B$)}
is a quadruple $\str{X} = (X,t,i,f)$
consisting of
a state space object (or carrier) $X$ in $\catC$,
a $\Alph$-indexed set of transition morphisms $\{t_a \colon X \to X \mid a \in \Alph\}$, 
an initialisation morphism $i \colon \I \to X$, and 
an output morphism $f \colon X \to \B$.
A \emph{$\catC$-automaton morphism} from $\str{X}_1 = (X_1,t_1,i_1,f_1)$ to $\str{X}_2 = (X_2,t_2,i_2,f_2)$
is a $\catC$-morphism $h \colon X_1 \to X_2$
such that for all $a \in \Alph$, $h \o t_{1,a} = t_{2,a} \o h$,
$f_1 = f_2 \o h$, and 
$h \o  i_1 = i_2$. 
Together, $\catC$-automata with initialisation in $\I$ and output in $\B$,
and their morphisms form a category which we denote by $\Aut^{I,B}_\catC$.
\end{definition}

A classic deterministic automaton is then easily seen to be a $\Set$-automaton with output in $2$
and initialisation in $1$.

A central observation in \cite{BBHPRS:ACM-TOCL} is that
automata can be seen as coalgebras with initialisation, or dually, as algebras with output,
as we briefly recall now.
Assuming that $\catC$ has products and coproducts,
the transition morphisms  $\{t_a \colon X \to X \mid a \in \Alph\}$
correspond uniquely to morphisms of the following type:
\begin{equation}
\begin{array}{c}
\tup{t_a}_{a \in \Alph}\colon X \to X^\Alph\\
\hline\hline
[t_a]_{a \in \Alph}\colon \Alph \cdot X \to X
\end{array}
\end{equation}

Letting $F$ and $G$ be endofunctors on $\catC$ given by
$FX = \B \times X^\Alph$
and 
$GX = \I + \Alph \cdot X$, we see that
a $\catC$-automaton is an $F$-coalgebra $\tup{f,\tup{t_a}_{a \in \Alph}}\colon X \to B \times X^\Alph$
with intialisation $i\colon I \to X$.
Or equivalently, a $G$-algebra $[i,[t_a]_{a \in \Alph}]\colon GX \to X$ with output $f\colon X \to B$.


\subsection{Dual Adjunctions of Coalgebras, Algebras and Automata}
\label{sec:dual-adj-lift}

The categorical picture that unifies the work in \cite{BKP:WoLLIC} and \cite{BBHPRS:ACM-TOCL}
is sketched in the diagram \eqref{eq:unify-pic} below. 
This picture starts with a base dual adjunction that is lifted
to a dual adjunction between coalgebras and algebras.
This adjunction captures the construction in \cite{BKP:WoLLIC}
for obtaining observable coalgebras via duality.
The coalgebra-algebra adjunction is then lifted to a dual adjunction between automata
which captures the formalisation of the Brzozowski algorithm from \cite{BBHPRS:ACM-TOCL},
which uses automata with initial states.
\rem{
\blue{As shown above, automata have an algebra-part and a coalgebra part.
In order to generalise the notion of reachability on both sides
of the adjunction, we need an algebraic view on both sides.
Similarly, in order to generalise the notion of language semantics
on both sides, we need a coalgebraic view on both sides. [HH: I am thinking about how to present this with the minimal means...maybe we don't need the dual adj below on the right.]}
}
In the remainder of the section, we will explain the details
of how this picture comes about. 
\begin{equation}\label{eq:unify-pic}
\ba{ccc}
  \xymatrix@C=0.8em@R=4em{
    \left(\Aut_\catC^{I,R(O)}\right)^\op \ar@/^1pc/[rrrr]^-{\ol{\Pred}'} \ar@{}|{\top}[rrrr] \ar[d] &&&&
    \Aut_\catD^{O,L(I)} \ar@/^1pc/[llll]^-{\ol{\Spc}'} \ar[d]
    \\
    \Coalg_\catC(F_\catC)^\op \ar@/^1pc/[rrrr]^-{\ol{\Pred}} \ar@{}|{\top}[rrrr] \ar[d]
    &&&& \Alg_\catD(G_\catD) \ar@/^1pc/[llll]^-{\ol{\Spc}} \ar[d]
    \\
  \ar@(dl,ul)^-{F_\catC^\op} \catC^\op \ar@/^1pc/[rrrr]^-{\Pred} \ar@{}|{\top}[rrrr] &&&& 
  \ar@(dr,ur)_-{G_\catD} \catD \ar@/^1pc/[llll]^-{\Spc}
  }
\rem{
\xymatrix@C=0.8em@R=4em{
    \Aut_\catC^{I       ,R(O)} \ar@/^1pc/[rrrr]^-{\wt{L}'} \ar@{}|{\bot}[rrrr] \ar[d] &&&&
    \left(\Aut_\catD^{O,L(I)}\right)^\op \ar@/^1pc/[llll]^-{\wt{R}'} \ar[d]
    \\
    \Alg_\catC(G_\catC) \ar@/^1pc/[rrrr]^-{\wt{L}} \ar@{}|{\bot}[rrrr] \ar[d]
    &&&& \Coalg_\catD(F_\catD)^\op \ar@/^1pc/[llll]^-{\wt{R}} \ar[d]
    \\
  \ar@(dl,ul)^-{G_\catC} \catC \ar@/^1pc/[rrrr]^-{L} \ar@{}|{\bot}[rrrr] &&&& 
  \ar@(dr,ur)_-{F_\catD^\op} \catD^\op \ar@/^1pc/[llll]^-{R}
  }
  }
  \\\\
  F_\catC = \Spc (O) \times (-)^\Alph, \qquad
  G_\catD = O + \Alph \cdot (-)
\rem{
  \ba{rcl}
  F_\catC &=& \Spc (O) \times (-)^\Alph\\
  G_\catD &=& O + \Alph \cdot (-)
  \ea}
  
\rem{
&
  &
  \ba{rcl}
  F_\catD &=& \L (I) \times (-)^\Alph\\
  G_\catC &=& I + \Alph \cdot (-)
  \ea
}
\\
\rem{
\ba{l}
  \xi \colon  F_\catC \Spc \natIso \Spc G_\catD\\
  \rho\colon G_\catD \Pred \To \Pred F_\catC
  \ea
}
\rem{
&&
  \ba{l}
  \chi \colon \L G_\catC \natIso F_\catD \L\\
  \delta\colon G_\catC \Spc \To \Spc F_\catD
  \ea
}
\ea     
\end{equation}

\subsubsection{Base dual adjunction}\label{sec:base}

Our starting point is a dual adjunction $\Spc \dashv \Pred$ between categories $\catC$ and $\catD$
as in the above picture.
We will generally try to avoid the use of superscript $\op$,
and treat $\Pred$ and $\Spc$ as contravariant functors.
The units of the dual adjunction will be denoted
$\eta \colon \Id \To \Pred\Spc$
and
$\coun \colon \Id \To \Spc\Pred$.
The natural isomorphism of Hom-sets $\theta_{X,Y} \colon \catC(X,\Spc Y) \to \catD(Y, \Pred X)$,
will sometimes be written in both directions simply
as $f \mapsto \transp{f}$.
For $f \colon X \to \Spc{Y}$, its adjoint is $\transp{f} = \Pred{f} \o \eta_Y$,
and for $g \colon Y \to \Pred{X}$, its adjoint is $\transp{g} = \Spc{g} \o \coun_X$.

In all our examples, $\catC$ and $\catD$ are concrete categories, and 
the dual adjunction arises from homming into a dualising object $\dua$ (cf.~\cite{PorTho91}),
i.e., $\Pred = \catC(-,\dua)$ and $\Spc = \catD(-,\dua)$,
and we will often denote both of them by $\dua^{(-)}$.  
This means that adjoints are obtained simply by swapping arguments.
E.g., for $f\colon Y \to \dua^X$ we have
$\transp{f}(x)(y) = f(y)(x)$.
Moreover, the units are given by evaluation. E.g. $\eta_X \colon X \to \dua^{\dua^X}$
is defined by $\eta_X(x)(f) = f(x)$.

\begin{example}\label{exm:self-dual-powerset}
  A central example is the self-dual adjunction of $\Set$ 
given by the \emph{contravariant powerset functor} $\cPo = \Set(-,2)$
which maps a set $X$ to its powerset $2^X$ and a function $f:X\to Y$
to its inverse image map $f^{-1}\colon 2^Y \to 2^X$.
The functor $\cPo$ is dually self-adjoint with $\cPo^\op\dashv\cPo$,
and the isomorphism of Hom-sets is given by taking exponential transposes,
i.e., for $f\colon X \to 2^Y$ we have $\transp{f}\colon Y \to 2^X$.
\end{example}

Dual adjunctions are also called \emph{logical connections} as they form the basis of semantics
for coalgebraic modal logics \cite{BonKur05,Kli07,JacSok10}.
In this logic perspective,
$\catC$ is a category of state spaces,
$\catD$ is a category of algebras (e.g. Boolean algebras) encoding a propositional logic,
and the functor $G_\catD$ encodes a modal logic.
Intuitively, the adjoint $\Pred$ maps a state space $C$ to the predicates over $C$,
and $\Spc$ maps a predicate $A$ to the theories of $A$. 
The logic given by $G_\catD$ can be interpreted over
$F_\catC$-coalgebra by providing a so-called one-step modal semantics in the form of a
natural transformation 
$\rho\colon G_\catD \Pred \To \Pred F_\catC$,
or equivalently via its mate
$\xi \colon  F_\catC \Spc \To \Spc G_\catD$.
The pair $(G_\catD,\rho)$ is referred to as a logic.
By assuming that the initial $G_\catD$-algebra $(A_0,\alpha_0)$ exists, and viewing its elements as formulas,
the semantics of formulas in a $F_\catC$-coalgebra is $(C,\gamma)$ is obtained by initiality:
$\sm{G_\catD} \colon (A_0,\alpha_0) \to \Pred(\gamma) \o \rho_C$,
i.e., as an underlying $\catD$-map, it has type
$\sm{G_\catD} \colon A_0 \to \Pred(C)$.
Alternatively, the semantics can be specified by the theory map
$\th{G_\catD} \colon C \to \Spc(A_0)$ which is defined as 
the adjoint of $\sm{G_\catD}$.
We refer to \cite{BonKur05,Kli07,JacSok10} for a more detailed introduction
to coalgebraic modal logic via dual ajdunctions.

\subsubsection{Dual adjunction between coalgebras and algebras}\label{sec:liftcoalg}

We lift the base dual adjunction to coalgebras and algebras using some
some basic results from \cite{HermidaJacobs98,KerKoeWes14}. 
We assume that
$\catC$ has products, $\catD$ has coproducts,
and that we have functors $F_\catC$ and $G_\catD$ 
as given above, i.e.,
\[
F_\catC(C) = \Spc(O) \times C^\Sigma \quad\text{ and }\quad
G_\catD(D) =  O + \Sigma \cdot D
\]

We know from 
\cite[Cor.~2.15]{HermidaJacobs98} (see also \cite[Thm.~2.5]{KerKoeWes14}),
that
the dual base adjunction $\Spc \dashv \Pred$ lifts to a dual adjunction 
$\ol{\Spc} \dashv \ol{\Pred}$  between $\Coalg_\catC(F_\catC) = \Alg_{\catC^\op}(F_\catC^\op)$ and $\Alg_{\catD}(G_\catD)$
if there is a natural isomorphism $\xi\colon F_\catC \Spc \natIso \Spc G_\catD$.
We have for all $D \in \catD$,          
\begin{equation}\label{eq:xi}        
\Spc G_\catD(D) = \Spc(O + \Alph \cdot D) \cong \Spc(O)\times \Spc(D)^\Alph = F_\catC\Spc(D)
\end{equation}
since $\Spc$ (as a dual adjoint functor) turns colimits into limits.
Hence there is a natural isomorphism $\xi\colon F_\catC \Spc \natIso \Spc G_\catD$.
Let $\rho\colon G_\catD \Pred \To \Pred F_\catC$ be the mate of $\xi$,
i.e., the adjoint of $\xi_\Pred \circ F_\catC \coun$:
\begin{equation}\label{eq:mate}
\rho = \Pred F_\catC \coun \o \Pred \xi\Pred \o \eta G_\catD \Pred 
\end{equation}
The lifted adjoint functors are defined for
all $F_\catC$-coalgebras $\gamma\colon C \to F_\catC(C)$,
all  $F_\catC$-coalgebra morphisms $f$,
all $G_\catD$-algebras $\alpha \colon G_\catD(D) \to D$,
and all $G_\catD$-algebra morphisms $g$ by:
\begin{equation}\label{eq:ol-LR}
\begin{array}{rclc}
\ol{\Pred}(\gamma) &=&  \Pred\gamma \o \rho_C \colon G_\catD \Pred C \to \Pred C, & \ol{\Pred}(f) = \Pred(f) \\
\ol{\Spc}(\alpha) &=& \xi_D \o \Spc\alpha \colon \Spc D \to F_\catC \Spc D, & \ol{\Spc}(g) = \Spc(g) 
\end{array}
\end{equation}

\rem{
The following lemma states the precise result.
\begin{lemma}\label{lem:lift-coalg-alg}
Assume a dual adjunction $\Spc \dashv \Pred$ between $\catC$ and $\catD$
as described above, and that
the functors $F_\catC\colon\catC\to\catC$ and $G_\catD\colon\catD\to\catD$ are
as above, i.e.,
\[
F_\catC(C) = \Spc(O) \times C^\Sigma \quad\text{ and }\quad
G_\catD(D) =  O + \Sigma \cdot D
\]
where $\Alph$ is a finite set. 
Then the dual adjunction $\Spc \dashv \Pred$ 
lifts to a dual adjunction $\ol{\Spc} \dashv \ol{\Pred}$
between $\Coalg_\catC(F_\catC)$ and $\Alg_{\catD}(G_\catD)$.
The lifted adjoint functors are defined for
all $F_\catC$-coalgebras $\gamma\colon C \to F_\catC(C)$,
all  $F_\catC$-coalgebra morphisms $f$,
all $G_\catD$-algebras $\alpha \colon G_\catD(D) \to D$,
and all $G_\catD$-algebra morphisms $g$ by:
\begin{equation}\label{eq:ol-LR}
\begin{array}{rclc}
\ol{\Pred}(\gamma) &=&  \Pred\gamma \o \rho_C \colon G_\catD \Pred C \to \Pred C, & \ol{\Pred}(f) = \Pred(f) \\
\ol{\Spc}(\alpha) &=& \xi_D \o \Spc\alpha \colon \Spc D \to F_\catC \Spc D, & \ol{\Spc}(g) = \Spc(g) 
\end{array}
\end{equation}
\end{lemma}
\begin{proof}
The lemma follows from \cite[Cor.~2.15]{HermidaJacobs98} (or \cite[Thm.~2.5]{KerKoeWes14}),
and the fact that $\Coalg_{\catD^\op}(G_\catD^\op) = \Alg_\catD(G_\catD)$.
\end{proof}
}

\rem{
The lifting to algebras and coalgebras in the above diagram on the right comes about similarly.
For all $C \in \catC$, we have
\begin{equation}\label{eq:chi}
\Pred G_\catC (C) = \Pred(I + \Alph\cdot C) = \Pred(I) \times (\Pred C)^\Alph = F_\catD L(C)
\end{equation}
Hence there is a natural isomorphism $\chi \colon \Pred G_\catC \natIso F_\catD \Pred$.
Let $\delta\colon G_\catC \Spc \To \Spc F_\catD$ be the mate of $\chi^{-1}$,
i.e., the adjoint of $\chi_\Spc^{-1}\circ F_\catD \coun$:
\begin{equation}\label{eq:mate-delta}
\delta = \Spc F_\catD\coun \o \Spc \xi^{-1}\Spc \o \coun G_\catC\Spc
\end{equation}

\blue{
[HH: Need to check whether we really need this lemma.]
\begin{lemma}\label{lem:lift-alg-coalg}
Assume a dual adjunction $\Spc \dashv \Pred$ between $\catC$ and $\catD$
as described above, and that the functors
$G_\catC\colon\catC\to\catC$ and $F_\catD\colon\catD\to\catD$ are
as above, i.e.,
\[
F_\catD(D) = \Pred(I) \times C^\Sigma \quad\text{ and }\quad
G_\catC(C) =  I + \Sigma \cdot D
\]
where $\Alph$ is a finite set. 
Then the dual adjunction $\Spc \dashv \Pred$ 
lifts to a dual adjunction $\wt{\Spc} \dashv \wt{\Pred}$
between $\Alg_\catC(G_\catC)$ and $\Coalg_{\catD}(F_\catD)$.
The lifted adjoint functors are defined for
all $F_\catD$-coalgebras $\gamma\colon D \to F_\catD(D)$,
all  $F_\catD$-coalgebra morphisms $f$,
all $G_\catC$-algebras $\alpha \colon G_\catC(C) \to C$,
and all $G_\catC$-algebra morphisms $g$ by:
\begin{equation}\label{eq:wt-LR}
\begin{array}{rclc}
\wt{\Spc}(\gamma) &=&  \Spc\gamma \o \delta_D \colon G_\catC \Spc D \to \Spc D, & \wt{\Spc}(f) = \Spc(f) \\
\wt{\Pred}(\alpha) &=& \chi_C \o \Pred\alpha \colon \Pred C \to F_\catD \Pred C, & \wt{\Pred}(g) = \Pred(g) 
\end{array}
\end{equation}
\end{lemma}
\begin{proof}
The lemma follows from \cite[Thm.~2.14]{HermidaJacobs98} 
and the fact that $\Alg_{\catD^\op}(F_\catD^\op) = \Coalg_\catD(F_\catD)$.
\end{proof}
}
}

\begin{remark}\label{rem:isoF}
If $F'_\catC\colon\catC\to\catC$ is $F'_\catC(C) = B \times C^\Alph$ with $B \cong \Spc(O)$,
then $F'_\catC \natIso F_\catC$, and hence $\Coalg_\catC(F'_\catC) \cong \Coalg_\catC(F_\catC)$,
so we can think of $F'_\catC$-coalgebras as $F_\catC$-coalgebras.
\end{remark}

The isomorphism $\ol{\theta}$ of Hom-sets for $\ol{\Spc} \dashv \ol{\Pred}$
is simply the restriction of the isomorphism $\theta$ of Hom-sets
for $\Spc \dashv \Pred$ to the relevant morphisms.
\rem{
\[
\ba{c}
\gamma \sgoes{\transp{h}} \ol{\Spc}(\alpha)\\
\hline\hline
\alpha \sgoes{h} \ol{\Pred}(\gamma)
\ea
\]
}

The natural transformation $\rho \colon G_\catD\Pred \To \Pred F_\catC$
provides the one-step semantics for a modal logic for $F_\catC$-coalgebras
as described at the end of Section~\ref{sec:base}.
This makes most sense when the dual adjunction arises from a dualising object $\dua$
in which case $\dua$ is a domain of truth-values, i.e., the logic is $\dua$-valued,
and when $\catD$ is category of algebras with operations given by a signature $Sgn$.
The algebra functor $G_\catD = O + \Alph\cdot(-)$ then corresponds to a modal language
$\LLang{G_\catD}$ that has atomic propositions from $O$, labelled modalities $\boxmod{a}$, $a \in \Alph$,
and the propositional connectives are the operations from $Sgn$.
That is, formulas in $\LLang{G_\catD}$ are generated by the following grammar:
\[
  \phi ::= q \in O \mid \boxmod{a}\phi, a \in \Alph \mid \sigma(\Psi), \sigma \in Sgn
\]
where $\Psi$ is a set of formulas of cardinality matching the arity of the operation $\sigma$.

\rem{
For our specific choice of functors $F_\catC$ and $G_\catD$,
we compute the concrete definition of $\rho$ from \eqref{eq:mate}
when the adjunction arises from a dualising object.
\[\ba{lcccccl}
\rho_X\colon\\
O + \Alph\cdot\Pred{X} & \sgoes{\eta_{G\Pred{X}}}
& \Pred\Spc(O+\Alph\cdot\Pred{X}) & \sgoes{\sim}
& \Pred(\Spc{O}\times(\Spc\Pred{X})^\Alph) & \sgoes{\Pred F_\catC\coun_X}
& \Pred(\Spc{O} \times X^\Alph)
\\
O + \Alph\cdot\dua^{X} & \sgoes{\eta_{G(\dua^{X})}}
& \dua^{\dua^{(O+\Alph\cdot\dua^{X})}} & \sgoes{\sim}
& \dua^{(\dua^{O}\times(\dua^{\dua^{X}})^\Alph)} & \sgoes{\dua^{F_\catC\coun_X}}
& \dua^{\dua^{O} \times X^\Alph}
\\
q & \mapsto & \eta(q) & \mapsto &  \eta(q) & \mapsto &
\lambda (j,d) . j(q)
\\
\qquad (a,p) & \mapsto & \eta(a,p) & \mapsto &  \eta(a,p) & \mapsto &
\lambda (j,d) . p(d(a))
\ea\]
Here we used that the units evaluate, and that
\[\ba{lccc}
\dua^{F_\catC\coun_X} \colon
& \dua^{(\dua^{O}\times(\dua^{\dua^{X}})^\Alph)} & \sgoes{}
& \dua^{\dua^{O} \times X^\Alph}
\\
& h & \mapsto &
\lambda(j,d) . h (j, \lambda a. \lambda g. g(d(a)))
\ea\]
In short,
\begin{equation}\label{eq:rho-concrete}
\ba{rc}
\rho_X\colon\colon
& O + \Alph\cdot\dua^{X} 
\sgoes{}
\dua^{\dua^{O} \times X^\Alph}
\\
&
\ba{rcl}
\rho(q)(j,d) &=& j(q)\\
\rho(a,p)(j,d) &=& p(d(a))
\ea            
\ea
\end{equation}
}

For our specific choice of functors $F_\catC$ and $G_\catD$,
and when the adjunction arises from a dualising object $\dua$,
we can compute the concrete definition of $\rho$ from \eqref{eq:mate}
(see Appendix~\ref{app:rho-sem}) and we get the following $\dua$-valued modal semantics
of the language $\LLang{G_\catD}$:
\[
\ba{rcll}
\sem{q}(x) &=& j(q), & \text{where } \gamma(x) = \tup{j\colon \dua^O,d\colon X^\Alph}
\\
\sem{\boxmod{a}\phi}(x) &=& \sem{\phi}(d(a)), &  \text{where } \gamma(x) = \tup{j\colon \dua^O,d\colon X^\Alph}
\\
\sem{\sigma(\Psi)}(x) &=& \sigma(\{\sem{\psi}(x) \mid \psi \in \Psi\})
\ea\]

This shows that $\rho$ gives the expected modal semantics for
$F_\catC$-coalgebras viewed as deterministic $\Alph$-labelled Kripke frames with observations from $O$.
In particular, the modalities are ``deterministic'' Kripke box/diamond-modalities. 

\begin{example}\label{exa:dfa}
We consider the case of classic deterministic automata.
Here $\catC =\catD=\Set$, 
$F_\Set = 2 \times (-)^\Alph$ and $G_\Set = 1 + \Alph\cdot(-)$,
and the self-dual adjunction of $\Set$ is
given 
by the contravariant powerset functor $\cPo = \Set(-,2)$ (Example~\ref{exm:self-dual-powerset}).
The logic we obtain is \emph{trace logic} \cite{Kli07},
but here interpreted over DFAs rather than labelled transition systems as in \cite{Kli07}.
The initial $G_\Set$-algebra is $\Alph^*$, the set of finite words over $\Alph$,
and these are the formulas $\LLang{G_\catD}$,
since $\catD = \Set$ means that there are no propositional connectives.
The natural transformation $\rho$ has type
$\rho_X \colon 1 + \Alph\cdot 2^X \to 2^{2 \times X^\Alph}$,
and is given concretely here together with the induced semantics, where we write $x \Vdash \phi$ iff $\sem{\phi}(x) = 1$:
\[
\ba{rcl}
\rho_X(*) &=&  \{ (b,d) \in 2 \times X^\Alph \mid b = 1\}\\
\rho_X(a,U) &=&  \{ (b,d) \in 2 \times X^\Alph \mid d(a) \in U \}
\ea
\ba{||lcl}
x \Vdash *  & \iff & \text{$x$ is accepting}\\
x \Vdash \boxmod{a}\phi  & \iff & x \sgoes{a} y \text{ and } y \Vdash \phi 
\ea
\]
\end{example}

\subsubsection{Dual adjunction between automata}\label{sec:liftaut}

In order to obtain the upper adjunction in \eqref{eq:unify-pic} (which formalises Brzozowski),
we will use algebra and coalgebra structure on both sides,
hence we assume that $\catC$ and $\catD$ both have products and coproducts.
The lifting is a small extension of $\ol{\Spc} \dashv \ol{\Pred}$
obtained by defining how an initialisation map $I \to C$ for an $F_\catC$-coalgebra $\gamma$
is turned into an observation map $\Pred{C} \to \Pred{I}$ for the $G_\catD$-algebra $\ol{\Pred}(\gamma)$,
and vice versa for $\ol{\Spc}$.
\rem{
by defining $\ol{\Pred}'\colon \Aut_\catC^{I,\Spc O} \to \Aut_\catD^{O, \Pred I}$
and $\ol{R}' \colon \Aut_\catD^{O, \Pred I} \to \Aut_\catC^{I,\Spc O}$ as follows
for all $\gamma\colon C \to F_\catC C$ and $\alpha \colon G_\catD D \to D$:
\begin{equation}\label{eq:def-Aut-LR}
\begin{array}{rclc}
\ol{\Pred}'(\gamma, i\colon I \to C) &=&  (\ol{\Pred}(\gamma) \colon G_\catD\Pred C \to \Pred C, \Pred(i)\colon \Pred C \to \Pred I), & \ol{\Pred}'(f) = \Pred(f) \\
\ol{\Spc}'(\alpha, j\colon D \to \Pred I) &=& (\ol{R}(\alpha)\colon \Spc D \to F_\catC\Spc D,\transp{j}\colon I \to \Spc D) , & \ol{\Spc}'(g) = \Spc(g) 
\end{array}
\end{equation}
}

\begin{theorem}\label{thm:liftaut}
Under the assumptions of section~\ref{sec:liftcoalg},    
the dual adjunction 
$\ol{\Spc} \dashv \ol{\Pred}$ between $\Coalg_\catC(F_\catC)$ and $\Alg_{\catD}(G_\catD)$
lifts to a dual adjunction $\ol{\Spc}' \dashv \ol{\Pred}'$ 
between $\Aut_\catC^{I,\Spc O}$ and $\Aut_\catD^{O, \Pred I}$ by
defining $\ol{\Pred}'$ 
and $\ol{R}'$ 
as follows for all $\gamma\colon C \to F_\catC C$ and $\alpha \colon G_\catD D \to D$:
\[\begin{array}{rclc}
\ol{\Pred}'(\gamma, i\colon I \to C) &=&  (\ol{\Pred}(\gamma) \colon G_\catD\Pred C \to \Pred C, \Pred(i)\colon \Pred C \to \Pred I), & \ol{\Pred}'(f) = \Pred(f) \\
\ol{\Spc}(\alpha, j\colon D \to \Pred I) &=& (\ol{R}(\alpha)\colon \Spc D \to F_\catC\Spc D,\transp{j}\colon I \to \Spc D) , & \ol{\Spc}'(g) = \Spc(g) 
\end{array}
\]
\end{theorem}
\begin{proof}
This is a minor generalisation of Prop.~9.1 in \cite{BBHPRS:ACM-TOCL}.
It suffices to show that
for all $\catC$-arrows $i \colon I \to C$, and
all $\catD$-arrows $g \colon D \to \Pred{I}$ and $h\colon D \to \Pred{X}$:
$
g = \Pred{i} \o h \text{ iff } \transp{g} = \transp{h} \o i.
$
First, if $g = \Pred{i} \o h$, then
$\transp{g} = \Spc{g} \o \coun_I = \Spc{h} \o \Spc\Pred{i} \o \coun_I = \Spc{h} \o \coun_X \o i = \transp{h} \o i$,
where the third equality follows from naturality of $\coun$.
Conversely, if $\transp{g} = \transp{h} \o i$, then
$g = \Pred{\transp{g}} \o \eta_D = \Pred{i} \o \Pred{\transp{h}} \o \eta_D = \Pred{i} \o h$.
\end{proof}

\rem{
\blue{Define $\wt{L}'$ and $\wt{R}'$ similarly, and show that 
 $\wt{L}' = \ol{L}'$ and $\wt{R}' = \ol{R}'$.}
}




It is straightforward to verify that for our choice of $F_\catC$,
the final $F_\catC$-coalgebra exists,
and we usually view it as having carrier
$\Spc(O)^{\Alph^*}$, hence for $\gamma \colon C \to F_\catC(C)$,
the final morphism $!_\gamma \colon C \to \Spc(O)^{\Alph^*}$
assigns to each state in $C$ what can be seen as an $\Spc(O)$-weighted language. 
For $\str{X} = \tup{\gamma,i} \in \Aut_\catC^{I,\Spc(O)}$,
we define its language semantics as the composition
$\I \stackrel{i}{\to} C \stackrel{!_\gamma}{\to} \Spc(O)^{\Alph^*}$.
This $\catC$-morphism can be seen as an $\Alph^*$-indexed family of
$\catC$-morphisms $\langsem{\str{X}}_w \colon I \to \Spc{O}$
defined for all $w = a_1 \cdots a_k \in \Alph^*$ by
\[
\langsem{\str{X}}_w \;\;=\;\; I \sgoes{i} X \sgoes{t_{a_1}} \cdots  \sgoes{t_{a_k}} X \sgoes{f} \Spc(O)
\]
Computing the adjoint transpose $\transp{\langsem{\str{X}}_w} = \Pred{\langsem{\str{X}}} \o \eta_O$,
we get the $\catD$-morphism:
\[
\transp{\langsem{\str{X}}_w} \;\;=\;\; \Pred(I) \sfrom{i} \Pred(X) \sfrom{\Pred t_{a_1}} \cdots \sfrom{\Pred t_{a_k}} \Pred(X) \sfrom{\transp{f}} O 
\]
Hence $\transp{\langsem{\str{X}}_w} = \langsem{\ol{\Pred}'(\str{X})}_{w^R}$
where $w^R = a_k \cdots a_1$ is the reversal of $w$.
Similarly, we find that
for all $\str{Y} \in \Aut_\catD^{O,\Pred(I)}$,
$\transp{\langsem{\str{Y}}_w} = \langsem{\ol{\Spc}'(\str{Y})}_{w^R}$.
In the case of classic DFAs from Example~\ref{exa:dfa}
where $I=O=1$ and $\Spc(O) = \Pred(I) \cong 2$,
the above says that the adjoint functors reverse the language accepted by the automaton.

\rem{In the case of PODFAs and BAOs
where $O = \FBA(\Obs) = 2^{2^\Obs}$ for a finite set $\Obs$ of observations,
$\Pred(X) = 2^X$, $\Spc(A) = \Uf(A)$,
and we choose an initial state with $i \colon 1 \to X$,
\[\ba{rll}
\langsem{\str{X}}_w: & 1 \to 2^\Obs = \Uf(\FBA(\Obs)) & \text{ in } \Set\\
\hline\hline
\langsem{\ol{\Pred}'(\str{X})}_{w^R}: & \FBA(\Obs) \to 2^1 & \text{ in } \BA\\
\hline\hline
\langsem{\ol{\Pred}'(\str{X})}_{w^R}: & \Obs \to 2^1 & \text{ in } \Set\\
\ea\]
}


\subsection{Language Semantics and Trace Logic}\label{sec:trace-logic}

In this section, we give a general condition on the output sets that ensures
that we can link trace logic with the full modal logic via an adjunction.
This places trace logic in the general picture. 
In \cite{BKP:WoLLIC}, it was shown in each of the concrete examples that trace logic
is equally expressive as the full modal logic.
The results of this section give a general explanation of this fact.

Assume that the category $\catD$ is monadic over $\Set$
with adjunction
$
\begin{tikzcd}[sep=1.5em, cramped]
    \FreeD : \catD
          \arrow[r, shift left=1.7pt]
      & \Set : \ForgD
          \arrow[l, shift left=1.7pt]
\end{tikzcd}
$
This adjoint situation allows us to relate the $\Set$-based language semantics to
the final $F_\catC$-coalgebra semantics as we will show now.

Consider the functor $G \colon\Set\to\Set$ defined as
$G(X) = \Obs + \Alph \cdot X = \Obs + \Alph \times X$ where $\Obs$ is a finite 
set of observations.
Then the set $\Alph^*\Obs$ is an initial $G$-algebra with algebra structure
$\Obs + \Alph \times (\Alph^*\Obs) \to \Alph^*\Obs$
given by prefixing $\obs\in\Obs$ with the empty word
$\obs \mapsto \emptyword \obs$
and concatenation $(a,w) \mapsto aw$.
Let $\Phi_\catD \dashv U_\catD$ be an adjunction between $\Set$ and $\catD$.
Then we can compose with the dual adjunction $\Spc \dashv \Pred$
to obtain a dual adjunction between $\catC$ and $\Set$ as follows:
\begin{equation}\label{eq:comp-adj}
\xymatrix@C=0.8em@R=4em{
  \ar@(dl,ul)^-{F_\catC^\op} \catC^\op \ar@/^1pc/[rrrr]^-{\Pred} \ar@{}|{\top}[rrrr] &&&& 
  \ar@(dl,dr)_-{G_\catD} \catD \ar@/^1pc/[llll]^-{\Spc}   \ar@/^1pc/[rrrr]^-{\ForgD} \ar@{}|{\top}[rrrr] &&&& 
  \ar@(dr,ur)_-{G} \Set \ar@/^1pc/[llll]^-{\FreeD}
} 
\end{equation}

\begin{lemma}\label{lem:comp-adj}
Assume we have the situation in \eqref{eq:comp-adj}, and that $F_\catC$, $G_\catD$, $G$ are defined by:
\[
F_\catC(C) = \Spc\FreeD(\Omega) \times C^\Alph, \qquad
G_\catD(D) = \FreeD(\Omega) + \Alph \cdot D, \qquad
G(X) = \Obs + \Alph \cdot X.
\]
\rem{
\[\ba{llcr}
G \colon \Set \to \Set, & G(X) &=& \Obs + \Alph \cdot X
\\
G_\catD \colon \catD \to \catD, & G_\catD(D) &=& \FreeD(\Omega) + \Alph \cdot D
\\
F_\catC \colon \catC \to \catC, & F_\catC(C) &=& \Spc\FreeD(\Omega) \times C^\Alph
\ea
\]
}
Then \eqref{eq:comp-adj} lifts to
\begin{equation}\label{eq:comp-adj-lifted}
\xymatrix@C=0.8em@R=4em{
    \Coalg_\catC(F_\catC)^\op \ar@/^1pc/[rrrr]^-{\ol{\Pred}} \ar@{}|{\top}[rrrr] 
    &&&& \Alg_\catD(G_\catD) \ar@/^1pc/[llll]^-{\ol{\Spc}}   \ar@/^1pc/[rrrr]^-{\ol{\ForgD}}  \ar@{}|{\top}[rrrr] 
    &&&& \Alg_\Set(G) \ar@/^1pc/[llll]^-{\ol{\FreeD}} 
\rem{
\\
  \ar@(dl,ul)^-{F_\catC^\op} \catC^\op \ar@/^1pc/[rrrr]^-{\Pred} \ar@{}|{\top}[rrrr] &&&& 
  \ar@(dl,dr)_-{G_\catD} \catD \ar@/^1pc/[llll]^-{\Spc}   \ar@/^1pc/[rrrr]^-{\ForgD} \ar@{}|{\top}[rrrr] &&&& 
  \ar@(dr,ur)_-{G} \Set \ar@/^1pc/[llll]^-{\FreeD}
}
} 
\end{equation}
\end{lemma}
\begin{proof}
The dual adjunction on the left lifts because of a special case of \eqref{eq:xi}.
For similar reasons, the adjunction on the right lifts, because
there is a natural isomorphism $\kappa: \FreeD G \natIso G_\catD \FreeD$ that can
be obtained as follows
\begin{equation}\label{eq:kappa}
\kappa \colon \FreeD G X = \FreeD (\Obs + \Alph\cdot X) \cong \FreeD (\Obs) + \Alph\cdot \FreeD(X)
= G_\catD \FreeD(X),
\end{equation}
since $\FreeD$ (being a left adjoint) preserves colimits.
By \cite[Thm.~2.14]{HermidaJacobs98}, $\FreeD \dashv \ForgD$ lifts to an adjunction between
 $\ol{\FreeD} \dashv \ol{\ForgD}$  between $\Alg_\catD(G_\catD)$ and $\Alg_\Set(G)$ where the
 functor $\ol{\FreeD}$ maps a $G$-algebra $(X,\alpha)$ to the $G_\catD$-algebra
 $(\FreeD(X),\FreeD \alpha \circ \kappa^{-1})$.

By composition of adjunctions, also $\Spc\FreeD \dashv \ForgD\Pred$ lifts.
This could also be verified by noticing that for all sets $X$, there is natural isomorphism
\begin{equation}\label{eq:xi-trace}
  \xitrc \coloneqq  \Spc \kappa \circ \xi\FreeD : F_\catC\Spc\FreeD \natIso  \Spc\FreeD G
\end{equation}
where $\xi\colon F_\catC \Spc \natIso \Spc G_\catD$ from \eqref{eq:xi}    
is the mate of the modal logic $(G_\catD,\rho)$.
Hence by \cite[Thm.~2.14,Cor.~2.15]{HermidaJacobs98} (see also \cite[Thm.~2.5]{KerKoeWes14}),
the adjunction $\Spc\FreeD \dashv \ForgD\Pred$ lifts to one between $\Coalg_\catC(F_\catC)^\op$ and $\Alg_\Set(G)$.
\end{proof}

Letting $\rhotrc\colon G \ForgD\Pred \To \ForgD\Pred F_\catC$ be the mate of
$\xitrc$ from \eqref{eq:xi-trace}, then $(G,\rhotrc)$ is a modal logic for $F_\catC$-coalgebras. 
Since its formulas are the elements of the intial $G$-algebra of traces,
we refer to $(G,\rhotrc)$ as a trace logic.

\begin{lemma}\label{lem:trc-coincidence}
The theory maps $\th{G}$ and $\th{G_\catD}$ of the logics  $(G,\rhotrc)$ and  $(G_\catD,\rho)$ coincide.
\end{lemma}
\begin{proof}
Due to the adjunctions in \eqref{eq:comp-adj},
the intial $G$-algebra $\Trc$ of traces is mapped by $\ol{\FreeD}$ to an initial $G_\catD$ algebra,
which in turn is mapped by $\ol{\Spc}$ to a final $F_\catD$-coalgebra.
The coincidence of the theory maps follows from them being adjoints of the initial maps.
A more detailed argument is given in Appendix~\ref{app:theory-maps}.
\end{proof}

\rem{
\begin{example}
In the case where $\catD = \BA$, the fact that $\FreeD(\Trc)$ is an initial $G_\catD$-algebra implies that elements of the
Lindenbaum-Tarski algebra of $\LLang{G_\catD}$ are equivalent to a Boolean combination of trace formulas.
The elements of ${\FreeD}(\Trc)$ can therefore be seen as normal forms.
\end{example}
}


Since the mates $\xi$ and $\xitrc$ are both natural isomorphisms,
it follows from \cite{Kli07,JacSok10} (and $\catC$ having a suitable factorisation system, cf.~Theorem~\ref{thm:min-algo})
that the full modal logic $(G_\catD,\rho)$ and trace logic $(G,\rhotrc)$ are both
expressive for $F_\catC$-coalgebras.
In other words, the propositional connectives from $\catD$-structure in the logic language $\LLang{G_\catD}$
do not add any epxressive power to $\LLang{G}=\Trc$.
In summary, we arrive at the following proposition.
\begin{proposition}\label{prop:trc-expressive}
With the above assumptions, the trace logic $(G,\rhotrc)$
and the full logic $(G_\catD,\rho)$ are equally expressive over $F_\catC$-coalgebras,
meaning that for all $F_\catC$-coalgebras $\gamma\colon C \to F_\catC(C)$,
and all states $c_1, c_2$ in $C$ (recall that $\catC$ is a concrete category),
$c_1$ and $c_2$ are logically equivalent for $(G,\rhotrc)$
iff they are logically equivalent for $(G_\catD,\rho)$. 
\end{proposition}

By the uniqueness of final coalgebras up to isomorphism, 
it follows that there is an isomorphism
$\sigma\colon \Spc\FreeD(\Obs)^{\Alph^*} \stackrel{\sim}{\to} \Spc\,\FreeD(\Trc)$
which links the language semantics in the automata/coalgebraic sense
with trace logic semantics given by initiality.
\begin{proposition}\label{prop:lang-trc}
For all $F_\catC$ coalgebras $\gamma$, its language semantics defined
as the unique morphism into the final $F_\catC$-coalgebra
$\Spc\FreeD(\Obs)^{\Alph^*}$
corresponds to the trace theory map $\th{G}$ into the final $F_\catC$-coalgebra
$\ol{\Spc}\,\ol{\FreeD}(\Trc)$,
(and with the theory map $\th{G_\catD}$)
via the isomorphism $\sigma$.
\end{proposition}

We remark that it is straightforward to extend $\ol{\FreeD} \dashv \ol{\ForgD}$
to an adjunction of automata by taking adjoints of additional output maps to the algebras.
We omit the details.

Finally, we show that trace logic expressiveness can be extended to coalgebras for
what we can think of as subfunctors of $F_\catC$. This will be needed for the
topological automata in section~\ref{sec:top-gelfand}.

\begin{remark}\label{rem:subF}
  Let $F'_\catC$ be a functor on $\catC$ which preserves monos and
  such that there is a natural transformation
  $\tau \colon F'_\catC \To F_\catC$ which is abstract mono,
  i.e., all components are mono.
  Assume that $\catC$ has factorisation system $(E,M)$ with $E \subseteq Epi$
  and $M \subseteq Mono$.
  Defining  $\xi' = \xitrc \o \tau_\Spc$,
  then $\xi'\colon F'_\catC \Spc\FreeD \To \Spc\FreeD G$ defines semantics of trace formulas
  over $F'_\catC$-coalgebras which is essentially the same as the semantics over
  $F_\catC$-coalgebras.
  Since $\tau$ is abstract mono and $\xitrc$ is a natural iso,
  it follows that $\xi'$ is abstract mono, and hence the asscoiated logic
  is expressive~\cite{Kli07,JacSok10}.
\end{remark}

\subsection{Reachability and Observability }
\label{sec:reach-obs}

A main point emphasised in \cite{BBHPRS:ACM-TOCL} is that reachability is an algebraic concept, and
observability is a coalgebraic concept, and both concepts apply to automata
as they are both coalgebras and algebras.
We will call an algebra \emph{reachable} if it has no proper subalgebras,
and a coalgebra is \emph{observable} if it has no proper quotients.

Both \cite{BKP:WoLLIC} and \cite{BBHPRS:ACM-TOCL} use that a reachable algebra dualises to an
observable coalgebra, only the perspectives differ.
Note that in \cite{BKP:WoLLIC}, observable coalgebras are referred to as minimal automata.
%
In \cite{BKP:WoLLIC}, the reachable part of a $G_\catD$-algebra is defined as its least subalgebra,
and its existence was ensured by assuming that $\catC$ is wellpowered. 
In \cite{BBHPRS:ACM-TOCL}, automata were generally considered as automata over $\Set$,
and the reachable part of an automaton was  defined as the image of the
initial $G$-algebra inside the automaton (using its $G$-algebra structure, after possibly forgetting $\catD$-structure).
In Appendix~\ref{app:coincidence-reach}, we show that the two reachability notions in
\cite{BKP:WoLLIC} and \cite{BBHPRS:ACM-TOCL} coincide when conditions for both are satisfied.

In \cite{BKP:WoLLIC}, 
the least $G_\catD$-algebra of the dual $G_\catD$-algebra $\ol{\Pred}(\gamma)$
was characterised as the subalgebra $\alpha_\mathrm{Def}$ of
$\LLang{G_\catD}$-definable subsets of $C$
(or more abstractly $\dua$-valued predicates on $C$).
It was observed that $\alpha_\mathrm{Def}$ is generated by subsets
definable by trace formulas on $(C,\gamma)$.
The general statement of this fact is Proposition~\ref{prop:trc-expressive}.
Hence to compute $\alpha_\mathrm{Def}$, it suffices to compute trace logic definable subsets.
In the case of classic DFA, these subsets are precisely the reachable states
(in the usual transition-sense) of $\ol{\Pred}(\gamma)$.


%

\rem{
As mentioned, observability is a coalgebraic concept.
A coalgebra is observable if it has no proper quotients,
hence a coalgebra can be made observable by taking a largest quotient.
In \cite{BKP:WoLLIC}, largest quotients of $F_\catC$-coalgebras were ensured to exist by assuming that
$\catC$ is co-wellpowered which entails that $\Coalg_\catC(F_\catC)$ is co-wellpowered.
As with initial $G_\catD$-algebras, the final $F_\catC$-coalgebra generally did not exist. 
In \cite{BBHPRS:ACM-TOCL}, the final $F_\catC$-coalgebra always existed, and
observability was defined as the final morphism being injective.
It is straightforward to formulate and prove a dual version of Lemma~\ref{lem:wellp-initial}.
We omit the details due to lack of space.
}

The general setup described in Lemma~\ref{lem:comp-adj} is most closely related to
that of \cite{BBHPRS:ACM-TOCL},
as we have an initial $G$-algebra and an initial $G_\catD$-algebra. The latter is mapped by $\ol{\Spc}$ to a final $F_\catC$-coalgebra 
since the dual adjoint functors turn colimits into limits.
For this reason, $\ol{\Spc}$ maps epis to monos, but monos are not necessarily mapped to epis.
In particular,
we cannot argue that a least subalgebra of $\ol{\Pred}(\gamma)$
is mapped by $\ol{\Spc}$ to a largest quotient of $\gamma$. 
But using factorisation and existence of an initial $G_\catD$-algebra,
we obtain that the reachable part of $\ol{\Pred}(\gamma)$ is mapped by $\ol{\Spc}$
to an observable coalgebra.

\begin{proposition}\label{prop:reach-D}
Under the assumptions of Lemma~\ref{lem:comp-adj},
and assuming further that 
$\catD$ has a factorisation system $(E,M)$
such that $E \subseteq Epi$ and $M \subseteq Mono$,
we then have:

For all $(D,\delta) \in \Alg_\catD(G_\catD)$, let $\reach{D,\delta}$ be
the reachable part of $(D,\delta)$ obtained
by $(E,M)$-factorisation of the initial morphism:
$$
\ol{\FreeD}(\Trc,\alpha) \stackrel{e}{\surj} \reach{D,\delta} \stackrel{m}{\subto} (D,\delta).
$$
Then  $\ol{\Spc}(\reach{D,\delta})$ is an observable $F_\catC$-coalgebra.
\end{proposition}
\begin{proof}
The epimorphism $e\colon \ol{\FreeD}(\Trc,\alpha) \surj \reach{D,\delta}$
is mapped by $\ol{\Spc}$ to a monomorphism
 $$\ol{\Spc}(e) \colon \ol{\Spc}(\reach{D,\delta}) \subto \ol{\Spc}\ol{\FreeD}(\Trc,\alpha).$$
 Since $\ol{\Spc}\ol{\FreeD}(\Trc,\alpha)$ is a final $F_\catC$-coalgebra, we can conclude that 
  $\ol{\Spc}(\reach{D,\delta})$ is an observable $F_\catC$-coalgebra.
\end{proof}

The above proposition thus tells us how to obtain an observable $F_\catC$-coalgebra, and hence also
an observable $\catC$-automaton, by taking the reachable part on the dual side.

Extending the notion of reachable part to $\catD$-automata is done simply by taking the reachable part of their $G_\catD$-algebraic part and restricting the output map.
Brzozowski's algorithm produces a minimal $\catC$-automaton by taking the reachable part
of the resulting observable $\catC$-automaton, that is, with respect to the
algebraic structure of $\catC$-automata given by $G_\catC = \I + \Alph\cdot(-)$.
In order to do so, we need that also $\catC$ has an suitable factorisation system.


\subsection{Abstract minimisation algorithms}
\label{sec:abstract-brz}


We now put everything together into one diagram with which we can describe both approaches from
\cite{BKP:WoLLIC} and \cite{BBHPRS:ACM-TOCL} including the role of trace logic.

\begin{equation}\label{eq:comp-adj-all}
\ba{c}
\xymatrix@C=0.7em@R=4em{
    \left(\Aut_\catC^{\I,\Spc\FreeD(\Obs)}\right)^\op \ar@/^1pc/[rrrr]^-{\ol{\Pred}'} \ar@{}|{\top}[rrrr] \ar[d] &&&&
    \Aut_\catD^{\FreeD(\Obs),\Pred(I)} \ar@/^1pc/[llll]^-{\ol{\Spc}'} \ar@/^1pc/[rrrr]^-{\ol{\ForgD}'} \ar@{}|{\top}[rrrr] \ar[d] &&&&
    \Aut_\Set^{\Obs,\ForgD\Pred(I)} \ar@/^1pc/[llll]^-{\ol{\FreeD}'} \ar[d]
\\
    \Coalg_\catC(F_\catC)^\op \ar@/^1pc/[rrrr]^-{\ol{\Pred}} \ar@{}|{\top}[rrrr] \ar[d]
    &&&& \Alg_\catD(G_\catD) \ar@/^1pc/[llll]^-{\ol{\Spc}}   \ar@/^1pc/[rrrr]^-{\ol{\ForgD}}  \ar@{}|{\top}[rrrr] \ar[d]
    &&&& \Alg_\Set(G) \ar@/^1pc/[llll]^-{\ol{\FreeD}} \ar[d]
\\
\ar@(dl,ul)^-{F_\catC^\op} \catC^\op \ar@/^1pc/[rrrr]^-{\Pred} \ar@{}|{\top}[rrrr] &&&& 
  \ar@(dl,dr)_-{G_\catD} \catD \ar@/^1pc/[llll]^-{\Spc}   \ar@/^1pc/[rrrr]^-{\ForgD} \ar@{}|{\top}[rrrr] &&&& 
  \ar@(dr,ur)_-{G} \Set \ar@/^1pc/[llll]^-{\FreeD}
}
\\
F_\catC(C) = \Spc\FreeD(\Omega) \times C^\Alph, \quad
G_\catD(D) = \FreeD(\Omega) + \Alph \cdot D, \quad
G(X) = \Obs + \Alph \cdot X.
\ea
\end{equation}

\begin{theorem}\label{thm:min-algo}
Let $\catC, \catD$ be concrete categories, both having products and coproducts, and
both having factorisation systems $(E,M)$  such that $E \subseteq Epi$ and $M \subseteq Mono$.
Let $\Obs$ be a finite set (of observations), and $\I$  an (initialisation) object in $\catC$, and
assume that we have the adjoint situation between $\catC$, $\catD$, $\Set$ and functors
as described at the bottom level of \eqref{eq:comp-adj-all}.
Then the lower adjunctions lift to the upper two levels in \eqref{eq:comp-adj-all}
as shown in sections~\ref{sec:liftcoalg},~\ref{sec:liftaut}~and~\ref{sec:trace-logic},
and we have the following abstract algorithms:
\begin{description}
\item[Algo1]
  Given an $F_\catC$-coalgebra $\gamma$, compute $\ol{\Spc}(\reach{\ol{\Pred}(\gamma)})$ which will be an observable $F_\catC$-coalgebra.
\rem{\noindent\textbf{Algo1b}
  Given an $F_\catC$-coalgebra $\gamma$, compute $\ol{\Spc}\ol{\FreeD}(\reach{\ForgD\ol{\Pred}(\gamma)})$ which will be an observable $F_\catC$-coalgebra.\\
}
\item[Algo2]
  Given a $\catC$-automaton $(\gamma,i)$, compute $\reach{\ol{\Spc}'(\reach{\ol{\Pred}'(\gamma,i)})}$, which will be a reachable and observable (i.e., minimal)  $\catC$-automaton.\\
  \rem{
    \noindent\textbf{Algo2b}
  Given a $\catC$-automaton $(\gamma,i)$, compute $\reach{\ol{\Spc}'\,\ol{\FreeD}'(\reach{\ol{\ForgD}'\,\ol{\Pred}'(\gamma,i)})}$, which will be a reachable and observable (i.e., minimal)  $\catC$-automaton.\\
}
\end{description}

  \rem{
and assume there are contravariant adjoint functors
$
\begin{tikzcd}[sep=1.5em, cramped]
    \Pred : \catC
          \arrow[r, shift left=1.7pt]
      & \catD : \Spc
          \arrow[l, shift left=1.7pt]
\end{tikzcd}
$
and covariant adjoint functors
$
\begin{tikzcd}[sep=1.5em, cramped]
    \ForgD : \catD
          \arrow[r, shift left=1.7pt]
      & \Set : \FreeD
          \arrow[l, shift left=1.7pt]
\end{tikzcd}
$
Let $\Obs$ is a finite set of observations, and let $\I$ be an (initialisation) object in $\catC$,
and assume the following functor definitions:
$F_\catC(C) = \Spc\FreeD(\Omega) \times C^\Alph$,
$G_\catD(D) = \FreeD(\Omega) + \Alph \cdot D$,
$G(X) = \Obs + \Alph \cdot X$.
}
\end{theorem}

Of course, the abstract algorithms only become actual algorithms, when all structures involved have
finite representations.
Furthermore, we note that the one could consider an algorithm that uses the horizontally composed adjunction in \eqref{eq:comp-adj-all}, i.e.,
compute  $\ol{\Spc}\ol{\FreeD}(\reach{\ForgD\ol{\Pred}(\gamma)})$.
Although, the result will be an observable coalgebra, this is, however, not a good choice in general, because the reachable part is now computed over $\Set$, and this may yield an infinite coalgebra/automaton whereas it might have been finitely generated as a coalgebra/automaton over $\catD$. An example where this happens is found in Example~8.3 of  \cite{BBHPRS:ACM-TOCL}.

Although \cite{BKP:WoLLIC} does not describe a concrete algorithm,
the implicit abstract algorithm is essentially \textbf{Algo1},
since the conceptual emphasis is placed on computing the least $G_\catD$-subalgebra of $\ol{\Pred}(\gamma)$
as the subalgebra of $(G_D,\rho)$-definable subsets/predicates.
The characterisation of this $G_\catD$-subalgebra as being freely generated by the least $G$-subalgebra of
$\ol{\ForgD}\ol{\Pred}(\gamma)$
(i.e. the reachable part in the usual set-theoretic sense) can be viewed as an optimisation: to determine
the reachable part of a given $G_\catD$-algebra it suffices to compute the part that can be ``reached''/defined via trace formulas. 
This is the information contained in the right-hand side of the diagram.

In comparison,
Brzozowski's algorithm and its generalisation to weighted automata in Section~\ref{sec:weighted}
are instances of \textbf{Algo2} as they use initial states.
Classic Brzozowski is the case $\catC = \catD = \Set$, $G_\catD = G$, and $\Obs = \I = 1$.
The set-based algorithm for weighted automata in \cite{BBHPRS:ACM-TOCL} is neither of the above algorithms, but it can be described as constructing  $\reach{\ol{\ForgD}'\,\ol{\Pred}'(\gamma,i)}$, and then dualise back (without going through $\SMod$) to get a $\Set$-based Moore automaton. As mentioned above, this may result in the reachable part of the reversed automaton being infinite. 

In the case where the dual adjunction is a full duality, the initial state is easily found back in the observable coalgebra resulting from \textbf{Algo1} as its language equivalence class, so the extension to \textbf{Algo2} seems almost trivial. In case the dual adjunction is not a full duality, the transformation of the initial state goes via the adjunction, and factorisation on the dual side, and this is what Theorem~\ref{thm:liftaut} formalises.

We end this section by observing that the requirements regarding products, coproducts and
factorisation systems hold in all our examples, since $\catC$ and $\catD$ are monadic over $\Set$
meaning that they are equivalent to an
Eilenberg-Moore category $\EM{T}$ for a $\Set$-monad $T$.
For such a category $\EM{T}$, we know that it is
complete, cocomplete and exact \cite[Thm~4.3.5]{BorceuxII}.
W.r.t factorisation systems, $(Epi,Mono)$ is generally not a factorisation system for $\EM{T}$, rather $(RegEpi,Mono)$ is. Using that regular epis are the surjective morphisms, and monos are the injective morphisms, one can prove that in $\Coalg_\catC(F_\catC)$ and $\Alg_\catD(G_\catD)$ the surjective and injective morphisms form a factorisation system. We refer to Lemma~\ref{lem:factor-sys}.

\rem{
In particular, $\KHaus$ is the Eilenberg-Moore category of the ultrafilter monad \cite{Manes69}.
The monadicity of complex-valued commutative, unital $C^*$-algebras was proved in \cite{Negrepontis71},
see also \cite[Sec 4-c]{PorTho91}.
This specialises to the real-valued case.
In order to view  $\KHaus \cong\CCStar^\op$ as a concrete duality it seems necessary to enlarge
$\catC$ to e.g. locally compact Hausdorff spaces and $\CCStar$,
since $\bbR$ with the standard topology is not compact, but it is locally compact.
\blue{[Need to look closer at \cite[Sec 4-c]{PorTho91}.]}
}

\rem{
  \bi
  \setlength\itemsep{0pt}
  \item $\catC, \catD$ are concrete categories.
  \item  (dual) adjunctions and functors as indicated above.
  \item $\Obs$ is a finite set of observations.
  \item $F_\catC(C) = \Spc\FreeD(\Obs)) \times C^\Alph$ ($\catC$ has products).
  \item $G_\catD(D) = \FreeD(\Omega) + \Alph \cdot D$ ($\catD$ has coproducts).
  \item $G(X) = \Omega + \Alph \cdot X \cong \Omega + \Alph\times X$. 
  \item $\catC, \catD$ have (Epi,Mono)-factorisation systems.
  \item For Brzozowski (to get a minimal automaton, not only observable coalgebra): Object of initialisation $\I$ in $\catC$.
  \ei
}

\rem{
\begin{center}
\begin{tabular}{|l|c|c|c|c|}
\hline
automata & $\catC$, $\catD$ & $\dua$ & $\Obs$ \\
\hline\hline
 Brzozowski DFA & $\catC=\catD=\Set$ & 2 & 1 \\
\hline
Brzozowski Moore aut & $\catC=\catD=\Set$ & $B$ & 1\\
\hline
 PODFA/DKM & $\catC=\Set$, $\catD=\BA$ & 2 & $\Obs$ \\
 \hline
 Nondeterm. aut ($\Powf$-aut) & $\catC=\catD=\EM{\Powf} = \cat{JSL}$  & $\Two$ & 1  \\
 \hline
 Weighted aut ($\V$-aut) & $\catC=\catD=\SMod = \EM{\V}$ & $\S$ & 1 \\
 \hline
Linear Weighted aut ($\mathcal{V}_\k$-aut) & $\catC=\catD=\Veck = \EM{\mathcal{V}_\k}$ & $\k$ & 1 \\
 \hline
 Probabilistic aut ($\D$-aut) & $\catC=\KHaus$, $\catD=\CCStar$ & $\Real$ & 1\\
 \hline
 Alternating aut ($\N$-aut) & $\CABA = \EM{\N}$, $\Set$ & $2$ & 1\\
 \hline
\end{tabular}
\end{center}
}

\rem{
\[\xymatrix{
\cat{Aut}_\catC(\Lang) \ar[r]^-{\ol{\L}'}  & 
\cat{Aut}_{\catD^\op}(\rev(\Lang))^\op \ar[d]^-{\blue{reach^\op}} 
\\
\cat{Aut}_\catC(\Lang) \ar[d]_-{\blue{reach}}& 
\cat{Aut}_{\catD^\op}(\rev(\Lang))^\op \ar[l]_-{\ol{\Spc}'}
\\
\blue{\cat{Aut}_\catC(\Lang)} & 
}\]
}

\takeout{
\[\xymatrix{
\Coalg(FT)(\Lang) \text{ over }\Set \ar[d]^-{\det{(-)}} 
\\
\Aut_{\EM{T}}^{I,R(O)}(\Lang) \ar[r]^-{\ol{\L}'}  & 
\cat{Aut}_{\catD^\op}^{L(I),O}(\rev(\Lang))^\op \ar[d]^-{\blue{reach^\op}} 
\\
\cat{Aut}_{\EM{T}}^{I,R(O)}(\Lang) \ar[d]_-{\blue{reach}}& 
\cat{Aut}_{\catD^\op}^{L(I),O}(\rev(\Lang))^\op \ar[l]_-{\ol{\Spc}'}
\\
\blue{\cat{Aut}_{\EM{T}}^{I,R(O)}(\Lang)} & 
}\]
}

\rem{
\blue{Running Example: Brzozowski for Moore automata (includes classic DFA):}
\[\ba{lc}
& \xymatrix@C=1em@R=4em{
    \ar@(ul,ur)^-{F} \Set \ar@/^1pc/[rrrr]^-{\cPo} \ar@{}|{\bot}[rrrr] &&&& 
  \ar@(ul,ur)^-{G^\op} \Set^\op \ar@/^1pc/[llll]^-{\cPo^\op}
  }
 \\
 &
 \ba{lcl}
 F(X) = B \times X^\Alph & \qquad & G(X) = 1 + \Alph\cdot X
 \\
I = 1  & \qquad &  \cPo^\op(1) = B^1 \cong B 
 \ea
 \ea\]
 What is the modal logic? The logic of tests of the form: $\Alph^*p$ where $p$ is $B$-valued predicate?
}

\rem{
Overview of examples:
\rem{
\begin{center}
\begin{tabular}{|l|c|c|c|c|}
\hline
automata & $\catC$ & $\catD$ & dua.obj & restricts to duality?\\
\hline\hline
 DFA & $\Set$ & $\Set$ & 2 & No \\
\hline
 Moore aut & $\Set$ & $\Set$ & $B$ & No\\
\hline
 PODFA/DKM & $\Set$ & $\BA$ & 2 & $\cat{FSet} \cong \cat{FBA}^\op$\\
 \hline
 Nondeterm. aut ($\Powf$-aut) & $\EM{\Powf} = \cat{JSL}$ & $\EM{\Powf} = \cat{JSL}$  & $\Two$ &
 $\cat{JSL} \cong\cat{JSL}^\op$\\
 \hline
 Weighted aut ($\V$-aut) & $\SMod = \EM{\V}$ & $\SMod= \EM{\V}$ & $\S$ & ? (maybe fin.dim) \\
 \hline
Linear Weighted aut ($\mathcal{V}_\k$-aut) & $\Veck = \EM{\mathcal{V}_\k}$ & $\Veck = \EM{\mathcal{V}_\k}$  & $\S$ & $\FVec \cong \FVec^\op$ (fin.dim.)\\
 \hline
 Probabilistic aut ($\D$-aut) & $\KHaus$  & $\CCStar$ & $\Real$ & Gelfand: $\KHaus \cong\CCStar^\op$\\
 \hline
 Alternating aut ($\N$-aut) & $\CABA = \EM{\N}$ & $\Set$ & $2$ & $\CABA \cong \Set^\op$\\       
 \hline
\end{tabular}
\end{center}
}
\begin{center}
\begin{tabular}{|l|c|c|c|}
\hline
automata & $\catC$, $\catD$ & dua.obj & restricts to duality?\\
\hline\hline
 DFA & $\catC=\catD=\Set$ & 2 & No \\
\hline
 Moore aut & $\catC=\catD=\Set$ & $B$ & No\\
\hline
 PODFA/DKM & $\catC=\Set$, $\catD=\BA$ & 2 & $\cat{FSet} \cong \cat{FBA}^\op$\\
 \hline
 Nondeterm. aut ($\Powf$-aut) & $\catC=\catD=\EM{\Powf} = \cat{JSL}$  & $\Two$ &
 $\cat{JSL} \cong\cat{JSL}^\op$\\
 \hline
 Weighted aut ($\V$-aut) & $\catC=\catD=\SMod = \EM{\V}$ & $\S$ & ? (maybe fin.dim) \\
 \hline
Linear Weighted aut ($\mathcal{V}_\k$-aut) & $\catC=\catD=\Veck = \EM{\mathcal{V}_\k}$ & $\S$ & $\FVec \cong \FVec^\op$ (fin.dim.)\\
 \hline
 Probabilistic aut ($\D$-aut) & $\catC=\KHaus$, $\catD=\CCStar$ & $\Real$ & Gelfand: $\KHaus \cong\CCStar^\op$\\
 \hline
 Alternating aut ($\N$-aut) & $\CABA = \EM{\N}$, $\Set$ & $2$ & $\CABA \cong \Set^\op$\\       
 \hline
\end{tabular}
\end{center}
}


\section{Revisiting Examples}
\label{sec:old-ex}


\subsection{Deterministic Kripke Models}
\label{sec:deterministic}

A central example from~\cite{BKP:WoLLIC}
are deterministic Kripke models (in {\em loc.cit} referred
to as PODFAs, i.e., partially observable DFAs).
We will first recall the definitions of deterministic Kripke models and
their dual Boolean algebras with operators corresponding to a
modal logic of tests.
After that we will see how this duality can be seen
as a special case of our general duality picture, which has as
immediate corollary a minimisation algorithm for the case of finite models.
In addition, results from Section~\ref{sec:trace-logic} entail that
the modal test language without propositional  operators 
is sufficiently expressive to specify deterministic Kripke models up to bisimulation 
and to compute their observable quotient.

\begin{definition}
We define \emph{deterministic Kripke models} 
to be quintuples $\cS = (S,\Alph,\Obs,\gamma:S\to
2^{\Obs},\delta:S \to S^\Alph)$ where $S$ is a finite set of \emph{states}, $\Alph$ is a finite set
of \emph{actions}, $\Obs$ is a finite set of \emph{observations}, $\delta$
is a \emph{transition function} and $\gamma$ is an \emph{observation function}. 
A function $f:S_1 \to S_2$ is a morphism between Kripke models
$(S_1,\Alph,\Obs,\gamma_1,\delta_1)$
and $(S_2,\Alph,\Obs,\gamma_2,\delta_2)$
if for all $s \in S_1$ and all $a \in \Alph$ we have
$\gamma_1(s) =\gamma_2(s)$ and $f(\delta_1(s)(a)) = \delta_2(f(s))(a)$.
We let $\cat{DKM}$ denote the category of deterministic Kripke models.
\end{definition} 
In other words, deterministic Kripke models are Kripke models where for
each action $a \in \Alph$ the corresponding relation is the graph of a (total) function.
It is well-known that there is a duality between $\mathbf{DKM}$ and a suitable category $\cat{BAO}$ of
Boolean algebras. We will now recall the definition of $\cat{BAO}$ and some known facts
concerning this duality.

\begin{definition}
  The category $\cat{BAO}$ of (deterministic) \emph{Boolean algebras with
    operators} (BAOs) has as objects Boolean algebras $B$ with the
  usual operators $\land$ and $\neg$ with a greatest element $\top$ and
  least element $\bot$ together with unary operators $(a):B \to B$, for
  each action $a \in \Alph$, such that $(a)$ is a Boolean homomorphism.  For each
  observation $\omega \in \Obs$, we also have constants
  $\underline{\omega}$.  We denote an object of $\cat{BAO}$ by
 $$\cB =
  (B,\{(a)|a\in\Alph\},\{\underline{\omega}|\omega\in\Obs\},\top,
  \land,\neg).$$ 
  The $\cat{BAO}$ morphisms  are the usual Boolean homomorphisms
  preserving, in addition, the constants and commuting with the unary operators. 
  Finally we denote by $\cat{FBAO}$ the category of {\em finite} Boolean algebras with
  operators.
 \end{definition}

The following fact is well-known (cf.~e.g.~\cite{BdRV01, Givant09}).

\begin{fact}
    There is a dual adjunction between $\cat{Set}$ and $\cat{BA}$ as depicted in Figure~\ref{fig:DKM}
 given by the contravariant functor $\bbP$ 
 that maps a set to its Boolean algebra of subsets
 and the functor  $\Uf\mathrel{:=}\Hom(-,\bf{2})$, ie., the contravariant functor the maps a 
 Boolean algebra to its collection of ultrafilters.
 This adjunction restricts to a dual equivalence between the category  $\cat{FSet}$ of finite
 sets and the category $\cat{FBA}$ of finite Boolean algebras.
\end{fact}
%
We are now going to show how this example fits into our general
framework. As a corollary we obtain a minimisation procedure for finite
deterministic Kripke models.  

\begin{proposition}\label{prop:DKM(co)alg}
    We have the following equivalences:
    \begin{enumerate}
    \item $  \cat{DKM} \cong \Coalg_\Set(F) $ for $F = 2^\Obs \times X^\Alph$ 
    \item $\cat{BAO} \cong \Alg_\BA(G_\cat{BA})$ for $G_\cat{BA} =  \Phi_{\cat{BA}}(\Obs) + \Alph\cdot X$ 
    \end{enumerate}
\end{proposition}

\begin{figure}
\[\ba{lc}
&   \xymatrix@C=1em@R=4em{
    \ar@(ul,ur)^-{F^\op} \cat{Set}^\op \ar@/^1pc/[rrrr]^-{\bbP} \ar@{}|{\top}[rrrr] &&&& 
  \ar@(ul,ur)^-{G_\cat{BA}} \cat{BA} \ar@/^1pc/[llll]^-{\Uf}
  }
 \\
 &
 \ba{lcl}
 F(X) = 2^\Obs \times X^\Alph & \qquad &
 G_\cat{BA}(X) = \Phi_{\cat{BA}}(\Obs) + \Alph\cdot X
 \\
 I = 1 & \qquad &\Uf(\Phi_{\cat{BA}}(\Obs)) \cong 2^\Obs
 \ea
 \ea\]
 \caption{Functors and base adjunction for deterministic Kripke frames}
 \label{fig:DKM}
 \end{figure}
  
Both equivalences are an immediate consequence of the definitions. In the sequel,
we will make no distinction between $F$-coalgebras and deterministic Kripke models and, likewise,
between $G_{\cat{BA}}$-algebras and BAOs.
As a consequence of the proposition we obtain the following duality results by  applying our general framework.
\begin{proposition}
   The dual adjunction $\Uf \dashv \bbP$ lifts to a dual adjunction 
   between $\cat{DKM}$ and $\cat{BAO}$ and to an adjunction between 
   $\Aut_{\Set}^{1,2^\Obs}$ and $\Aut_{\cat{\BA}}^{2,F\Phi_{\cat{BA}}(\Obs)}$.
   If we start with the dual equivalence $\cat{FSet} \cong \cat{FBA}$, both liftings are dual equivalences as well.
   \end{proposition}
\begin{proof}
   For the dual adjunction between $\cat{DKM}$ and $\cat{BAO}$ recall from Proposition~\ref{prop:DKM(co)alg}
   that both categories are equivalent to categories of $F$-coalgebras and $G_\cat{BA}$-algebras for certain functors $F$ and $G_\cat{BA}$, respectively.
   Furthermore, we have $\Uf(\Phi_{\cat{BA}}(\Obs)) \cong 2^\Obs$, which follows
   from the well-known fact that the set of homomorphisms of type $\Phi_{\cat{BA}}(\Obs) \to \bf{2}$
   (i.e., ultrafilters)
 is in one-one correspondence with the set of functions of type $\Obs \to 2$.
   Therefore the functors $F$ and $G_\cat{BA}$ have the shape required by our general 
   lifting result from Section~\ref{sec:liftcoalg} and 
    we obtain functors $\ol{\bbP}: \Coalg(F)^\op \to \Alg(G_\cat{BA})$ and $\ol{Uf}:\Alg(G_\cat{BA}) \to \Coalg(F)^\op$
    with $\ol{Uf} \dashv \ol{\bbP}$.
    
    To extend the adjunction $\ol{Uf} \dashv \ol{\bbP}$ between $\Coalg(F)$ and $\Alg(G_\cat{BA})$ further to a dual adjunction $\ol{Uf}' \dashv \ol{\bbP}'$ between 
    $\Aut_{\Set}^{1,2^\Obs}$ and $\Aut_{\cat{\BA}^\op}^{2,\Phi_{\cat{BA}}(\Obs)}$
    - the latter is a slight extension of the former by adding a initial state
    to deterministic Kripke models and by viewing BAOs as some kind of automata with acceptance predicate - it suffices to note that
    $\bbP 1 \cong 2$ such that the result follows from the general theorem in Section~\ref{sec:liftaut}.
    
    The fact that the obtained adjunctions restrict to equivalences when we replace the base categories 
    $\cat{Set}$ and $\cat{BA}$ with $\cat{FSet}$ and $\cat{FBA}$, respectively, is a matter of routine checking.
\end{proof}

%

This shows, in particular, that we get a duality between finite deterministic Kripke models
and \cat{FBAO}s.  This is the key for obtaining a minimal
realization via logical theories.

\begin{definition}
Consider the language $\LLang{G_\cat{BA}}$:
\[ \varphi ::=  \top \mid \hat{\omega}, \omega \in \Obs \mid \boxmod{a} \varphi, a \in \Alph \mid \varphi_1\wedge \varphi_2
\mid \neg \varphi. \]
with semantics defined as below \eqref{eq:rho-concrete} on page~\pageref{eq:rho-concrete}. 
For a given automaton $\cS = (S,\gamma,\delta)$  we say that a subset $U$ of $S$ is \emph{definable by} $\LLang{G_\cat{BA}}$ if  $U=\sem{\varphi}$
for some $\varphi \in\LLang{G_\cat{BA}}$ where we identify the predicate $\sem{\varphi}: S \to 2$ with the set $\{ s \in S \mid \sem{\varphi}(s) =1\}$ .
We let $\mathrm{Def}(\cS) = (\mathrm{Def}(S),\{(a)_\cS \}_{a \in \Alph}, \{\sem{\hat{\omega}}\}_{\omega \in \Obs})$ be the BAO 
based on the definable subsets of $\cS$, where $(a)_\cS(\lsem \varphi \rsem) = \lsem \boxmod{a} \varphi \rsem$.

\end{definition}
In other words, the modal logic is the modal logic with (deterministic) 
$\Alph$-indexed modalities and atomic propositions from $\Obs$. 
For a given deterministic Kripke model, the algebra of definable subsets yields the reachable part (=zero generated subalgebra) of the dual automaton (=algebra).

\begin{proposition}
  Let $\cS=(S,\gamma,\delta) \in \cat{DKM}$ be a deterministic Kripke model, $s_I \in S$ an initial state and let $\mathrm{Def}(\cS)  \in \cat{BAO}$ the dual algebra
  of definable subset. 
  Then $\mathrm{Def}(\cS) \in \cat{BAO}$ is isomorphic to the reachable part  $\reach{\ol{\bbP} (\cS)}$ of 
  $\ol{\bbP} (\cS) \in \cat{BAO}$.
 \end{proposition}
\begin{proof}
    The result follows from the fact that  the unique morphism  $i_\mathrm{Def}(\cS)$ from the initial $G_{\cat{BA}}$-algebra $\cI$ to $\mathrm{Def}(\cS)$ 
    together with the embedding $m: \mathrm{Def}(\cS) \to \ol{\bbP} (\cS)$ is the
    image-factorisation of 
    $i_{\ol{\bbP} (\cS)}:\cI \to \ol{\bbP} (\cS)$,
    which is the $(E,M)$-factorisation obtained from the factorisation system of
    surjective and injective Boolean homomorphisms
    (cf.~Lemma~\ref{lem:factor-sys}).
\end{proof}

\begin{definition}
    We call a formula $\varphi$ of the form $ \boxmod{a_1}\dots \boxmod{a_n} \omega$ for some $n \in \bbN$
    and $a_i \in \Alph$ for $i \in \{1,\dots,n\}$ a trace formula. For a DKM $\cS$ we denote
    by $\mathrm{Def}^*(\cS)$ the collection of trace-definable subsets of $\cS$, i.e, the collection
    of subsets definable by a trace formula.    
\end{definition}

\begin{proposition}
    For all DKMs $\cS$ we have
    $$\ol{\Uf}(\mathrm{reach}(\ol{\bbP}(\cS))) \cong (\ol{\Uf}(\mathrm{Def}(\cS)) \cong (\ol{\Uf}(\ol{\Phi}_\cat{BA}(\mathrm{Def}^*(\cS))).$$
\end{proposition}
\begin{proof}
 By the results in Section~\ref{sec:trace-logic}, we have
 $\Phi_\cat{BA} \Trc$ is isomorphic to the Lindenbaum algebra of  $\LLang{G_\cat{BA}}$.
 Therefore the reachable part of $\ol{\bbP}(\cS)$ is obtained as
 the image of $\Phi_\cat{BA} \Trc$  under the initial morphism
 from $\Phi_\cat{BA} \Trc$ to $\ol{\bbP}(\cS)$ which can be easily checked to be $\Phi_\cat{BA}(\mathrm{Def}^*(\cS))$.
\end{proof}

We finish with a key observation from~\cite{BKP:WoLLIC} that allows to compute quotients of finite DKMs
via duality.

\begin{corollary}
    Given a finite DKM $\cS$, the quotient of $\cS$ modulo bisimulation is isomorphic to   $(\ol{\Uf}(\ol{\Phi}_\cat{BA}(\mathrm{Def}^*(\cS)))$.
\end{corollary}
\begin{proof}
    By Prop~\ref{prop:reach-D} we have that $(\ol{\Uf}(\mathrm{Def}(\cS))$ and thus $(\ol{\Uf}(\ol{\Phi}_\cat{BA}(\mathrm{Def}^*(\cS)))$ is observable.
    As $\cat{FSet}$ and $\cat{FBA}$ are dually {\em equivalent}, we get
    that $(\ol{\Uf}(\mathrm{Def}(\cS))$ and thus $(\ol{\Uf}(\ol{\Phi}_\cat{BA}(\mathrm{Def}^*(\cS)))$ are the maximal quotient of $\cS$.
\end{proof}


%

%

\subsection{Weighted Automata}
\label{sec:weighted}


\subsubsection{Semirings and semimodules} We need some basic definitions on semirings and semimodules to present the example of weighted automata. 

Recall that a \emph{semiring} is a tuple $(\S, +, \cdot, 0,1)$ where $(\S, + , 0)$ and $(\S,\cdot, 1)$ are monoids, the former of which is commutative, and multiplication  distributes over finite sums:
\[
\begin{array}{cccccc}
r \cdot 0 = 0 = 0 \cdot r 
& \qquad & 
r\cdot(s + t) = r \cdot s + r \cdot t
& \qquad & 
(r + s)\cdot t = r \cdot t + s \cdot t
\end{array}
\]
We just write $\S$ to denote a semiring. Examples of semirings are:  every field, the Boolean semiring $2$, the semiring $(\Nat,+, \cdot, 0, 1)$ of natural numbers, and the tropical semiring $(\Nat \cup \{\infty\}, \min, +, \infty, 0)$.  All these semirings are examples of {\em commutative} semirings, as the $\cdot$ operation is also commutative. 

For a semiring $\S$, an \emph{$\S$-semimodule} is a commutative monoid $(M, +, 0)$ with a left-action $\S \times M \to M$ denoted by juxtaposition $rm$ for $r \in \S$ and $m \in M$, such that for every $r,s \in \S$ and every $m, n \in M$ the following laws hold:
\[
\begin{array}{rcl@{\qquad}rcl}
  (r+s)m & = & rm + sm & r(m+n) & = & rm + rn \\
  0m &  = & 0 & r0 & = & 0 \\
  1m &  = & m & r(sm) & = & (r \cdot s) m
\end{array}
\]
Every semiring $\S$ is an $\S$-semimodule, where the action is taken to be just the semiring multiplication.  Semilattices are another example of semimodules (for the Boolean semiring $\S$).

An $\S$-semimodule homomorphism is a monoid homomorphism $h\colon M_1 \to M_2$ such that
$h(rm) = rh(m)$ for each $r \in \S$ and $m \in M_1$.  $\S$-semimodule homomorphisms are also
called $\S$-\emph{linear maps} or simply \emph{linear maps}.  The set of all linear maps from a $\S$-semimodule $M_1$ to $M_2$
is denoted by $\SMod(M_1,M_2)$.

\emph{Free $\S$-semimodules} over a set $X$ exist and can be built using
the functor $\V \colon \Set \to \Set$ defined on sets $X$ and maps $h\colon X \to Y$ 
as follows:
\[
\begin{array}{lcl}
  \label{eq:V}
  \V (X) & = &  \{\, \varphi\colon X \to \S \mid \text{$\varphi$ has finite support}\,\},
  \\[1em]
  \V (h(\varphi)) & = & \big(y \mapsto \sum_{x \in h^{-1}(y)} \varphi(x)\big),
\end{array}
\]
where a function $\varphi\colon X \to \S$ is said to have finite support if $\varphi(x) \neq 0$ holds 
only for finitely many elements $x \in X$. $\V(X)$ is the free $\S$-semimodule on $X$ when equipped with the following pointwise 
$\S$-semimodule structure:
\[
(\varphi_1 + \varphi_1 )(x) = \varphi_1(x) + \varphi_2(x)
\;\;\;\;
(s\varphi_1)(x) = s \cdot \varphi_1(x) \,.
\]
We sometimes  write the elements of $\V(X)$ as formal sums $s_1x_1 + \ldots + s_nx_n$ with $s_i \in \bbS$ and $x_i \in X$.
$\V(X)$ is a monad and the category of Eilenberg-Moore algebras is $\SMod$, 
the category of $\bbS$-semimodules and $\bbS$-linear maps. As usual, free $\bbS$-semimodules enjoy the following universal property: for every function $h\colon X \to M$ from a set $X$ to a semimodule $M$, there exists a unique linear map $h^{\sharp}\colon \V(X) \to M$ that is called the \emph{linear extension} of $h$.

We can define for an $\S$-semimodule $M$ its \emph{dual space} $M^\star$ to be the set $\SMod(M,\S)$ of all linear maps 
between $M$ and $\S$, endowed with the $\S$-semimodule structure obtained
by taking pointwise addition and monoidal action: $(g + h)(m) = g(m) + h(m)$, and 
$(sh)(m) = s \cdot h(m)$.  Note that $\S \cong V(1)$ and that $\S^\star = \SMod(\S,\S) \cong \S$.

\subsubsection{Weighted automata and weighted languages}

A \emph{weighted automaton} with finite input alphabet $\Sigma$ and weights over a semiring
$\S$ is given by a set of states $X$, a function $t\colon X \to \V(X)^\Sigma$ (encoding
the transition relation in the following way: the state $x \in X$ can make a transition to $y \in X$
with input $a \in \Sigma$ and weight $s \in \S$ if and only if $t(x)(a)(y) = s$), a final state
function $f\colon X \to \S$ associating an output weight with every state, and an initial
state function $i\colon 1 \to \V(X)$. Diagrammatically: 

  \[
  \xymatrix@C=.5cm{
 1\ar[rd]^{i}  & & {\bbS}&  && 
\\
 & X \ar[ru]^f \ar[d]^t & &&
 & 
\\
&  \V(X)^\Alph &&&
}
\]

The function $t\colon X \to \V(X)^\Sigma$ can be inductively extended to words $w\in \Sigma^*$: 
\begin{align*}
t(x)(\varepsilon) &= 1.x \\
t(x)(aw) &= v_1 t (x_1)(w) + \cdots +v_n t (x_n)(w) \text{, where }  t(x)(a)=v_1x_1+ \cdots +v_nx_n
\end{align*}

Weighted automata recognise functions in $\S^{\Sigma^*}$, or \emph{formal power series over $\S$}. 
More precisely, the formal power series recognised by a weighted automaton $\str{X} = (X,t,i,f)$
is the function $L(\str{X})\colon \Sigma^* \to \S$ that maps $w \in \Sigma^*$ to $f(t(i)(w)) \in \S$. More concretely, the value $L(\str{X})(w)$, for $w=a_1a_2\cdots a_n$, is the  sum of all $v_1 \cdot \ldots \cdot v_{n} \cdot f(x_{n+1})$ over all paths $p_w = x_1\xrightarrow{a_1,v_1} \ldots \xrightarrow{a_{n},v_{n}} x_{n+1}$ labelled by $w$. The value of $L(\str{X})(w)$ can be easily computed using the usual matrix representation of linear maps: the initial state function $i$ is then a column vector, the final state function $f$ is a row vector, and the transition relation $t$ can be represented as a $\Sigma$-indexed collection of $X\times X$-matrices $t_a$ where $t_a(y,x) = t(x)(a)(y)$ for all $x,y \in X$. $L(\str{X})(w)$ is then obtained by the following matrix multiplication $f \times t_{a_n} \times \ldots \times t_{a_0} \times i$.

Observe that $\bbS$ is (isomorphic to) the carrier of the free Eilenberg-Moore $\V$-algebra
on one generator $\V(1)$.
Hence, as described in Section~\ref{sec:det}, we can determinise a weighted automaton $\left <f,t\right >\colon X \to \bbS \times (\V X)^\Alph$ into a Moore automaton over $\SMod$.

$$\xymatrix@C=.5cm{
 1\ar[rd]^{i}  & & {\bbS}\ar@{-->}[rrrr]&  && & {\bbS}
\\
 & X \ar[ru]^f \ar[d]^t \ar[r]^\eta&V(X)\ar[ld]^{t^\sharp}\ar[u]^{f^\sharp}\ar@{-->}[rrr]
& &&\bbS^{\Alph^*} \ar[ru]^{o} \ar[d]^{d} 
 & 
\\
&  V(X)^\Alph \ar@{-->}[rrrr]&&&& (\bbS^{\Alph^*})^\Alph 
}
$$

The unique map into final Moore automaton of weighted languages gives precisely the language semantics concretely given above.

\subsubsection{Brozowski for Weighted Automata}

There is self-dual adjunction of $\SMod$ obtained by taking dual space: ${(-)^* = \SMod(-,\bbS)}$.
A special case is the self-dual adjunction of vector spaces in case $\S$ is a field,
which restricts to a duality between finite-dimensional vector spaces.
This duality was used in \cite{BKP:WoLLIC} to obtain observable Moore automata over vector spaces. 

We lift the base adjunction to one between Moore automata in $\SMod$ using Theorem~\ref{thm:liftaut}.
Let $\catC=\catD=\SMod=\EM{\V}$  and
$F_\SMod(X) = \bbS \times X^\Alph$ and $G_\SMod(X) = \bbS + \Alph\cdot X$.
Since $\S^* \cong \S$,
the conditions for Theorem~\ref{thm:liftaut} hold, and the adjunction lifts,
as illustrated here:
\[\xymatrix@C=1em@R=4em{
  \left(\Aut_{\SMod}^{\S,\S}\right)^\op \ar@/^1pc/[rrrr]^-{\ol{(-)^*}'} \ar@{}|{\top}[rrrr] \ar[d] &&&& 
  \Aut_{\SMod}^{\S,\S} \ar@/^1pc/[llll]^-{\ol{(-)^*}'} \ar[d]
  \\
  \ar@(dl,dr)_-{F_\SMod} \SMod^\op \ar@/^1pc/[rrrr]^-{{(-)}^*} \ar@{}|{\top}[rrrr] &&&& 
  \ar@(dl,dr)_-{G_\SMod^\op} \SMod \ar@/^1pc/[llll]^-{{(-)}^*}
  }
\]
We can now give the Brzozowski algorithm for weighted automata
by instantiating \textbf{Algo2} of Theorem~\ref{thm:min-algo} for the determinised automaton.
Start with a weighted automaton in $\Set$, determinise it into a Moore automaton in $\Aut_\SMod^{\S,\S}$ (to have a canonical representative of the accepted language), reverse and determinise, take the reachable part (w.r.t $G_\SMod$-structure over $\SMod$), reverse and determinise, take the reachable part again. Diagramatically, \textbf{Algo2} is (putting $^\op$ on the right-hand side to start and end in $\Aut_{\SMod}^{\S,\S}$):
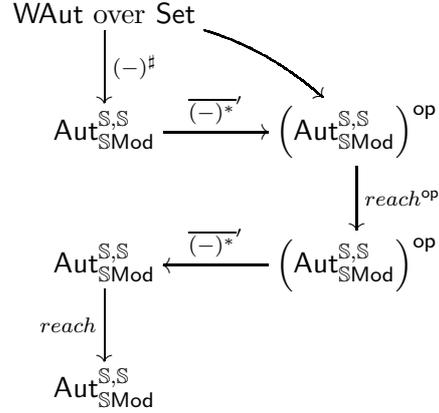
\[\xymatrix{
{\cat{WAut} \text{ over }\Set} \ar[d]^-{(-)^\sharp} \ar@/^1pc/[dr]
\\
\Aut_{\SMod}^{\S,\S} \ar[r]^-{\ol{(-)^*}'}  & 
\left(\Aut_{\SMod}^{\S,\S}\right)^\op\ar[d]^-{{reach^\op}} 
\\
\cat{Aut}_\SMod^{\S,\S} \ar[d]_-{{reach}}& 
\left(\Aut_{\SMod}^{\S,\S}\right)^\op \ar[l]_-{\ol{(-)^*}'}
\\
\Aut_{\SMod}^{\S,\S}
&
}\]
\rem{
\[\xymatrix{
\left(\Aut_{\SMod}^{\S,\S}\right)^\op \ar[r]^-{\ol{(-)^*}'}  & 
\cat{Aut}_\SMod^{\S,\S} \ar[d]^-{{reach}} 
\\
\left(\Aut_{\SMod}^{\S,\S}\right)^\op \ar[d]_-{{reach^\op}}& 
\cat{Aut}_\SMod^{\S,\S}\ar[l]_-{\ol{(-)^*}'}
\\
\left(\Aut_{\SMod}^{\S,\S}\right)^\op   
&
}\]
}
At this point we have built a minimal Moore automaton accepting the same language as the weighted automaton we started with and, moreover, the state space is a subsemimodule of the semimodule generated by the original state space. 

The last step missing is to recover a weighted automaton in $\Set$, with as state space the generators of the state space of the  minimal Moore automaton resulting after applying our Brzozowski algorithm. 
Unfortunately, subsemimodules of free, finitely generated semimodules are not necessarily free and finitely generated. Therefore our construction does not guarantee, in general, that the resulting automaton is actually a weighted automaton in $\Set$. Fortunately, we know from a result of Tan~\cite{tan16} that for a commutative semiring $\S$, every nonzero subsemimodule $N$ of a finitely generated free $\S$-semimodule $M$ is free if and only if $\S$ is a principal ideal domain (\cite[Theorem 4.3]{tan16}). Furthermore, because $N$ is free, it follows that it is also finitely generated and of rank smaller than that of $M$ [(\cite[Theorem 4.3]{tan16}). In other words, the minimal weighted automaton over a principal ideal domain exists and has a state space smaller or equal than that of the original automaton if the latter is finite. 

Recall that a principal ideal domain is an integral domain in which every ideal is principal, i.e., can be generated by a single element. Examples include any Euclidean domain, thus any field, the ring of integers, the ring of polynomials in one variable with coefficients in a field, and the ring of formal power series over a field and one variable. The ring of polynomials in two or more variables and the ring of polynomials with integer coefficients are not principal ideal domains.


\subsection{Topological Automata via Gelfand Duality}
\label{sec:top-gelfand}


A very popular model heavily used in reforcement learning is the
\emph{partially observable Markov decision process} (POMDP).  The idea is
that one can only see the observations and not exactly which state the
system is in.  Many algorithms in machine learning deal with this situation
by constructing a new automaton called the \emph{belief automaton}.  The
state space of this automaton is the set of probability distributions on
$S$.  When seeking to minimize this using duality~\cite{BKP:WoLLIC},
the original idea was to exploit the fact that the state space of the
belief automaton is a compact Hausdorff space and use Gelfand duality.
However, we have since felt that convex duality is a better match for this
situation.  Nevertheless, the notion of a topological automaton is
interesting in its own right and may be the basis for later extensions and
examples.  This section, therefore develops Gelfand duality and its
application to topological automata.

Given a finite set $X$ we write \(\SubD{X}\) for the set of discrete
\emph{sub}distributions on $X$ endowed with the relative topology when
viewed as a subset of \([0,1]^X\).  This is a compact Hausdorff space.
\begin{definition}
A \textbf{compact Hausdorff automaton} is a $5$-tuple
\[ \mathcal{H} = (S,t:S\x\Sigma\to S, f:S\to\SubD{\Obs}) \]
where $S$ is a compact Hausdorff space, $\Alph$ is a finite set of actions
or inputs, $\Obs$ is a finite set of observations, $t$ is a
\emph{continuous} transition function and $f$ is a continuous
observation function.   
\end{definition}

We recall a few basic facts about $C^*$-algebras, and refer to
\cite{Arveson76,Blackadar06,Joh82,Sakai71} for further information.
Usually $C^*$-algebras are considered as algebras over the complex field.
Here, we are concerned with probabilistic computation, and therefore we
consider $C^*$-algebras over the field $\bbR$ of real numbers.

A \emph{(real-valued) Banach algebra} $A$ is Banach space (complete normed
real vector space) equipped with an associative multiplication such that
\( \norm{xy} \leq \norm{x}\norm{y}\)
for all $x,y$.  This requirement makes multiplication continuous in the
norm topology.  A \emph{(real) $C^*$-algebra} is a Banach algebra together
with an \emph{involution} $(-)^*$ which is a linear, norm-preserving map on
$A$ such that $(xy)^* = y^*x^*$ and $(x^*)^*$, and which in addition
satisfies the $C^*$-axiom: $\|x^*x\| = \|x\|^2$ for all $x \in A$.  A
$C^*$-algebra $A$ is \emph{unital} if it has a multiplicative unit $1$
whose norm is $1 \in \bbR$, and $A$ is \emph{commutative} if the
multiplication is commutative.

A \emph{homomorphism of $C^*$-algebras} is a bounded, linear map that
preserves the multiplication and the involution.  A homomorphism of unital
$C^*$-algebras is additionally required to preserve the unit.  We denote by
$\CCStar$ the category of unital, commutative, real-valued $C^*$-algebras
and their homomorphisms.

In \cite{Negrepontis71} it was shown~\footnote{Strictly speaking, she
  showed it for complex-valued $C^*$-algebras, but the result also holds
  for real-valued ones.} 
that $U$ has a left adjoint
$M \colon \Set \to \CCStar$ given by
\begin{eqnarray}\label{eq:free-cstar}
M(X) & = & C([0,1]^X) = \{ f\colon [0,1]^X \to \bbR \mid f \text{ continuous}\}\\
M(g\colon X \to Y) & = & f (v \circ g) \quad \text{ where } f\in C([0,1]^X), v \in [0,1]^Y. 
\end{eqnarray}
where $[0,1]^X$ is equipped with the product topology.

We denote by $\KHaus$ the category of compact Hausdorff spaces and
continuous maps.  Given a compact Hausdorff space $X$, the hom-set
$\Cont(X)= \Hom_\KHaus(X,\bbR)$ becomes a commutative, unital, real-valued
$C^*$-algebra by defining operations pointwise. In particular, the unit is
the constantly 1 map, and for $f \in \Cont(X)$, and the norm is
$\|f\| = \sup\{|f(x)| \mid x \in X\}$; recall that for a compact space and
a continuous function the supremum is attained.  For a morphism
$g \colon X \to Y$ in $\KHaus$, defining $\Cont(g)(h) = h \circ g$ makes
$\Cont(-)$ a functor from $\KHaus$ to $\CCStar^\op$.

Conversely, for $A \in \CCStar$, the set $\hat{A} = \Hom_\CCStar(A,\bbR)$
becomes a compact Hausdorff space (called the \emph{spectrum of $A$}) by
equipping it with the weak $^*$-topology $\tau$ which is generated by the
sets $O_x = \{ \Phi \in \hat{A} \mid \Phi(x) \neq 0 \}$ for all $x \in A$.
We define $\Spec(A) = (\hat{A},\tau)$.  For a morphism $h\colon A \to B$ in
$\CCStar$, defining $\Spec(h)(\Phi) = \Phi \circ h$ makes $\Spec$ a functor
from $\CCStar^\op$ to $\KHaus$.

The functors $\Cont$ and $\Spec$ establish a dual equivalence between
$\KHaus$ and $\CCStar$ known as Gelfand duality
\begin{equation}\label{eq:gelfand-dua}
\xymatrix@R=4em{
   \KHaus \ar@/^1pc/[r]^-{\Cont} \ar@{}|{\cong}[r] & 
   \CCStar^\op \ar@/^1pc/[l]^-{\Spec}
}
\end{equation}

For the purposes of this paper, we only need a dual adjunction.  We will
take $\Cont$ to be the right adjoint.  As this dual adjunction is in fact a
dual equivalence, the unit and the counit of this adjunction are natural
isomomorphisms.  The unit $\eta_A\colon A \to \Cont(\Spec(A))$ is
known as the \emph{Gelfand transform}, and is given by
$\eta_A(x)(\Phi) = \Phi(x)$.  For all $A \in \CCStar$, $\epsilon_A$ is
an isometric isomorphism in $\CCStar$.

We will lift the base dual adjunction between $\KHaus$ and $\CCStar$ to an
adjunction between the category $\cat{CHA}$ of compact Hausdorff automata
and $\cat{CAO}$ of $\CCStar$-automata. These are obtained from
$F$-coalgebras and $G$-algebras, respectively, where 
\begin{equation}\label{eq:gelfand-fctrs}
\ba{lll}
  F \colon \KHaus \to \KHaus,   & F(X) = \SubD(\Obs) \times X^\Sigma \\
  G \colon \CCStar \to \CCStar, & G(A) = M(\Obs)/J + \Sigma\cdot A
\ea\end{equation}
where $\SubD(\Obs)$ is the set of discrete, subdistributions on the set
$\Obs$ of observations equipped with the relative topology viewed as a
subset $\SubD(\Obs) \subseteq [0,1]^\Obs$. This makes $\SubD(\Obs)$ a
compact Hausdorff space. 
 Recall from \eqref{eq:free-cstar} that $M$ is the
left adjoint of the unit interval functor $U$.  Finally, $J \subseteq M(\Obs)$ is an
ideal of the $\CCStar$-algebra $M(\Obs)$ which we describe in a moment.
Note that $\CCStar$ has coproducts.  This can readily be seen from the fact
that $\KHaus$ has products and using Gelfand duality.

 In order to lift the base dual adjunction $\Spec \dashv \Cont$
 to a dual adjunction between $\Coalg(F)$ and $\Alg(G)$ as
 in section~\ref{sec:liftcoalg},
 we need to show that $\Spec(M(\Obs)/J) \cong \SubD(\Obs)$.
First, we define the ideal $J$.
Fix a finite set $Y$ and consider the $C^*$-algebra $M(Y) \in \CCStar$
defined by \([0,1]^Y\).
For each $y \in Y$, we have a projection map $\pi_y \in M(Y)\to \Real$
given by $\pi_y(v) = v(y)$.  Let $\pi = \sum_{y \in Y}\pi_y$.
Then $\pi\colon [0,1]^Y \to \Real$ is linear and $\pi \in  M(Y)$.
We will take $J$ to be the ideal corresponding to the congruence generated
by the equality obtained by rewriting $\pi \pleq 1$ as an equality as
follows: 
\[\begin{array}{rcl}
\pi \pleq 1 & \iff & \pi \lor 1 = 1 \\
 & \iff & \tfrac{1}{2}(\pi+1) + \tfrac{1}{2}|\pi+1| = 1\\
 & \iff & |1 - \pi| = 1 - \pi
\end{array}\]

\begin{definition}\label{def:cong-J}
  We define the ideal $J$ of $M(Y)$ as the principal ideal generated by the
  element $(|\pi^-| - \pi^-)$ where  $\pi^- := 1-\pi$. That is, 
  \[
  J = \{ m \in M(Y) \mid \exists k \in M(Y): m = k(|\pi^-| - \pi^-)\}.
  \]
  The congruence relation  $\icong{J}$ on $M(Y)$ arising from the ideal $J$
  is then defined standardly as follows: 
For $m, n \in M(Y)$, $m \icong{J} n \quad\text{ if }\quad  m-n \in J.$
We write $M(Y)/J$ for the quotient of $M(Y)$ with respect to $\icong{J}$. 
\end{definition}

The rather technical proof of the following isomorphism lemma is in Appendix~\ref{app:topo-aut-iso}.

\begin{lemma}\label{lem:topo-aut-iso}
For any set $Y$,  $\SubD(Y) \cong \Spec(M(Y)/J)$ in $\KHaus$.
\end{lemma}
\rem{
\begin{proof}
  Let $\alpha\colon \SubD(Y) \to \Spec(M(Y)/J)$ be defined by
  $\alpha(\phi)(m) = m(\phi)$ where $\phi \in \SubD(Y)$ and $m \in M(Y)$.
  Note that this is well defined.  If $m \icong{J} m'$ then their
  difference lies in $J$ which means that
  \(m(\phi) - m'(\phi) =k(\phi)(|\pi^-|-\pi^-)(\phi)\).
  This condition is equivalent to \(1-\pi\)
  is positive and hence $|\pi^-|$ and $\pi^-$ are equal and hence the
  second term is $0$, whence \(m(\phi) = m'(\phi)\).
  Note that the topology of \(\Spec(M(Y)/J)\)
  is generated by the taking as the closed sets, sets of maximal ideals
  that contain a fixed element $\phi$ of $M(Y)$.  Any maximal ideal
  consists of the functions that vanish at a point $y$, call this $m_y$.
  So if we fix such an $\phi$ for it to be in a maximal ideal $m_y$, we
  have $\phi(y) =0$.  This means that $\alpha^{-1}$ of a closed set is the
  set of subdistributions that assign $0$ to a particular element $y$; this
  is a closed set so $\alpha$ is continuous.

Let $\beta\colon \Spec(M(Y)/J) \to \SubD(Y)$ be defined by
$\beta(\Phi)(y) = \Phi(\pi_y)$ where $\pi_y \in M(Y)/J$ projects onto $y$.
We check that $\beta(\Phi)$ is a subdistribution:
\[\sum_{y\in Y}\beta(\Phi)(y)
= \sum_{y\in Y}\Phi(\pi_y)
= \Phi\left(\sum_{y\in Y} \pi_y \right)
\pleq \Phi(1) = 1
\]
The second and last identity hold because $\Phi$ is a
$\CCStar$-homomorphism (hence linear and unital); 
the inequality holds since we are in $M(Y)/J$ (which says that $\pi\pleq
1$) and $\Phi$ is monotone.

We now show that for $\phi \in \SubD(Y)$,
$\beta(\alpha(\phi)) = \phi$. Let $y \in Y$, we then have:
$
\beta(\alpha(\phi))(y) = \alpha(\phi)(\pi_y) = \pi_y(\phi) = \phi(y)
$
We show that for $\Phi \in \Spec(M(Y)/J)$, $\alpha(\beta(\Phi)) = \Phi$.
Let $m\in M(Y)$, we then have
$
\alpha(\beta(\Phi))(m) = m(\beta(\Phi)).
$
By the Stone-Weierstrass theorem, the polynomials on the compact Hausdorff
space $[0,1]^Y$ are dense in $\Cont([0,1]^Y)$. 
Since $m \colon [0,1]^Y \to \bbR$ is continuous, 
it therefore suffices to show that for all polynomials $p$ on $[0,1]^Y$ we
have that $p(\beta(\Phi)) = \Phi(p)$.  

\emph{Case $p = 1$:} $p(\beta(\Phi)) = 1 = \Phi(1) = \Phi(p)$.
\emph{Case $p = r \in \bbR$:} $p(\beta(\Phi)) = r = \Phi(r) = \Phi(p)$.
\emph{Case $p = \sum_{y\in Y} r_y \pi_y$:}
$\Phi(p) = \Phi(\sum_{y\in Y} r_y \pi_y) = \sum_{y \in Y}r_y\Phi(\pi_y) = \sum_{y \in Y} r_y\beta(\Phi)(y) = p(\beta(\Phi))$. 
Finally, let $p = \sum_{y\in Y} r_y \pi_y$ and $q = \sum_{y\in Y} s_y \pi_y$.
Then
$
\Phi(pq) = \Phi(p)\Phi(q) = p(\beta(\Phi))q(\beta(\Phi)) = (pq)(\Phi).
$
It follows that $p(\beta(\Phi)) = \Phi(p)$ holds for all polynomials
$p$\footnote{Clearly all polynomials can be expressed as sums of products
  of lower degree polynomials.}
and we have now shown that $\alpha \colon \SubD(Y) \to \Spec(M(Y)/J) $ is a
bijection with inverse $\beta$.  
\end{proof}
}

From Lemma~\ref{lem:topo-aut-iso} and section~\ref{sec:liftcoalg},
it follows that the base dual adjunction
lifts to one between $F$-coalgebras and $G$-algebras.
\[\xymatrix@R=4em{
   \cat{CHA}^\op \ar@/^1pc/[r]^-{\ol{\Cont}} \ar@{}|{\top}[r] \ar[d] & 
   \quad \cat{CAO} \ar@/^1pc/[l]^-{\ol{\Spec}} \ar[d] \\
   \KHaus^\op \ar@/^1pc/[r]^-{\Cont} \ar@{}|{\top}[r] & 
   \CCStar \ar@/^1pc/[l]^-{\Spec}
}
\]

The abstract algorithm \textbf{Algo1} applies
since $\KHaus$ and $\CCStar$ are monadic over $\Set$
(cf.~Section~\ref{sec:abstract-brz})
In particular, $\KHaus$ is the Eilenberg-Moore category of
the ultrafilter monad \cite{Manes69}.
In order to show that the associated trace logic is expressive
we need an extra argument, since the functor $F$
defined in \eqref{eq:gelfand-fctrs} does not have the shape required
by Lemma~\ref{lem:comp-adj} and Theorem~\ref{thm:min-algo}.
However, we can apply Remark~\ref{rem:subF} after observing
the following.
Let $F' := \Spec(M(\Omega)) \times (-)^\Alph$.
Then 
the associated natural isomorphism
${\xitrc}' \colon F' \Spec M \To \Spec M G$ 
specifies semantics of trace logic over $F'$-coalgebras.
To obtain a suitable $\tau \colon F \To F'$ note that
quotienting with $J$ in $\CCStar$ yields an epi
$e: M(O) \surj M(O)/J$
from which we get a mono in $\KHaus$
$\Spec(e): \Spec(M(O)/J) \inj \Spec(M(O))$.
Pre-composing $\Spec(e)$ with the isomorphism
$h: SubD(O) \stackrel{\sim}{\to} Spec(M(O)/J)$ given
by Lemma~\ref{lem:topo-aut-iso}
and defining
$\tau := (Spec(e) \circ h) \times id$,
it follows that 
$\tau \colon F \To F'$ has all components mono in $\KHaus$. 
It now follows that trace logic is also expressive for
$F$-coalgebras, i.e., for compact Hausdorff automata.

\begin{remark}
In order to view Gelfand duality \eqref{eq:gelfand-dua}
as a concrete dualty obtained from a dualising object,
we need to expand the setting a bit,
since $\bbR$ is not a compact Hausdorff space. 
This can be done by considering the dual adjunction
between locally compact Hausdorff spaces and
not-necessarily unital commutative $C^*$-algebras.
Gelfand duality is a restriction of this dual adjunction.
\end{remark}


\section{Alternating Automata}
\label{sec:alternating}




\emph{Alternating finite automata} (aka \emph{Boolean automata} or \emph{parallel automata}) were first studied in \cite{K76a,CS76,CKS81a,Leiss81,BrzozowskiLeiss80} as a finite-state analog of alternating Turing machines. Let $\Sigma$ be a fixed finite
\emph{input alphabet}. An \emph{alternating finite automaton} (AFA) over $\Sigma$ is a tuple $\A = (X,\iota,\delta,F)$, where
\begin{itemize}
\item
$X$ is a finite set of \emph{states},
\item
$F\subs X$ are the \emph{final states},
\item
$\delta : \Sigma \to X \to 2^X \to 2$ is the \emph{transition function}, and
\item
$\iota : 2^X \to 2$ is the \emph{acceptance condition}.
\end{itemize}
Intuitively, the machine $\A$ operates as follows. Let $k=\len X$.
Initially $k$ processes are started, each assigned to a different state, reading the first symbol
of the input word $w\in\Sigma^*$.
In each step, a process at state $s$ reads the next input symbol $a$ and
spawns $k$ child processes, each of which moves to a different state and continues in the same fashion,
while the parent process at $s$ waits for the child processes to report back a Boolean value. In this way
a $k$-branching computation tree is generated. When the end of the input word is reached,
a process at state $s$ reports 1 back to its parent if $s\in F$, 0 otherwise.
A non-leaf process waiting at state $s$, having read input symbol $a$, collects the $k$-tuple
$b\in 2^X$ of Boolean values reported by its children, computes $\delta asb$, and
reports that Boolean value back to its parent. When the initial processes have all received
values, say $c\in 2^X$, the machine accepts if $\iota c=1$, otherwise it rejects.

Alternating automata accept all and only regular sets. It was shown in \cite{K06a} by 
combinatorial means that a language $L\subs\Sigma^*$ is accepted by a $k$-state AFA iff its reverse
$\set{w^R}{w\in L}$ is accepted by a $2^k$ state deterministic finite automaton (DFA).

Our purpose in this section is to recast this result in the framework of our general duality principle.
The duality involves \emph{complete atomic Boolean algebras} (\CABA) and \emph{discrete spaces} (\Set), which underlie \emph{powerset Boolean algebras}.


\subsection{\CABA, \EMN, and \Setop}
\label{sec:CABA-EMN}

\subsubsection{\CABA}

A \emph{complete Boolean algebra} (\CBA) is a structure $(B,\neg,\bigvee,\bigwedge,0,1,{\le})$, where $B$ is a set, $\neg$ is a unary operation on $B$, $\bigvee$ and $\bigwedge$ are infinitary operations on the powerset of $B$, $0$ and $1$ are constants, and $\le$ is a partial order on $B$, such that
\begin{itemize}
\item
$(B,\neg,\vee,\wedge,0,1,{\le})$ is an ordinary Boolean algebra (\BA), where $\vee$ and $\wedge$ are the restrictions of $\bigvee$ and $\bigwedge$, respectively, to two-element sets; and
\item
$\bigvee A$ and $\bigwedge A$ give the supremum and infimum of $A$, respectively, with respect to $\le$.
\end{itemize}
The morphisms of \CBA are \BA homomorphisms that preserve $\bigvee$ and $\bigwedge$.

An \emph{atom} of a \BA is a $\le$-minimal nonzero element. A \BA is \emph{atomic} if every nonzero element has an atom $\le$-below it. A \emph{complete atomic Boolean algebra} (\CABA) is an atomic \CBA. The morphisms of \CABA are just the morphisms of \CBA.


It is known that every \CABA is isomorphic to the powerset Boolean algebra on its atoms,
thus every element is the supremum of the atoms below it.
{\CBA}s and {\CABA}s satisfy infinitary de Morgan and distributive laws:
\begin{align*}
\neg\bigvee a &= \bigwedge\set{\neg x}{x\in a}
& (\bigvee a) \wedge x &= \bigvee \set{y\wedge x}{y\in a}\\
\neg\bigwedge a &= \bigvee\set{\neg x}{x\in a}
& (\bigwedge a) \vee x &= \bigwedge \set{y\vee b}{y\in a}
\end{align*}
as well as other useful infinitary properties such as commutativity, associativity, and idempotence of $\bigvee$ and $\bigwedge$. The free \CABA on generators $X$ is the powerset \CABA $(\TTwo X,\bigcup,\bigcap,\setcompl,\emptyset,2^X)$.
See \cite{Halmos74,GivantHalmos09,MonkBonnet89,Sikorski66} for further information on the theory of {\CBA}s and {\CABA}s.

%

\rem{
\subsubsection{$\cPo$ and $\Qop$}
\label{sec:Q-and-Qop}

The \emph{contravariant powerset functor} $\cPo :\Set\to\Setop$ maps objects $X$ to their powerset $2^X$ and morphisms $f:X\to Y$ to
\begin{align}
& \cPo f = f^{-1}:2^Y\to 2^X & \cPo f(b) = f^{-1}(b) &= \set{x\in X}{f(x)\in b}.\label{eq:Qdef}
\end{align}
The functor $\cPo $ is dually self-adjoint with $\cPo \dashv\Qop$.
The natural bijection between $\Set(X,2^Y)$ and $\Setop(2^X,Y)$ in both directions is given by the \emph{exponential transpose} $f\mapsto \transp f = \lam{yx}{fxy}$, an involution that simply switches the order of the arguments.
\begin{align}
\begin{array}c
\begin{tikzpicture}[->, >=stealth', node distance=28mm, auto]
\small
  \node (NW) {$2^X$};
  \node (NE) [right of=NW] {$Y$};
  \node (SW) [below of=NW, node distance=12mm] {$X$};
  \node (SE) [below of=NE, node distance=12mm] {$2^Y$};
  \path (NW) edge[<-] node {$\transp f=\lam{yx}{fxy}$} (NE);
  \path (NE) edge node {$\Qop$} (SE);
  \path (SW) edge node {$\cPo $} (NW)
             edge node[swap] {$f = \trtr f$} (SE);
  \node (NL) [left of=NW, node distance=8mm, anchor=east] {$\Setop$:};
  \node (SL) [left of=SW, node distance=8mm, anchor=east] {$\Set$:};
\end{tikzpicture}
\end{array}
\label{eq:adjoint}
\end{align}
In terms of sets, $\transp f = \lambda y.\set{x\in X}{y\in f(x)}$.
The unit $\eta_X : X \to \TTwo X$ and counit $\eps_X = \eta_X^\op$ are opposites but represent the same set function $\eta_X(x)=\set{a}{x\in a}$.

}

\subsubsection{\EMN}

The self-dual adjunction $\cPo^\op\dashv\cPo$ of the contravariant powerset functor (Example~\ref{exm:self-dual-powerset}) gives rise to a $\Set$-monad $N = \cPo \o \cPo^\op$, where for $X$ a set and $f:X\to Y$ a set function,
\begin{align*}
N X &= \cPo \cPo^\op X = 2^{2^X} & Nf &= (f^{-1})^{-1}:\TTwo X\to\TTwo Y
\end{align*}
The unit and multiplication are
\begin{align*}
\eta_X(x) = \set a{x\in a},\quad\quad 
\mu_X(H) = \set a{\eta_{\cPo{X}}(a)\in H} = \eta_{\cPo X}^{-1}(H).
\end{align*}
This is called the \emph{double powerset} or \emph{neighborhood monad}. The category of Eilenberg-Moore algebras of $N$ is denoted \EMN.

\subsubsection{Equivalence of \CABA, \Setop, and \EMN}
\label{sec:Setop=to-EM}

It is known that the Eilenberg-Moore algebras of the double powerset monad $N$ are exactly the {\CABA}s. These two categories are also dually equivalent to \Set, that is, equivalent to \Setop, as observed in \cite{Taylor02}.

The equivalence of the three categories can be shown via the composition of three faithful functors that are injective on objects:
\begin{align}
\begin{array}c 
\begin{tikzpicture}[->, >=stealth', node distance=20mm, auto]
\small
  \node (A) {\Setop};
  \node (B) [right of=A] {\EMN};
  \node (C) [right of=B] {\CABA};
  \node (D) [right of=C] {\Setop.};
  \path (A) edge node {$J$} (B);
  \path (B) edge node {$\EMNtoCABA$} (C);
  \path (C) edge node {$\At$} (D);
\end{tikzpicture}
\end{array}
\label{eq:threefunctors}
\end{align}
Here $J$ is the Eilenberg-Moore comparison functor \cite{AdaHerStr90:ACC,MacLane71}. Concretely, $J$ sends a set $X$ to the \CABA $2^X$ and a function $f: X \to Y$ to its inverse image map. That is, $J = \Set(-,2)$ with Boolean structure.
The functor $\At$ takes a \CABA to its set of atoms and a \CABA morphism $f:A\to B$ to $\At f:\At B\to\At A$, where $\At\!f\,(b)$ is the unique atom $a$ of $A$ such that $\up a = f^{-1}(\up b)$ and $\up a$ and $\up b$ are the principal ultrafilters on atoms $a$ and $b$, respectively.
In a \CABA, there is a bijection between principal ultrafilters and atoms, and we have that $\At \cong \CABA(-,2)$. In other words, the equivalence given by $J$ and $\At$ is a concrete duality with dualising object $2$.

Although the equivalence between \EMN and \CABA is fairly well known, the details are rarely provided. We therefore describe the functor $\EMNtoCABA$ that produces a \CABA from an \EMN-algebra $(X,\alpha)$. Let $TX$ be the term monad for \CABA terms over indeterminates $X$.\footnote{$TX$ consists of \CBA terms with the arity of the infinitary operations bounded by $2^{2^{\len X}}$. There can be no such bound for \CBA in general, as there are {\CBA}s of arbitrarily large cardinality generated by $X$; thus there is no term monad for \CBA. However, {\CABA}s generated by $X$ are of cardinality at most double exponential in $\len X$, and we can bound arities accordingly.} Let $\EMNtoCABA(X,\alpha) = (X,\EMNtoCABA\alpha)$, where
\begin{align}
& \EMNtoCABA\alpha:TX\to X & \EMNtoCABA\alpha &= \alpha\circ(\tau N\circ T\eta)_X,\label{def:alphap}
\end{align}
where $T\eta_X:TX\to TNX$ substitutes $\eta_X(x)$ for $x\in X$ in a term and $\tau_{NX}:TNX\to NX$ is the evaluation map of the powerset \CABA $(\TTwo X,\bigcup,\bigcap,\setcompl,\emptyset,2^X)$.
In particular (and in more conventional notation), this gives the following definitions of the Boolean operations:
\begin{gather}
\begin{array}c
\bigvee_n x_n = \alpha(\bigcup_n \eta_X(x_n)) \qquad\quad
\bigwedge_n x_n = \alpha(\bigcap_n \eta_X(x_n))\\[8pt]
\neg x = \alpha(\setcompl{\eta_X(x)}) \qquad\quad 0 = \alpha(\emptyset) \qquad\quad 1 = \alpha(2^X).
\end{array}
\label{eq:BAdefs}
\end{gather}
The action of $\EMNtoCABA$ on morphisms is the identity.

The natural transformation $\tau N\circ T\eta:T\to N$ in \eqref{def:alphap} relating \CABA terms and double powerset is invertible up to \CABA equivalence. Consider the natural transformation
\begin{align*}
& \upsilon : N\to T & \upsilon_X(A) &= \bigvee_{a\in A} (\bigwedge_{x\in a} x \wedge \bigwedge_{x\not\in a} \neg x),\ \ A\in \TTwo X.
\end{align*}
It can be shown that
\begin{align*}
\tau N\circ T\eta\circ\upsilon &= \id_{N} & \upsilon\circ\tau N\circ T\eta &\equiv \id_{T}.
\end{align*}
By the latter we mean that for any term $t\in TX$,
$\upsilon_X(\tau_{NX}(T\eta_X(t))) \equiv t$ modulo the axioms of \CABA.
This essentially says that there is a disjunctive normal form for \CABA terms.

\subsection{Language acceptance of alternating automata}

\rem{ 
Given an $\EMN$-algebra $(A,\alpha)$ and a set function $f:X\to A$,
we denote the free homomorphic extension of $f$ by $\fext{f} = \eps_A \circ N f : (N X,\mu_X) \to (A,\alpha)$.
For a function $f: X \to Z^Y$, we denote the exponential transpose (in both directions) of $f$ by $\transp{f}: Y \to Z^X$.
Thus $f=\trtr f$.
}
Let $\A = (X,\delta,f,\iota)$ be an AFA with states $X$ and components
\begin{align*}
& \iota:1\to\TTwo X && \delta_a:X\to\TTwo X,\ a\in\Sigma && f:X\to 2
\end{align*}
where $\iota$ is the (transposed) acceptance condition, $\delta_a$ are the transitions, and $f : X \to 2$ is the characteristic function for the subset $F$ of accepting states.

The language accepted by $\A$ is $\Lang(\A) \eqdef \set{w\in\Sigma^*}{\iota(\delta'_w(F))=1}$, where
\begin{align*}
& \delta'_w : 2^X \to 2^X
& \delta'_\eps(A) &= A
& \delta'_{aw}(A)(s) &= \delta_a(s)(\delta'_w(A)).
\end{align*}

As constructed in \cite{K06a}, the associated DFA for the reverse language is $\A'$ with states $2^X$ and components
\begin{align*}
& \transp f:1\to 2^X && \transp\delta_a:2^X\to 2^X,\ a\in\Sigma && \transp\iota:2^X\to 2.
\end{align*}
This is a deterministic automaton, that is, a coalgebra for the functor $F = 2\times(-)^\Sigma$ with start state $\transp f$, transitions $\transp\delta_a$, and accept states $\transp\iota$. The language accepted by $\A'$ is $\Lang(\A') \eqdef \set{w\in\Sigma^*}{\transp\iota(\transp\delta_w(\transp f)) = 1}$, where
\begin{align*}
\transp\delta_\eps &= \id_{2^X} & \transp\delta_{wa} &= \transp\delta_a\circ\transp\delta_w.
\end{align*}
The combinatorial construction of \cite{K06a} amounts to recurrying the components of the automata. Denoting the reverse of a string $w$ by $w^R$ and using the fact that $\transp\delta_a = \delta'_a$, it can be shown inductively that $\transp\delta_w = \delta'_{w^R}$, therefore the language accepted by $\A'$ is the reverse of the language accepted by $\A$:
\begin{align*}
\Lang(\A') &= \set{w}{\transp\iota(\transp\delta_w(\transp{f})) = 1}
= \set{w}{\iota(\delta'_{w^R}(f)) = 1}
= \set{w}{w^R\in \Lang(\A)}.
\end{align*}

\subsection{Alternating automata as $\EMN$-automata}

We now show how the relationship between $\A$ and $\A'$ comes about as an instance of a dual adjunction of automata as described in Section~\ref{sec:dual-adj-lift}, in particular Section~\ref{sec:liftaut}.
We use the base equivalence between $\EMN$ and $\Setop$ described in \Sec~\ref{sec:CABA-EMN}.
For the sake of uniformity with the general setup in Section~\ref{sec:dual-adj-lift}, we take
$R$ as the right adjoint (hence we put the $\op$ on $\EMN$), and consider $R$ and $J$ as contravariant functors.
\begin{equation}
\label{eq:dual-aut-AA}
\begin{array}c
\xymatrix@C=3em@R=1.5cm{
 \left(\Aut_{\EMN}^{N(1),2}\right)^\op \ar@{->}@/^1pc/[r]^-{\bar{R}} \ar@{}|{\top}[r] \ar[d]
 & 
 \Aut_{\Set}^{1,2} \ar@{->}@/^1pc/[l]^-{\bar{J}} \ar[d]
\\
 \EMN^\op \ar@{->}@/^1pc/[r]^-{R} \ar@{}|{\cong}[r] 
 & 
 \Set \ar@{->}@/^1pc/[l]^-{J} 
}
\end{array}
\end{equation}

More precisely,
we show that $\A' = \bar{R}(\detA)$, where $\detA$ is the deterministic automaton over $\EMN$ obtained by applying the determinisation construction from Section~\ref{sec:det} for $N$ to $\A$.
The functor $R$ is the composition $R = \At\circ \EMNtoCABA$ (see \eqref{eq:threefunctors}). 

Recall from  Section~\ref{sec:det} that determinisation for $N$ takes free extensions of  the transition function and output function. 
That is, given an alternating automaton
$\A$ with states $X$ and components
\begin{align*}
& \iota : 1 \to \TTwo{X} && \delta_a: X \to \TTwo{X},\ a \in \Sigma && f : X\to 2
\end{align*}
over $\Set$,
we have a deterministic automaton $\detA$ with
\begin{align*}
& \fext{\iota} : \TTwo 1 \to \TTwo{X} && \fext\delta_a: \TTwo{X} \to \TTwo{X},\ a \in \Sigma && \fext f : \TTwo{X} \to 2
\end{align*}
over $\EMN$, using the \CABA structure on 2.
In $\detA$, we leave algebraic structure on $\TTwo{X}$ and $2$ implicit. Formally, they are the powerset {\CABA}s on $\TTwo X$ and $2$, respectively; these are isomorphic to the free \EMN-algebras $(NX,\mu_X)$ and $(N\emptyset,\mu_\emptyset)$ on generators $X$ and $\emptyset$, respectively.

We easily see that $\tup{\TTwo X,\fext{f},\fext{\delta},\fext{\iota}}$ instantiates the definition
from Section~\ref{sec:aut-alg-coalg} of an $\EMN$-automaton with initialisation in $\TTwo 1$
and output in 2, i.e., $\detA$ is in $\Aut_{\EMN}^{N(1),2}$.
For ease of notation, we will sometimes write the initialisation morphism $\fext{\iota}$
as its corresponding $\Set$-function $\iota$.


A dual automaton in $\Aut_\Set^{1,2}$ (with states $X$) is a coalgebra for $F = 2 \times (-)^\Alph$
together with an initial state $j: 1 \to X$, or equivalently an algebra for $G = 1 +\Alph\times(-)$
with output $f: X \to 2$. 
It is easy to check that the conditions for Theorem~\ref{thm:liftaut} hold.
First note that $I=\TTwo 1$ and $O=1$.
We then easily verify that $F_{\EMN} \cong J(1)\times(-)^\Alph$
by noting that $J(1) = 2^1 \cong 2$.
Similarly, to see that $G \cong R(\TTwo 1) +\Alph\cdot(-)$, 
we note that $R(\TTwo 1) = \At(\TTwo 1) = 2^1 \cong 2$.
Hence the base dual adjunction $J \dashv R$ lifts to $\bar{J} \dashv \bar{R}$ between automata categories, and the lifted adjoints are given by \eqref{eq:ol-LR} and Theorem~\ref{thm:liftaut}.
We describe the reversal functor $\bar{R}$ a bit more concretely as a contravariant functor
from $\Aut_{\EMN}^{N(1),2}$ to $\Aut_{\Set}^{1,2}$.
The base adjunction of \eqref{eq:dual-aut-AA} gives us a bijection of homsets:
\begin{align*}
\theta: \EMN((A,\alpha), JX) \to \Set(X, R(A,\alpha))
\end{align*}
natural in $(A,\alpha)$ and $X$.
Given an automaton in $\Aut_{\EMN}^{N(1),2}$
\begin{align*}
& \iota:1\to(A,\alpha) && \delta_a:(A,\alpha)\to(A,\alpha) && f:(A,\alpha)\to 2
\end{align*}
(again, we leave the algebraic structure on $2$ implicit), $\bar{R}$ produces the deterministic automaton over $\Set$ 
\begin{align}
& \theta f:1\to R(A,\alpha) && R(\delta_a):R(A,\alpha)\to R(A,\alpha) && R\fext\iota:R(A,\alpha)\to 2.\label{eq:detautom}
\end{align}
%
Applying $\bar{R}$ to $\detA$, which is
\begin{align*}
& \iota:1\to\TTwo X && \fext\delta_a:\TTwo X\to\TTwo X && \fext f:\TTwo X\to 2
\end{align*}
we get the reversed, deterministic automaton $\bar{R}(\detA)$ (over $\Set$):
\begin{align*}
& \theta\fext f:1\to 2^X && \bar R\fext\delta_a:2^X\to 2^X && R\fext\iota:2^X\to 2.
\end{align*}

\begin{theorem}
\label{thm:alternating}
For any alternating automaton
\begin{align*}
\A = (X, \{\delta_a: X \to \TTwo{X} \mid a \in \Sigma\},\iota : 1 \to \TTwo{X},f : X \to 2),
\end{align*}
$\A' \cong \bar{R}(\detA)$. 
\end{theorem}
\begin{proof}
The state space of $\A'$ is $2^X$ and the state space of $\bar{R}(\detA)$
is the set of atoms of the CABA $\EMNtoCABA(\TTwo{X},\mu_X)$ which is the set $\set{\{a\}}{a \subs X}$.
The correspondence between the initial and final states is shown in Lemmas \ref{lem:R-init} and \ref{lem:R-transp}.
The transition function 
$\delta$ is the tupling of maps $\delta_a: X \to \TTwo{X}$, i.e.,
$\delta = \angle{\delta_a}_{a \in \Sigma}$, and
similarly for $\fext{\delta}$ and $\transp{\delta}$.
The result follows by applying Lemma \ref{lem:R-transp} to each $\delta_a$
and retupling. 
\end{proof}

The relationship between an AFA and its determinised version can be understood as follows.
In an AFA, when reading an input word, we generate a computation tree downwards,
and once we reach the end of the word, we evaluate the outputs going back up using Boolean functions,
and at the top all outputs are aggregated into a single Boolean value with the acceptance condition.
In the determinised AFA, we propagate the acceptance condition forwards as a Boolean function (encapsulated in the state) and once we reach the end of the input word, we use the Boolean function to evaluate immediately instead of propagating back up.
The dual DFA of an AFA represents its logical semantics, or predicate transformer semantics,
where the observations at the end of the word are propagated backwards to the initial state.
Since predicate transformers move backwards, the language of an AFA is the reversed language of the dual DFA. 

Finally, we note that all conditions for Theorem~\ref{thm:min-algo} hold (with $\catD = \Set$ and $\FreeD=\ForgD = \Id$).
Hence we also get a Brzozowski style minimisation algorithm for alternating automata by instantiating \textbf{Algo2} of Section~\ref{sec:abstract-brz}.
Reachability in  $\Aut_\Set^{1,2}$  is just the standard automata-theoretic notion,
whereas now the more abstract algebraic notion from Section~\ref{sec:reach-obs}
is relevant ``on the left'' in the category $\Aut_{\EMN}^{N(1),2}$.
As with weighted automata (cf.~Section~\ref{sec:weighted}), we are not guaranteed that the result of the minimisation algorithm is again an alternating automaton (understood as an $FN$-coalgebra over $\Set$), since a subalgebra of a free \CABA need not be free.

\nocite{Street72}
\nocite{Huntington04}


%
%
%
%
%
%


\section{Conclusion and Related Work}
\label{sec:related-work}

In this paper we presented a unifying categorical perspective on the minimisation constructions presented in \cite{BKP:WoLLIC} and \cite{BBHPRS:ACM-TOCL}, revisited 
some examples from these two papers in light of the general framework,
and presented a new example of alternating automata.
We also filled in some details regarding topological automata (belief automata) that were missing from 
\cite{BKP:WoLLIC}.


Our starting points are Brzozowski's algorithm~\cite{Brz62} for the minimisation of deterministic automata and the use of Stone-type duality between computational processes and their logical characterisation~\cite{Abr91}. The connection between these two seemingly unrelated points is given by the duality principle between reachability and observability originally introduced in systems theory~\cite{Kalman59} and then extended to automata theory in \cite{AZ69,AM75,AM80:Hankel}. 

The duality between reachability and observability has been 
studied, e.g. in~\cite{BHK01} to relate coalgebraic and algebraic
specifications in terms of observations and constructors. In this context most notable is the use of Stone-type dualities between 
automata and varieties of formal languages~\cite{Geh09,GGP08,Rou11}
which recently culminated into a general algebraic and coalgebraic understanding of equations, coequations, Birkhoff's and Eilenberg–type correspondences~\cite{BCR15,SBBCR15,SBR16,AMUM15,Sal17,AMMU18}.

Our unifying categorical perspective is based on a dual adjunction between base categories lifted to a dual adjunction between coalgebras and algebras, as introduced in ~\cite{BK06,Kli07, KerKoeWes14} in the context of coalgebraic modal logic, and in~\cite{BKP:WoLLIC,KlinR16} to capture the observable behaviour of a coalgebra. Our novelty is to lift the coalgebra-algebra adjunction to a dual adjunction between automata which generalises the formalisation of Brzozowski's algorithm from~\cite{BBHPRS:ACM-TOCL},
and formalising the relationship of trace logic to the full modal logic and language semantics.

   Our paper focuses on comparing and unifying our earlier approaches from \cite{BKP:WoLLIC} and \cite{BBHPRS:ACM-TOCL} under a common umbrella, but we hasten to remark that
   the concept of minimisation via logic presented in section~\ref{sec:dual-adj-lift} is already in \cite{Rot16}.
   At its core, \cite{Rot16} uses a dual adjunction that is lifted to a dual adjunction between coalgebras and algebras.
   A logic is then used to provide a construction for obtaining observable coalgebras.
   This is esssentially what we call \textbf{Algo1}.
   The setting of \cite{Rot16} is more general as no assumptions are made on the specific shape of the algebra and coalgebra functors involved. Instead the necessary functor requirements are axiomatised.
   One achievement of~\cite{Rot16} is to generalise the setup in~\cite{BKP:WoLLIC} from
   dual equivalences to dual adjunctions. 
   The central contribution in~\cite{Rot16} 
   is to combine the duality-based framework with coalgebraic partition-refinement~\cite{AdamekBHKMS12} such that a logic-based treatment of Brzozowski and partition refinement is obtained. 
   Compared to~\cite{Rot16}, our framework is more restricted, as we
   confine ourselves to functors of certain shapes, but we believe this strikes a good balance between
   generality and a categorical setting for studying many different types of automata.
   Furthermore, our categorical framework incorporates a formalisation of the full Brzozowski algorithm
   via the small extension of the coalgebra-algebra adjunction to the adjunction of automata, i.e., 
   structures that have both initial and final states.   

Other categorical approaches to automata minimisation have been proposed in the literature;
we mention here just a few.
In \cite{CP17} languages and their acceptors are regarded as functors which provides a different perspective on minimisation in which Brzozowski can also be formulated.
In \cite{AdamekBHKMS12} the authors study coalgebras in categories equipped with factorisation structures in order to devise a generic partition refinement algorithm.
From the language-theoretic point of view, the relation between the automata constructions resulting from the automata-based congruences, together with the duality between right and left congruences, allows to relate determinisation and minimisation operations~\cite{GGV19}.





\begin{thebibliography}{10}

\bibitem{Abr91}
Samson Abramsky.
\newblock Domain theory in logical form.
\newblock {\em Annals of Pure and Applied Logic}, 51(1):1 -- 77, 1991.

\bibitem{AdamekBHKMS12}
Jir{\'{\i}} Ad{\'{a}}mek, Filippo Bonchi, Mathias H{\"{u}}lsbusch, Barbara
  K{\"{o}}nig, Stefan Milius, and Alexandra Silva.
\newblock A coalgebraic perspective on minimization and determinization.
\newblock In Lars Birkedal, editor, {\em Foundations of Software Science and
  Computational Structures - 15th International Conference, ({FOSSACS} 2012)},
  volume 7213 of {\em Lecture Notes in Computer Science}, pages 58--73.
  Springer, 2012.

\bibitem{AdaHerStr90:ACC}
Jir\'{\i} Ad{\'a}mek, Horst Herrlich, and George~E. Strecker.
\newblock {\em Abstract and Concrete Categories - The Joy of Cats}.
\newblock Dover Publications, 2009.

\bibitem{AMUM15}
Jiri Adamek, Robert S.~R. Myers, Henning Urbat, and Stefan Milius.
\newblock Varieties of languages in a category.
\newblock In {\em Proceedings of the 2015 30th Annual ACM/IEEE Symposium on
  Logic in Computer Science (LICS)}, LICS ’15, page 414–425. IEEE Computer
  Society, 2015.

\bibitem{AMMU18}
Ji\v{r}\'{\i} Ad\'{a}mek, Stefan Milius, Robert~S.R. Myers, and Henning Urbat.
\newblock Generalized eilenberg theorem: Varieties of languages in a category.
\newblock {\em ACM Transaction of Computational Logic}, 20(1), 2018.

\bibitem{Arbib74}
M.~A. Arbib and E.~G. Manes.
\newblock Machines in a category: An expository introduction.
\newblock {\em SIAM Review}, 16:163--192, 1974.

\bibitem{AM75:slatt}
M.~A. Arbib and E.~G. Manes.
\newblock Extensions of semilattices.
\newblock {\em The American Mathematical Monthly}, 82(7):744--746, 1975.

\bibitem{AM75:Fuzzy}
M.~A. Arbib and E.~G. Manes.
\newblock Fuzzy machines in a category.
\newblock {\em Bulletin of the Australian Mathematical Society},
  13(2):169--210, 1975.

\bibitem{AZ69}
M.A. Arbib and H.P. Zeiger.
\newblock On the relevance of abstract algebra to control theory.
\newblock {\em Automatica}, 5:589--606, 1969.

\bibitem{AM75}
Michael~A. Arbib and Ernest~G. Manes.
\newblock Adjoint machines, state-behavior machines, and duality.
\newblock {\em J. of Pure and Applied Algebra}, 6(3):313 -- 344, 1975.

\bibitem{AM80:Hankel}
Michael~A. Arbib and Ernest~G. Manes.
\newblock Foundations of system theory: The {Hankel} matrix.
\newblock {\em Journal of Computer and System Sciences}, 20:330--378, 1980.

\bibitem{AM80}
Michael~A. Arbib and Ernest~G. Manes.
\newblock Machines in a category.
\newblock {\em J. of Pure and Applied Algebra}, 19:9--20, 1980.

\bibitem{Arveson76}
William Arveson.
\newblock {\em An Invitation to $C^*$-Algebras}, volume~39 of {\em Graduate
  Texts in Mathematics}.
\newblock Springer-Verlag, 1976.

\bibitem{BCR15}
A.~Ballester-Bolinches, E.~Cosme-Ll\'{o}pez, and J.~Rutten.
\newblock The dual equivalence of equations and coequations for automata.
\newblock {\em Information and Computation}, 244(C):49–75, 2015.

\bibitem{Bartels:PhD}
F.~Bartels.
\newblock {\em On Generalised Coinduction and Probabilistic Specification
  Formats}.
\newblock PhD thesis, Vrije Universiteit Amsterdam, 2004.

\bibitem{BKP:WoLLIC}
Nick Bezhanishvili, Clemens Kupke, and Prakash Panangaden.
\newblock Minimization via duality.
\newblock In L.~Ong and R.~de~Queiroz, editors, {\em Proceedings of WoLLIC’
  12}, volume 7456 of {\em LNCS}, pages 191--205. Springer, 2012.

\bibitem{BHK01}
Michel Bidoit, Rolf Hennicker, and Alexander Kurz.
\newblock On the duality between observability and reachability.
\newblock In Furio Honsell and Marino Miculan, editors, {\em FoSSaCS}, volume
  2030 of {\em Lect. Notes in Comp. Sci.}, pages 72--87. Springer, 2001.

\bibitem{Blackadar06}
Bruce Blackadar.
\newblock {\em Operator algebras: theory of C*-algebras and von Neumann
  algebras}, volume 122 of {\em encyclopedia of Mathematical Sciences}.
\newblock Springer-Verlag, 2006.

\bibitem{BdRV01}
P.~Blackburn, M.~de~Rijke, and Y.~Venema.
\newblock {\em Modal logic}.
\newblock Cambridge University Press, Cambridge, 2001.

\bibitem{BBHPRS:ACM-TOCL}
Filippo Bonchi, Marcello Bonsangue, Helle~Hvid Hansen, Prakash Panangaden, Jan
  Rutten, and Alexandra Silva.
\newblock Algebra-coalgebra duality in brzozowski's minimization algorithm.
\newblock {\em ACM Transactions on Computational Logic,}, 15(1), 2014.

\bibitem{BBRS12}
Filippo Bonchi, Marcello~M. Bonsangue, Michele Boreale, Jan J. M.~M. Rutten,
  and Alexandra Silva.
\newblock A coalgebraic perspective on linear weighted automata.
\newblock {\em Information and Computation}, 211:77--105, 2012.

\bibitem{BonKur05}
M.~M. Bonsangue and A.~Kurz.
\newblock Duality for logics of transition systems.
\newblock In {\em FoSSaCS’05}, 2005.

\bibitem{BK06}
Marcello~M. Bonsangue and Alexander Kurz.
\newblock Presenting functors by operations and equations.
\newblock In Luca Aceto and Anna Ing{\'{o}}lfsd{\'{o}}ttir, editors, {\em
  Foundations of Software Science and Computation Structures, 9th International
  Conference, {FOSSACS} 2006, Held as Part of the Joint European Conferences on
  Theory and Practice of Software, ({ETAPS} 2006)}, volume 3921 of {\em Lecture
  Notes in Computer Science}, pages 172--186. Springer, 2006.

\bibitem{BorceuxII}
Francis Borceux.
\newblock {\em Handbook of Categorical Algebra 2: Categories and Structure}.
\newblock Cambridge University Press, 1994.

\bibitem{BrzozowskiLeiss80}
J.~A. Brzozowski and E.~Leiss.
\newblock On equations for regular languages, finite automata, and sequential
  networks.
\newblock {\em Theoretical Computer Science}, 10:19--35, 1980.

\bibitem{Brz62}
Janusz~A. Brzozowski.
\newblock Canonical regular expressions and minimal state graphs for definite
  events.
\newblock In {\em Mathematical Theory of Automata}, volume~12 of {\em {MRI}
  Symposia Series}, pages 529--561, Polytechnic Institute of Brooklyn, 1962.
  Polytechnic Press.

\bibitem{CZ97}
A.~Chagrov and M.~Zakharyaschev.
\newblock {\em Modal logic}, volume~35 of {\em Oxford Logic Guides}.
\newblock The Clarendon Press, New York, 1997.

\bibitem{CKS81a}
Ashok Chandra, Dexter Kozen, and Larry Stockmeyer.
\newblock Alternation.
\newblock {\em J. Assoc. Comput. Mach.}, 28(1):114--133, 1981.

\bibitem{CS76}
Ashok~K. Chandra and Larry~J. Stockmeyer.
\newblock Alternation.
\newblock In {\em Proc. 17th Symp. Found. Comput. Sci.}, pages 98--108. IEEE,
  October 1976.

\bibitem{CP17}
Thomas Colcombet and Daniela Petrisan.
\newblock Automata minimization: a functorial approach.
\newblock In Filippo Bonchi and Barbara K{\"{o}}nig, editors, {\em 7th
  Conference on Algebra and Coalgebra in Computer Science, ({CALCO} 2017)},
  volume~72 of {\em LIPIcs}, pages 8:1--8:16. Schloss Dagstuhl -
  Leibniz-Zentrum f{\"{u}}r Informatik, 2017.

\bibitem{ganty}
Pierre Ganty, Elena Guti{\'{e}}rrez, and Pedro Valero.
\newblock A congruence-based perspective on automata minimization algorithms.
\newblock In Peter Rossmanith, Pinar Heggernes, and Joost{-}Pieter Katoen,
  editors, {\em 44th International Symposium on Mathematical Foundations of
  Computer Science, ({MFCS} 2019)}, volume 138 of {\em LIPIcs}, pages
  77:1--77:14. Schloss Dagstuhl - Leibniz-Zentrum f{\"{u}}r Informatik, 2019.

\bibitem{GGV19}
Pierre Ganty, Elena Guti{\'{e}}rrez, and Pedro Valero.
\newblock A congruence-based perspective on automata minimization algorithms.
\newblock In Peter Rossmanith, Pinar Heggernes, and Joost{-}Pieter Katoen,
  editors, {\em 44th International Symposium on Mathematical Foundations of
  Computer Science, ({MFCS} 2019)}, volume 138 of {\em LIPIcs}, pages
  77:1--77:14. Schloss Dagstuhl - Leibniz-Zentrum f{\"{u}}r Informatik, 2019.

\bibitem{Geh09}
Mai Gehrke.
\newblock Stone duality and the recognisable languages over an algebra.
\newblock In Alexander Kurz, Marina Lenisa, and Andrzej Tarlecki, editors, {\em
  CALCO}, volume 5728 of {\em Lect. Notes in Comp. Sci.}, pages 236--250.
  Springer, 2009.

\bibitem{GGP08}
Mai Gehrke, Serge Grigorieff, and Jean-Eric Pin.
\newblock Duality and equational theory of regular languages.
\newblock In Luca Aceto, Ivan Damg{\aa}rd, Leslie~Ann Goldberg, Magn{\'u}s~M.
  Halld{\'o}rsson, Anna Ing{\'o}lfsd{\'o}ttir, and Igor Walukiewicz, editors,
  {\em ICALP (2)}, volume 5126 of {\em Lect. Notes in Comp. Sci.}, pages
  246--257. Springer, 2008.

\bibitem{GHKLMS03}
G.~Gierz, K.~H. Hofmann, K.~Keimel, J.~D. Lawson, M.~Mislove, and D.~S. Scott.
\newblock {\em Continuous lattices and domains}, volume~93 of {\em Encyclopedia
  of Mathematics and its Applications}.
\newblock Cambridge University Press, Cambridge, 2003.

\bibitem{Givant09}
Steven Givant and Paul Halmos.
\newblock {\em Introduction to Boolean Algebras}.
\newblock Undergraduate Texts in Mathematics. Springer-Verlag, 2009.

\bibitem{GivantHalmos09}
Steven Givant and Paul Halmos.
\newblock {\em Introduction to Boolean Algebras}.
\newblock Springer, 2009.

\bibitem{Halmos74}
Paul~R. Halmos.
\newblock {\em Lectures on Boolean Algebras}.
\newblock Springer, 1974.

\bibitem{HermidaJacobs98}
Claudio Hermida and Bart Jacobs.
\newblock Structural induction and coinduction in a fibrational setting.
\newblock {\em Information and Computation}, 145:107--152, 1998.

\bibitem{Huntington04}
Edward~V. Huntington.
\newblock Sets of independent postulates for the algebra of logic.
\newblock {\em Trans. Amer. Math. Soc.}, 5(3):288--309, July 1904.

\bibitem{Jacobs:bialg-dfa-regex}
B.~Jacobs.
\newblock A bialgebraic review of deterministic automata, regular expressions
  and languages.
\newblock In K.~Futatsugi, J.-P. Jouannaud, and J.~Meseguer, editors, {\em
  Algebra, Meaning and Computation: Essays dedicated to Joseph A. Goguen on the
  Occasion of his 65th Birthday}, volume 4060 of {\em LNCS}, pages 375--404.
  Springer, 2006.

\bibitem{JacSok10}
Bart Jacobs and Ana Sokolova.
\newblock Exemplaric expressivity of modal logics.
\newblock {\em J. Log. Comput.}, 20(5):1041--1068, 2010.

\bibitem{Joh82}
P.~T. Johnstone.
\newblock {\em Stone spaces}.
\newblock Cambridge University Press, Cambridge, 1982.

\bibitem{Johnstone:Adj-lif}
P.T. Johnstone.
\newblock Adjoint lifting theorems for categories of algebras.
\newblock {\em Bulletin London Mathematical Society}, 7:294--297, 1975.

\bibitem{Kalman59}
R.~Kalman.
\newblock On the general theory of control systems.
\newblock {\em IRE Transactions on Automatic Control}, 4(3):110--110, 1959.

\bibitem{KerKoeWes14}
Henning Kerstan, Barbara K{\"o}nig, and Bram Westerbaan.
\newblock Lifting adjunctions to coalgebras to (re)discover automata
  constructions.
\newblock In Marcello~M. Bonsangue, editor, {\em Coalgebraic Methods in
  Computer Science}, pages 168--188, Berlin, Heidelberg, 2014. Springer Berlin
  Heidelberg.

\bibitem{Kli07}
Bartek Klin.
\newblock Coalgebraic modal logic beyond sets.
\newblock In Marcelo Fiore, editor, {\em Proceedings of the 23rd Conference on
  the Mathematical Foundations of Programming Semantics, ({MFPS} 2007)}, volume
  173 of {\em Electronic Notes in Theoretical Computer Science}, pages
  177--201. Elsevier, 2007.

\bibitem{KlinR16}
Bartek Klin and Jurriaan Rot.
\newblock Coalgebraic trace semantics via forgetful logics.
\newblock {\em Logical Methods in Computer Science}, 12(4), 2016.

\bibitem{K76a}
Dexter Kozen.
\newblock On parallelism in {T}uring machines.
\newblock In {\em Proc. 17th Symp. Found. Comput. Sci.}, pages 89--97. IEEE,
  October 1976.

\bibitem{K06a}
Dexter Kozen.
\newblock {\em Theory of Computation}.
\newblock Springer, New York, 2006.

\bibitem{Kurz:PhD}
Alexander Kurz.
\newblock {\em Logics for Coalgebras and Applications to Computer Science}.
\newblock PhD thesis, Ludwigs-Maximilians-Universit\"at M\"nchen, 2000.

\bibitem{MacLane71}
Saunders~Mac Lane.
\newblock {\em Categories for the Working Mathematician}.
\newblock Springer-Verlag, New York, 1971.

\bibitem{Leiss81}
Ernst Leiss.
\newblock Succinct representation of regular languages by {B}oolean automata.
\newblock {\em Theoretical Computer Science}, 13:323--330, 1981.

\bibitem{Manes69}
E.~Manes.
\newblock A triple-theoretic construction of compact algebras.
\newblock In B.~Eckman, editor, {\em Seminar on Triples and Categorical
  Homology Theory}, number~80 in Lect. Notes Math., pages 91--118. Springer,
  1969.

\bibitem{MonkBonnet89}
J.~Donald Monk and eds. R.~Bonnet.
\newblock {\em Handbook of Boolean Algebras}.
\newblock North-Holland, 1989.

\bibitem{Negrepontis71}
Joan~W. Negrepontis.
\newblock Duality in analysis from the point of view of triples.
\newblock {\em Journal of Algebra}, 19:228--253, 1971.

\bibitem{PorTho91}
Hans-E. Porst and Walter Tholen.
\newblock Concrete dualities.
\newblock In H.~Herrlich and Hans-E. Porst, editors, {\em Category Theory at
  Work}. Heldermann Verlag, 1991.

\bibitem{Rot16}
Jurriaan Rot.
\newblock Coalgebraic minimization of automata by initiality and finality.
\newblock In Lars Birkedal, editor, {\em The Thirty-second Conference on the
  Mathematical Foundations of Programming Semantics, ({MFPS} 2016)}, volume 325
  of {\em Electronic Notes in Theoretical Computer Science}, pages 253--276.
  Elsevier, 2016.

\bibitem{Rou11}
F.~Roumen.
\newblock Canonical automata via duality.
\newblock Unpublished note., 2011.

\bibitem{Rutten00}
J.~J. M.~M. Rutten.
\newblock Universal coalgebra: A theory of systems.
\newblock {\em Theoretical Computer Science}, 249(1):3--80, 2000.

\bibitem{Sakai71}
Shoichiro Sakai.
\newblock {\em $C^*$-Algebras and $W^*$-algebras}.
\newblock Springer-Verlag, 1971.

\bibitem{Sal17}
Julian Salamanca.
\newblock Unveiling eilenberg-type correspondences: Birkhoff's theorem for
  (finite) algebras + duality.
\newblock {\em CoRR}, abs/1702.02822, 2017.

\bibitem{SBBCR15}
Julian Salamanca, Adolfo Ballester{-}Bolinches, Marcello~M. Bonsangue, Enric
  Cosme{-}Ll{\'{o}}pez, and Jan J. M.~M. Rutten.
\newblock Regular varieties of automata and coequations.
\newblock In Ralf Hinze and Janis Voigtl{\"{a}}nder, editors, {\em Mathematics
  of Program Construction - 12th International Conference, ({MPC} 2015),
  Proceedings}, volume 9129 of {\em Lecture Notes in Computer Science}, pages
  224--237. Springer, 2015.

\bibitem{SBR16}
Julian Salamanca, Marcello~M. Bonsangue, and Jurriaan Rot.
\newblock Duality of equations and coequations via contravariant adjunctions.
\newblock In Ichiro Hasuo, editor, {\em Coalgebraic Methods in Computer Science
  - 13th {IFIP} {WG} 1.3 International Workshop, ({CMCS} 2016)}, volume 9608 of
  {\em Lecture Notes in Computer Science}, pages 73--93. Springer, 2016.

\bibitem{Sch61}
Marcel~Paul Sch{\"u}tzenberger.
\newblock On the definition of a family of automata.
\newblock {\em Information and Control}, 4(2-3):245--270, 1961.

\bibitem{Sikorski66}
Roman Sikorski.
\newblock {\em Boolean Algebras}.
\newblock Springer, 1966.

\bibitem{SilvaBBR10}
A.~Silva, F.~Bonchi, M.M. Bonsangue, and J.J.M.M. Rutten.
\newblock Generalizing the powerset construction, coalgebraically.
\newblock In Kamal Lodaya and Meena Mahajan, editors, {\em IARCS Annual
  Conference on Foundations of Software Technology and Theoretical Computer
  Science, FSTTCS 2010, December 15-18, 2010, Chennai, India}, volume~8 of {\em
  LIPIcs}, pages 272--283. Schloss Dagstuhl - Leibniz-Zentrum f\"{u}r
  Informatik, 2010.

\bibitem{Sto36}
M.~H. Stone.
\newblock The theory of representations for {B}oolean algebras.
\newblock {\em Trans. Amer. Math. Soc.}, 40(1):37--111, 1936.

\bibitem{Sto37}
M.~H. Stone.
\newblock Topological representation of distributive lattices and {B}rouwerian
  logics.
\newblock {\em \v{C}asopis {P}e\v{s}t. {M}at. {F}ys.}, 67:1--25, 1937.

\bibitem{Street72}
Ross Street.
\newblock The formal theory of monads.
\newblock {\em J. Pure and Applied Algebra}, 2:149--168, 1972.

\bibitem{tan16}
Yi-Jia Tan.
\newblock Free sets and free subsemimodules in a semimodule.
\newblock {\em Linear Algebra and its Applications}, 496:527--548, 2016.

\bibitem{Taylor02}
Paul Taylor.
\newblock Subspaces in abstract {S}tone duality.
\newblock {\em Theory and Applications of Categories}, 10(13):300--366, 2002.

\end{thebibliography}

\section*{Appendix}
\renewcommand{\thesection}{A} 
\subsection{Example: PODFA-style minimisation vs Brzozowski}
\label{app:example}
  Classic DFAs are PODFAs with a single observation, hence they can be minimised using the
  duality approach in \cite{BKP:WoLLIC} using the duality between finite sets and finite Boolean algebras. 
$
\begin{tikzcd}[sep=1.5em, cramped]
    \bbP : \cat{FinSet}^\op
          \arrow[r, shift left=1.7pt]
      & \cat{FinBA} : \Uf
          \arrow[l, shift left=1.7pt]
\end{tikzcd}
$
or using Brozowski's algorithm via the self-dual adjunction, cf.~\cite{BBHPRS:ACM-TOCL}.
$
\begin{tikzcd}[sep=1.5em, cramped]
    \cPo : \Set^\op
          \arrow[r, shift left=1.7pt]
      & \Set : \cPo
          \arrow[l, shift left=1.7pt]
\end{tikzcd}
$
Consider the DFA $\str{X}$ 
below left from (11) in \cite{BBHPRS:ACM-TOCL}.
The DFA $\str{X}$ 
accepts the language $(a+b)^*a$.
The result $\str{X}'$ after the first reverse-determinise step in Brozowski's algorithm is shown on the right.
Disregarding initial and final states, $\str{X}'$ 
is also the modal algebra  
obtained from $\str{X}$ 
We then take reachable parts to get the automaton $\str{Y}$   
and the subalgebra $\str{A}$ 
of definable subsets of the modal language with a single proposition letter $p$ which is true precisely at accepting states $y,z$ of $\gamma$: $\sem{p} = \{yz\}$.
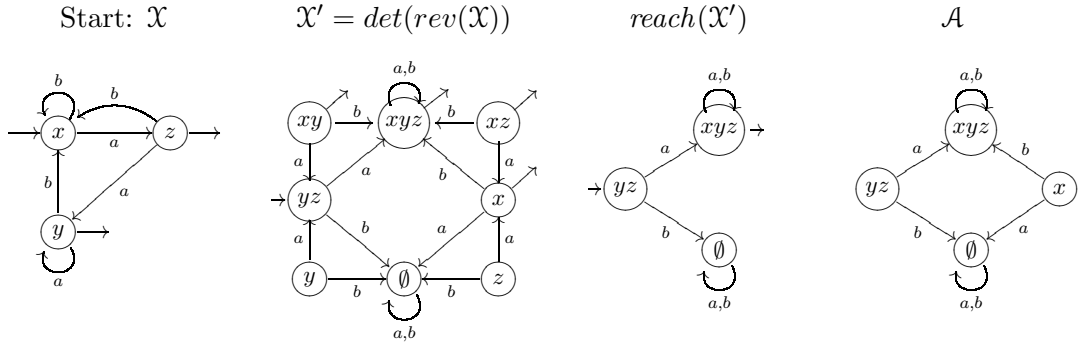
\[\begin{array}{cccccccc}
\text{Start: }  \str{X}  
& 
\str{X}' = det(rev(\str{X}))
& 
\reach{\str{X}'}
& 
\str{A}
\\[.6em]
\scalebox{.8}{\xymatrix@C=0.5cm@R=0.4cm{ %
&&&&
\\
  \ar[r] 
  & *++[o][F]{x} \ar@(ur,ul)_b \ar[rr]_a 
  & 
  &  *++[o][F]{z} \ar[r] \ar@/_1.2em/[ll]_b \ar[lldd]^a  & \\%
\\
  & *++[o][F]{y} \ar[r] \ar[uu]^b \ar@(dr,dl)^a  &
}} %
& 
\scalebox{.8}{\xymatrix@C=0.2cm@R=0.2cm{
&&&&&&
\\
 & *++[o][F]{xy} \ar[ur] \ar[dd]_a  \ar[rr]^b 
 && *++[o][F]{xyz} \ar@(ul,ur)^{a,b}   \ar[ur]
 &&  *++[o][F]{xz} \ar[ll]_b \ar[dd]^a \ar[ur]
 &
\\
&  && && &  \\
 \ar[r] 
 & *++[o][F]{yz} \ar[ddrr]^b \ar[rruu]_a 
 &&     
 &&  *++[o][F]{x} \ar[uull]^b \ar[lldd]_a \ar[ur] \\%
\\
             & *++[o][F]{y} \ar[rr]_b \ar[uu]^a  & & *++[o][F]{\emptyset} \ar@(dr,dl)^{a,b}  & & *++[o][F]{z} \ar[ll]^b \ar[uu]_a
}}
& 
\scalebox{.8}{\xymatrix@C=0.2cm@R=0.05cm{
\\\\
 &  && *++[o][F]{xyz} \ar@(ul,ur)^{a,b} \ar[r]  & &  \\
&  && &&   \\
      \ar[r] & *++[o][F]{yz} \ar[ddrr]_b \ar[rruu]^a & &      \\%
\\
             &   & & *++[o][F]{\emptyset} \ar@(dr,dl)^{a,b}  & &
}}
& 
\scalebox{.8}{\xymatrix@C=0.2cm@R=0.05cm{
\\\\
 &  && *++[o][F]{xyz} \ar@(ul,ur)^{a,b}   & &  \\
&  && &&   \\
       & *++[o][F]{yz} \ar[ddrr]_b \ar[rruu]^a & &  && *++[o][F]{x} \ar[uull]_-{b} \ar[ddll]^-{a}    \\%
\\
             &   & & *++[o][F]{\emptyset} \ar@(dr,dl)^{a,b}  & &
}}
\end{array}\]
After doing again reverse-determinise-reachability on
to complete the Brzozowski algorithm,
we get the automaton below on the left.
Taking the dual automaton (of atoms/ultrafilters) of $\str{A}$ 
we get the coalgebra below on the right.
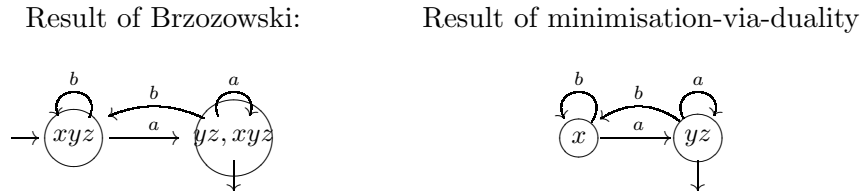
\[ \begin{array}{cccc}
\text{Result of Brzozowski:}
&&
\text{Result of minimisation-via-duality} 
\\[.6em]
\scalebox{.9}{\xymatrix@C=0.4cm@R=0.4cm{ %
\\
 \ar[r] 
  & *++[o][F]{xyz} \ar@(ur,ul)_b \ar[rr]^a 
  &&  *++[o][F]{yz,xyz} \ar@/_1.2em/[ll]_b \ar@(ul,ur)^a \ar[d] & 
\\&&&&
\\
}} 
& \quad &
\scalebox{.9}{\xymatrix@C=0.4cm@R=0.4cm{ %
\\
  & *++[o][F]{x} \ar@(ur,ul)_b \ar[rr]^a 
  &&  *++[o][F]{yz} \ar@/_1.2em/[ll]_b \ar@(ul,ur)^a \ar[d] & 
\\&&&&
\\
}} 
\\
\text{Accepts } \rev(\rev(L))= (a+b)^*a
&&
\text{State $x$ accepts } L = (a+b)^*a
\end{array}\]
The two automata (modulo initial state) are clearly isomorphic, but not identical. 



\subsection{Details of $\rho$}
\label{app:rho-sem}

For our specific choice of functors $F_\catC$ and $G_\catD$,
we compute the concrete definition of $\rho$ from \eqref{eq:mate}
when the adjunction arises from a dualising object.
\[\ba{lcccccl}
\rho_X\colon\\
O + \Alph\cdot\Pred{X} & \sgoes{\eta_{G\Pred{X}}}
& \Pred\Spc(O+\Alph\cdot\Pred{X}) & \sgoes{\sim}
& \Pred(\Spc{O}\times(\Spc\Pred{X})^\Alph) & \sgoes{\Pred F_\catC\coun_X}
& \Pred(\Spc{O} \times X^\Alph)
\\
O + \Alph\cdot\dua^{X} & \sgoes{\eta_{G(\dua^{X})}}
& \dua^{\dua^{(O+\Alph\cdot\dua^{X})}} & \sgoes{\sim}
& \dua^{(\dua^{O}\times(\dua^{\dua^{X}})^\Alph)} & \sgoes{\dua^{F_\catC\coun_X}}
& \dua^{\dua^{O} \times X^\Alph}
\\
q & \mapsto & \eta(q) & \mapsto &  \eta(q) & \mapsto &
\lambda (j,d) . j(q)
\\
\qquad (a,p) & \mapsto & \eta(a,p) & \mapsto &  \eta(a,p) & \mapsto &
\lambda (j,d) . p(d(a))
\ea\]
Here we used that the units evaluate, and that
\[\ba{lccc}
\dua^{F_\catC\coun_X} \colon
& \dua^{(\dua^{O}\times(\dua^{\dua^{X}})^\Alph)} & \sgoes{}
& \dua^{\dua^{O} \times X^\Alph}
\\
& h & \mapsto &
\lambda(j,d) . h (j, \lambda a. \lambda g. g(d(a)))
\ea\]
In short,
\begin{equation}\label{eq:rho-concrete}
\ba{rc}
\rho_X\colon\colon
& O + \Alph\cdot\dua^{X} 
\sgoes{}
\dua^{\dua^{O} \times X^\Alph}
\\
&
\ba{rcl}
\rho(q)(j,d) &=& j(q)\\
\rho(a,p)(j,d) &=& p(d(a))
\ea            
\ea
\end{equation}


\subsection{Theory maps of $(G,\rhotrc)$ and  $(G_\catD,\rho)$ coincide}
\label{app:theory-maps}

\textbf{Proof of Lemma~\ref{lem:trc-coincidence}}\\
Due to the natural isomorphisms of Hom-sets given by the adjunctions
in \eqref{eq:comp-adj},
we have the following correspondence for all $F_\catC$-coalgebras $\gamma$
and all $G$-algebras $\alpha$:
\begin{equation}\label{eq:G-adjoints}
\ba{ll}
\gamma \to \ol{\Spc}\,\ol{\FreeD}(\alpha) & \Coalg_\catD(F_\catC)\\
\hline\hline
\ol{\FreeD}(\alpha) \to \ol{\Pred}(\gamma) & \Alg_\catD(G_\catD)\\
\hline\hline
\alpha \to \ol{\ForgD}\,\ol{\Pred}(\gamma) & \Alg_\Set(G)^{\phantom{T^T}}
\ea
\end{equation}
In particular,
since $(\Trc,\alpha)$ is an initial $G$-algebra, it follows that
$\ol{\FreeD}(\Trc,\alpha)$ is an initial $G_\catD$-algebra,
since left adjoints preserve initial objects.

Furthermore, since contravariant adjoint functors turn colimit into limits,
$\ol{\Spc}\,\ol{\FreeD}(\Trc,\alpha)$ is a final $F_\catC$-coalgebra.

The semantic maps of the two logics are the unique morphisms from
the initial algebras, and we denote them
$\sm{G} \colon \Trc \to \ForgD\Pred(C)$ and
$\sm{G_\catD}\colon \FreeD(\Trc) \to \Pred(C)$.
The correspondence in \eqref{eq:G-adjoints} says that $\sm{G_\catD}$
is fully determined by $\sm{G}$.
When $\FreeD \dashv \ForgD$ is a free-forgetful adjunction (as in most of our examples), 
this tells us that formulas that are contained in both logics have the same semantics.

By definition,
$\th{G}$ is the ($\Spc\FreeD \dashv \ForgD\Pred$)-adjoint of $\sm{G}$
and
$\th{G_\catD}$ is the ($\Spc\dashv\Pred$)-adjoint of $\sm{G_\catD}$.
From \eqref{eq:G-adjoints} we see that $\th{G} = \th{G_\catD}$, and both are the unique map
into the final $F_\catC$-coalgebra.

An alternative argument using the mates $\xi$ and $\xitrc$ is as follows.

Consider the following diagram:
$$
\xymatrix{
      & & &  \Spc G_\catD \FreeD  \Trc   \\
      F_\catC C \ar@{-->}[r]^-{F_\catC  \th{G}} & F_\catC  \Spc \FreeD \Trc \ar[rr]^{\xitrc_{\Trc}} \ar[rru]^{\xi_{\FreeD \Trc}} 
       & & \Spc \FreeD G \Trc \ar[u]_{\Spc \kappa^{-1}} \\  
      C \ar[u]^\gamma \ar@{-->}[rrr]^{\exists ! \th{G}} & &  & \Spc \FreeD \Trc \ar[u]_{\Spc \FreeD \alpha}}
$$
where $\alpha$ denotes the initial $G$-algebra. By the universal property of the theory map (cf. e.g.~(2.5) in~\cite{KlinR16}) we have that  
$\th{G}$ is the unique map
that makes the above square commute. The upper triangle commutes by definition of $\xitrc$. As $\ol{\FreeD}$ maps
initial $G$-algebra to initial $G_\catD$-algebra we also have that $(\FreeD(\Trc),\FreeD \alpha \circ \kappa^{-1})$ is 
the initial $G_\catD$-algebra. Therefore the diagram shows that $\th{G}$ also satisfies the universal property of $\th{G_\catD}$ and thus
$\th{G} = \th{G_\catD}$ as claimed. 


\subsection{Coincidence of reachability notons}
\label{app:coincidence-reach}

\begin{lemma}\label{lem:wellp-initial}
Let $\cat{A}$ 
be a wellpowered category with initial object $0_{\cat{A}}$ and 
factorisation system $(E,M)$ such that 
$M \subseteq \mathit{Mono}(\cat{A})$.
For all objects $A \in \cat{A}$, the least subobject $\lsub{A}$ of $A$ is obtained
by $(E,M)$-factorisation of the unique morphism from $0_{\cat{A}}$ to $A$. 
\end{lemma}
\begin{proof}
Let $\xymatrix@1{0_{\cat{A}} \ar[r]^-{e} & \lsub{A} \ar[r]^-{m} & A}$ be the $(E,M)$-factorisation of the initial map from $0_{\cat{A}}$ to $A$. We show that $\lsub{A}$ is the least subobject of $A$. To this end, let $A' \stackrel{i}{\inj} A$ be a subobject and let 
$\xymatrix@1{0 \ar[r]^-{e'} & \lsub{A'} \ar[r]^-{m'} & A'}$ be the $(E,M)$-factorisation of the initial morphism for $A'$. By the diagonal fill-in property of $(E,M)$ we get a unique morphism $d\colon \lsub{A} \to \lsub{A'}$ such that $e' = d \circ e$ and $m = i \circ m' \circ d$ as shown here:
\[\xymatrix@C=1.2em@R=2em{
0 \ar[d]_-{e'} \ar[rr]^-{e} && \lsub{C} \ar@{>->}[d]^-{m} \ar@{-->}[dll]_-{d}
\\
\lsub{C'} \ar@{>->}[r]_-{m'} & C' \ar@{>->}[r]_-{i} & C
}\]
Take $h = m'\circ d \colon \lsub{A} \to A'$. To see that $h$ is unique, 
suppose $h' \colon \lsub{A} \to A'$ is such that $m = i \circ h'$, then
$i \circ h' = i \circ h$ and hence $h=h'$ since $i$ is a mono.
\end{proof}


\subsection{Factorisation systems for coalgebras and algebras}
\label{app:factor-sys}

\begin{lemma}\label{lem:factor-sys}
Assume that $\catC \cong \EM{T}$ for a $\Set$-monad $T$,
and $\catD \cong \EM{T}$ for a $\Set$-monad $T$.
We have:
\begin{enumerate}\setlength\itemsep{0pt}
\item 
$(E,M)$ with $E =$ surjective $F_\catC$-coalgebra morphisms
and $M = $ injective $F_\catC$-coalgebra morphisms is a factorisation system for $\Coalg_\catC(F_\catC)$.
\item 
$(E,M)$ with $E =$ surjective $G_\catD$-algebra morphisms
and $M = $ injective $G_\catD$-algebra morphisms is a factorisation system for $\Alg_\catC(G_\catD)$.
\item  
$(E,M)$ with $E =$ surjective $G$-algebra morphisms
and $M = $ injective $G$-algebra morphisms is a factorisation system for $\Alg_\Set(G)$.
\end{enumerate}
\end{lemma}

\rem{
  \begin{proof}
Item 1 follows from \cite[Thm.~3.1.7]{Kurz:PhD}.
Items 2 and 3 can be proved similarly to  \cite[Thm.~3.1.7]{Kurz:PhD} using some dual arguments.
More details are found in Appendix~\ref{app:factor-sys}.
\end{proof}
}
%
\begin{proof}
  Item 1 follows from \cite[Thm.~3.1.7]{Kurz:PhD}.
  We show that the conditions for  \cite[Thm.~3.1.7]{Kurz:PhD} hold.
  For all $\Set$-monads $T$,
$\EM{T}$ is a regular category (because it is exact) \cite[Thm~4.3.5]{BorceuxII},
and hence $(RegEpi,Mono)$ is a factorisation system.
Furthermore, in $\EM{T}$,
regular epis are the surjective homomorphisms (but epis need not be surjective) \cite[7.2]{AdaHerStr90:ACC}
and RegMono = Mono = injective homomorphisms \cite[6.9,6.12]{AdaHerStr90:ACC}.
It is easy to show that $F_\catC$ defined as above preserves monos.
Letting $V\colon \Coalg_\catC(F_\catC) \to \catC$ be the forgetful functor that
maps an $F_\catC$-coalgebra $(C,\gamma)$ to $C$, it then follows from
\cite[Thm.~3.1.7]{Kurz:PhD}, that $(V^{-1}RegEpi,V^{-1}Mono)$ is a factorisation system for
$\Coalg_\catC(F_\catC)$. Note that $V^{-1}RegEpi$ and  $V^{-1}Mono$ are the
surjective and injective $F_\catC$-coalgebra morphisms, respectively.
It is straightforward to prove that the functors $G_\catD$ and $G$ preserve
regular epis.

Items 2 and 3 can be proved similarly to  \cite[Thm.~3.1.7]{Kurz:PhD} using some dual arguments.
\rem{
\begin{lemma}\label{lem:lift-factsys-alg}
Let $H\colon \catA \to\catA$ be a functor on $\catA = \EM{T}$ for a $\Set$-monad $T$,
and let $V \colon \Alg_\catA(H) \to \catA$ be the forgetful functor that
maps an $H$-algebra $(A,\alpha)$ to $A$.
If $H$ preserves regular epis, 
then $(V^{-1}RegEpi,V^{-1}Mono)$ consisting of surjective and injective $H$-algebra morphisms
is a factorisation system for $\Alg_\catA(H)$.
\end{lemma}
}
Sketch: Using that $G_\catD$ preserves regular epis, one can show that
 $V \colon \Alg_\catD(G_\catD) \to \catD$ 
creates $(RegEpi,Mono)$-factorisations
in $\catD$ using the diagonal fill-in property (similar to \cite[Prop.~1.3.3]{Kurz:PhD}).
Using that regular epis in $\catD$ are surjective, one can show that the diagonal fill-in obtained from
$(RegEpi,Mono)$ in $\catD$ is an $G_D$-algebra morphism
(similar to \cite[Thm.~3.1.7]{Kurz:PhD} and \cite[Lem.~2.4]{Rutten00}).
It follows that $(V^{-1}RegEpi,V^{-1}Mono)$ is a factorisation system for $\Alg_\catD(G_D)$.
\end{proof}



\subsection{Isomorphism lemma for topological automata}
\label{app:topo-aut-iso}

\textbf{Proof of Lemma~\ref{lem:topo-aut-iso}}
  Let $\alpha\colon \SubD(Y) \to \Spec(M(Y)/J)$ be defined by
  $\alpha(\phi)(m) = m(\phi)$ where $\phi \in \SubD(Y)$ and $m \in M(Y)$.
  Note that this is well defined.  If $m \icong{J} m'$ then their
  difference lies in $J$ which means that
  \(m(\phi) - m'(\phi) =k(\phi)(|\pi^-|-\pi^-)(\phi)\).
  This condition is equivalent to \(1-\pi\)
  is positive and hence $|\pi^-|$ and $\pi^-$ are equal and hence the
  second term is $0$, whence \(m(\phi) = m'(\phi)\).
  Note that the topology of \(\Spec(M(Y)/J)\)
  is generated by the taking as the closed sets, sets of maximal ideals
  that contain a fixed element $\phi$ of $M(Y)$.  Any maximal ideal
  consists of the functions that vanish at a point $y$, call this $m_y$.
  So if we fix such an $\phi$ for it to be in a maximal ideal $m_y$, we
  have $\phi(y) =0$.  This means that $\alpha^{-1}$ of a closed set is the
  set of subdistributions that assign $0$ to a particular element $y$; this
  is a closed set so $\alpha$ is continuous.

Let $\beta\colon \Spec(M(Y)/J) \to \SubD(Y)$ be defined by
$\beta(\Phi)(y) = \Phi(\pi_y)$ where $\pi_y \in M(Y)/J$ projects onto $y$.
We check that $\beta(\Phi)$ is a subdistribution:
\[\sum_{y\in Y}\beta(\Phi)(y)
= \sum_{y\in Y}\Phi(\pi_y)
= \Phi\left(\sum_{y\in Y} \pi_y \right)
\pleq \Phi(1) = 1
\]
The second and last identity hold because $\Phi$ is a
$\CCStar$-homomorphism (hence linear and unital); 
the inequality holds since we are in $M(Y)/J$ (which says that $\pi\pleq
1$) and $\Phi$ is monotone.

We now show that for $\phi \in \SubD(Y)$,
$\beta(\alpha(\phi)) = \phi$. Let $y \in Y$, we then have:
$
\beta(\alpha(\phi))(y) = \alpha(\phi)(\pi_y) = \pi_y(\phi) = \phi(y)
$
We show that for $\Phi \in \Spec(M(Y)/J)$, $\alpha(\beta(\Phi)) = \Phi$.
Let $m\in M(Y)$, we then have
$
\alpha(\beta(\Phi))(m) = m(\beta(\Phi)).
$
By the Stone-Weierstrass theorem, the polynomials on the compact Hausdorff
space $[0,1]^Y$ are dense in $\Cont([0,1]^Y)$. 
Since $m \colon [0,1]^Y \to \bbR$ is continuous, 
it therefore suffices to show that for all polynomials $p$ on $[0,1]^Y$ we
have that $p(\beta(\Phi)) = \Phi(p)$.  

\emph{Case $p = 1$:} $p(\beta(\Phi)) = 1 = \Phi(1) = \Phi(p)$.
\emph{Case $p = r \in \bbR$:} $p(\beta(\Phi)) = r = \Phi(r) = \Phi(p)$.
\emph{Case $p = \sum_{y\in Y} r_y \pi_y$:}
$\Phi(p) = \Phi(\sum_{y\in Y} r_y \pi_y) = \sum_{y \in Y}r_y\Phi(\pi_y) = \sum_{y \in Y} r_y\beta(\Phi)(y) = p(\beta(\Phi))$. 
Finally, let $p = \sum_{y\in Y} r_y \pi_y$ and $q = \sum_{y\in Y} s_y \pi_y$.
Then
$
\Phi(pq) = \Phi(p)\Phi(q) = p(\beta(\Phi))q(\beta(\Phi)) = (pq)(\Phi).
$
It follows that $p(\beta(\Phi)) = \Phi(p)$ holds for all polynomials
$p$\footnote{Clearly all polynomials can be expressed as sums of products
  of lower degree polynomials.}
and we have now shown that $\alpha \colon \SubD(Y) \to \Spec(M(Y)/J) $ is a
bijection with inverse $\beta$.  


\subsection{Alternating automata}
\label{app:alternating}

These lemmas are needed for Theorem \ref{thm:alternating}.

\begin{lemma}
\label{lem:fext-transp}
Let $d: Y \to \TTwo{X}$. Then $\fext{d}={\transp d}^{-1}$.
\end{lemma}
\begin{proof}
For all $a \subs X$ and all $A \in \TTwo{Y}$,
\begin{align*}
a \in \fext{d}(A)
&\Iff a \in\mu_X(N d (A))
\Iff \eta_{QX}(a) \in N d (A)
\Iff \eta_{QX}(a) \in (d^{-1})^{-1} (A)\\
&\Iff d^{-1}(\eta_{QX}(a)) \in A
\Iff \set{x}{d(x)\in\eta_{QX}(a)} \in A
\Iff \set{x}{a\in d(x)} \in A\\
&\Iff \transp{d}(a) \in A.
\end{align*}
\end{proof}

The following lemmas show that the transition structure of $\A'$ coincides with the transition structure of $\bar{R}(\detA)$.

\begin{lemma}
\label{lem:R-init}
For all functions $f : X \to 2$, $\theta(\fext{f}) = \transp{f}$.
\end{lemma}
\begin{proof}
This is immediate from the fact that $\theta^{-1}(\transp f) = {\transp f}^{-1}$ and Lemma \ref{lem:fext-transp}.
\end{proof}

\begin{lemma}
\label{lem:R-transp}
For all $d : Y \to \TTwo{X}$ and all $a \subs X$,
$R(\fext{d})(\{a\}) = \{\transp{d}(a)\}$.
Thus $R(\fext d)=\transp d$ up to bijections
relating atoms $\{a\}$ and their singleton elements $a$.
In particular, for all functions $i : 1 \to \TTwo{X}$, $R(\fext{i}) = \transp{i}$ up to these bijections.
\end{lemma}
\begin{proof}
The atoms of $(\TTwo{X},\mu_X)$ are of the form $\{a\}$ for $a \subs X$, hence for $A\in\TTwo Y$,
\begin{align*}
R(\fext{d})(\{a\}) 
&= \bigwedge\set{A \in \TTwo{Y}}{\{a\}\le \fext{d}(A)}\\
&= \bigcap\set{A\in\TTwo{Y}}{a\in\fext{d}(A)} && \text{since $\bigwedge$ is $\bigcap$ in $(\TTwo{X},\mu_X)$}\\
&= \bigcap\set{A\in\TTwo{Y}}{\transp{d}(a)\in A} && \text{Lemma~\ref{lem:fext-transp}}\\
&= \{\transp d(a)\} && \text{since $\transp d(a)\in 2^Y$.}
\end{align*}
\end{proof}


\end{document}